\documentclass[12pt]{article}

\usepackage{amsmath}
\usepackage{amsthm}
\usepackage{amsfonts}
\usepackage{amssymb}
\usepackage{latexsym} 
\usepackage{epsfig}
\usepackage{graphicx}

\usepackage{calc}  
\usepackage[matrix,tips,graph,curve,web]{xy}

\makeatletter

\def\@maketitle{%
  \newpage
  \null
  \let \footnote \thanks
    {\normalfont\sffamily\bfseries\Large\noindent\@title \par}%
    \vskip 1em%
    {\normalfont\sffamily 
        \noindent
        \@author
        \par}
  \par
  \vskip 4em}
\def\@seccntformat#1{\csname the#1\endcsname{.\ }}
\renewcommand\section{\@startsection {section}{1}{\z@}%
                                   {-3.0ex \@plus -1ex \@minus -.2ex}%
                                   {1.5ex \@plus.2ex}%
                                   {\normalfont\large\bfseries}}
\renewcommand\subsection{\@startsection{subsection}{2}{\z@}%
                                     {-2.75ex\@plus -1ex \@minus -.2ex}%
                                     {1.5ex \@plus .2ex}%
                                   {\normalfont\large}}
\def\fnum@figure{\normalfont\footnotesize\figurename~\thefigure}
\renewcommand\tableofcontents{%
    \section*{\contentsname
        \@mkboth{%
           \MakeUppercase\contentsname}{\MakeUppercase\contentsname}}%
    \@starttoc{toc}%
    }
\renewcommand*\l@part[2]{%
  \ifnum \c@tocdepth >-2\relax
    \addpenalty\@secpenalty
    \addvspace{2.25em \@plus\p@}%
    \begingroup
      \setlength\@tempdima{3em}%
      \parindent \z@ \rightskip \@pnumwidth
      \parfillskip -\@pnumwidth
      {\leavevmode
       \large \bfseries #1\hfil \hb@xt@\@pnumwidth{\hss #2}}\par
       \nobreak
       \if@compatibility
         \global\@nobreaktrue
         \everypar{\global\@nobreakfalse\everypar{}}%
      \fi
    \endgroup
  \fi}
\renewcommand*\l@section[2]{%
  \ifnum \c@tocdepth >\z@
    \addpenalty\@secpenalty
    \addvspace{1.0em \@plus\p@}%
    \setlength\@tempdima{1.5em}%
    \begingroup
      \parindent \z@ \rightskip \@pnumwidth
      \parfillskip -\@pnumwidth
      \leavevmode \sffamily\bfseries
      \advance\leftskip\@tempdima
      \hskip -\leftskip
      #1\nobreak\hfil \nobreak\hb@xt@\@pnumwidth{\hss #2}\par
    \endgroup
  \fi}
\renewcommand*\l@subsection{\sffamily\@dottedtocline{2}{1.5em}{2.3em}}
\renewcommand*\l@subsubsection{\@dottedtocline{3}{3.8em}{3.2em}}
\renewcommand*\l@paragraph{\@dottedtocline{4}{7.0em}{4.1em}}
\renewcommand*\l@subparagraph{\@dottedtocline{5}{10em}{5em}}
\makeatother

\makeatletter 
\@addtoreset{equation}{section}
\makeatother

\renewcommand{\theequation}{\thesection.\arabic{equation}}

\theoremstyle{plain}
\newtheorem{theorem}[equation]{Theorem}
\newtheorem{corollary}[equation]{Corollary}
\newtheorem{lemma}[equation]{Lemma}
\newtheorem{proposition}[equation]{Proposition}
\newtheorem{conjecture}[equation]{Conjecture}

\theoremstyle{definition}
\newtheorem{definition}[equation]{Definition}

\newtheorem{example}[equation]{Example}

  {\begin{list}{}%
    {%
    \settowidth{\labelwidth}{#1}%
    \setlength{\itemindent}{0pt}%
    \setlength{\labelsep}{1em}%
    \setlength{\leftmargin}{\labelwidth+\parindent+\labelsep}%
    \setlength{\itemsep}{0pt}%
    \setlength{\parsep}{.6ex}}}%
  {\end{list}}

  {\begin{list}{}%
    {%
    \settowidth{\labelwidth}{#1}%
    \setlength{\itemindent}{0pt}%
    \setlength{\labelsep}{1em}%
    \setlength{\leftmargin}{\labelwidth+\labelsep}%
    \setlength{\itemsep}{0pt}%
    \setlength{\parsep}{.6ex}}}%
  {\end{list}}

\makeatletter
\newenvironment{subequations*}{
  \begingroup 
  \let\protect\@nx
  \edef\@tempa{\def\@nx\theparentequation{\theequation}}%
  \@xp\endgroup\@tempa
  \setcounter{parentequation}{\value{equation}}%
  \setcounter{equation}{0}%
  \def\theequation{\theparentequation\alph{equation}}%
  \ignorespaces
}{%
  \setcounter{equation}{\value{parentequation}}%
  \global\@ignoretrue
}

\newcommand{\sgn}{\rm sgn\,}

\renewcommand\det{{\rm det\,}}

\def\d/{/\mspace{-6.0mu}/}

\usepackage{helvet}
\usepackage{graphicx,amssymb,amsmath,amsfonts,amsthm,color,url, placeins}
\usepackage{wrapfig}
\usepackage{tabstackengine}

\renewcommand\section{\@startsection{section}{1}{\z@}%
                                   {-3.0ex \@plus -1ex \@minus -.2ex}%
                                   {1.5ex \@plus.2ex}%
                                   {\normalfont\sffamily\large\bfseries}}
\renewcommand\subsection{\@startsection{subsection}{2}{\z@}%
                                     {-2.75ex\@plus -1ex \@minus -.2ex}%
                                     {1.5ex \@plus .2ex}%
                                   {\normalfont\sffamily\large}}
\renewcommand\subsubsection{\@startsection{subsubsection}{3}{\z@}%
                                     {-2.75ex\@plus -1ex \@minus -.2ex}%
                                     {1.5ex \@plus .2ex}%
                                   {\normalfont\sffamily\large}}

\newcommand{\od}{\stackrel{\mbox {\tiny {def}}}{=}}

\def\RR{\mathbb{R}}

\def\d{\mathrm{d}}

\def\RR{\mathbb{R}}

\newtheorem{rules}{Rule}

\def\RR{\mathbb{R}}

\def\det{\operatorname{det}}

\def\max{\mathrm{max}}

\def\supp{\operatorname{supp}}

\def\od{\stackrel{\mathrm{def}}{=}}

\def\supp{\operatorname{supp}}

\def\sgn{\operatorname{sgn}}
\def\FP{\operatorname{FP}}

\def\idx{\operatorname{idx}}
\def\Wtil{\widetilde{W}}
\def\eig{\operatorname{eig}}
\def\one{\mathbf{1}}

\usepackage{color}
\definecolor{cherry}{rgb}{0.9,.1,.2}

\begin{document}

\noindent {\Large \bf Fixed points of competitive threshold-linear networks}\\
Carina Curto, Jesse Geneson, Katherine Morrison\\
\medskip 
August 3, 2018
\vspace{.25in}

\noindent {\bf Abstract.}
Threshold-linear networks (TLNs) are models of neural networks that consist of simple, perceptron-like neurons and exhibit nonlinear dynamics that are determined by the network's connectivity.  The fixed points of a TLN, including both stable and unstable equilibria, play a critical role in shaping its emergent dynamics. In this work, we provide two novel characterizations for the set of fixed points of a competitive TLN: the first is in terms of a simple sign condition, while the second relies on the concept of {\it domination}. We apply these results to a special family of TLNs, called combinatorial threshold-linear networks (CTLNs), whose connectivity matrices are defined from directed graphs. This leads us to prove a series of {\it graph rules} that enable one to determine fixed points of a CTLN by analyzing the underlying graph.
Additionally, we study larger networks composed of smaller ``building block'' subnetworks, and prove several theorems relating the fixed points of the full network to those of its components. Our results provide the foundation for a kind of ``graphical calculus'' to infer features of the dynamics from a network's connectivity. 

\tableofcontents

\section{Introduction}
Threshold-linear networks (TLNs) are commonly-used models of recurrent networks consisting of simple, perceptron-like neurons with nonlinear interactions. The dynamics are given by the system of ordinary differential equations:
\begin{equation}\label{eq:network2}
\dfrac{dx_i}{dt} = -x_i + \left[\sum_{j=1}^n W_{ij}x_j+b_i \right]_+, \quad i = 1,\ldots,n,
\end{equation}
where $n$ is the number of neurons, $x_i(t)$ is the activity level (or ``firing rate") of the $i$th neuron, $W_{ij}$ is the connection strength from neuron $j$ to neuron $i$, and $[\cdot]_+ \od \max\{\cdot, 0\}$ is the threshold nonlinearity.  The external inputs $b_i \in \RR$ may be heterogeneous, but are assumed to be constant in time. We refer to a given choice of TLN as $(W,b)$.  {\it Competitive} TLNs have the additional requirement that all interactions are effectively inhibitory, with matrix entries $W_{ij} \leq 0$.

TLNs have previously been studied through the lens of permitted and forbidden sets \cite{Seung-Nature, XieHahnSeung, HahnSeungSlotine, flex-memory, net-encoding, pattern-completion}, though this work was largely restricted to the case of symmetric networks.   TLNs that are both symmetric and competitive generically exhibit multistability, with activity that is guaranteed to converge to a stable fixed point irrespective of initial conditions \cite{HahnSeungSlotine}.  This property motivated the use of these networks as models for associative memory encoding and retrieval, similar to the Hopfield model \cite{Hopfield1}. Earlier results on the mathematical theory of TLNs thus focused primarily on stable fixed points of symmetric networks. The papers \cite{Seung-Nature, XieHahnSeung, HahnSeungSlotine} gave characterizations and applications of {\it permitted sets}, which are subsets of neurons that have the capacity to support a stable fixed point of~\eqref{eq:network2} for some (potentially unknown) external input vector $b$. These authors also found nice properties satisfied by the full collection of permitted sets of a symmetric TLN \cite{HahnSeungSlotine}. The theory of permitted sets was further extended and developed by different authors in \cite{flex-memory, net-encoding}. Finally, in \cite{pattern-completion}, attention was shifted to the study of permitted sets that can actually be realized as fixed point supports for a known, and uniform, external input. Like previous results, this work was largely restricted to {\it stable} fixed points of {\it symmetric} TLNs.

The non-symmetric case, however, is considerably more interesting: asymmetric TLNs exhibit the full repertoire of nonlinear dynamic behavior, including limit cycles, quasiperiodic attractors, and chaos \cite{CTLN-preprint}. Furthermore, recent work has highlighted the importance of 
{\it unstable} fixed points in shaping a network's dynamic attractors \cite{book-chapter, CTLN-paper}.
Nevertheless, a mathematical theory connecting unstable fixed points of~\eqref{eq:network2} to the structure of $(W,b)$ has been lacking.

In this paper, we study the set of all fixed point supports, denoted $\FP(W,b)$, for asymmetric, competitive $W$ and nonnegative $b$.
In particular, we provide two new characterizations of $\FP(W,b)$: the first in terms of a sign condition (Theorem~\ref{thm:sgn-condition}), and the second in terms of {\it domination} (Theorem~\ref{thm:domination}), which for simplicity is restricted to networks with uniform $b$.  We introduce the language of {\it permitted motifs} of $(W,b)$ to refer to subsets of neurons that support a fixed point in their restricted subnetwork (all other subsets are called {\it forbidden motifs}). We find that whether or not a permitted motif supports a fixed point in the full TLN depends critically on the embedding of the subnetwork inside the larger network.

In order to investigate how the fixed points are shaped by qualitative aspects of a network's connectivity structure, we focus our attention on applications and further development of this theory in the special case of \emph{combinatorial threshold-linear networks} (CTLNs), first introduced in \cite{CTLN-preprint}. CTLNs are a family of competitive TLNs with uniform inputs and connectivity matrices $W$ that are defined from simple\footnote{A graph is \emph{simple} if it does not have loops or multiple edges between a pair of nodes.} directed graphs (see Figure~\ref{fig:network-setup-and-3cycle}). Given a graph $G$, and continuous parameters $\varepsilon, \delta$ and $\theta$, the associated CTLN is the network $(W,\theta)$ where $W = W(G,\varepsilon,\delta)$ has entries:
\begin{equation} \label{eq:binary-synapse}
W_{ij} = \left\{\begin{array}{cc} 0 & \text{ if } i = j, \\ -1 + \varepsilon & \text{ if } j \rightarrow i \text{ in } G,\\ -1 -\delta & \text{ if } j \not\rightarrow i \text{ in } G. \end{array}\right. \quad \quad \quad \quad
\end{equation}
Note that $ j \rightarrow i$ indicates the presence of an edge from $j$ to $i$ in the graph $G$, while $j \not\rightarrow i$ indicates the absence of such an edge.  We additionally require that $\theta > 0$,\; $\delta >0,$ and $0 < \varepsilon < \frac{\delta}{\delta+1}$; when these conditions are met, we say that the parameters are within the \emph{legal range}. 

\begin{figure}[!ht]
\begin{center}
\includegraphics[height=1.75in]{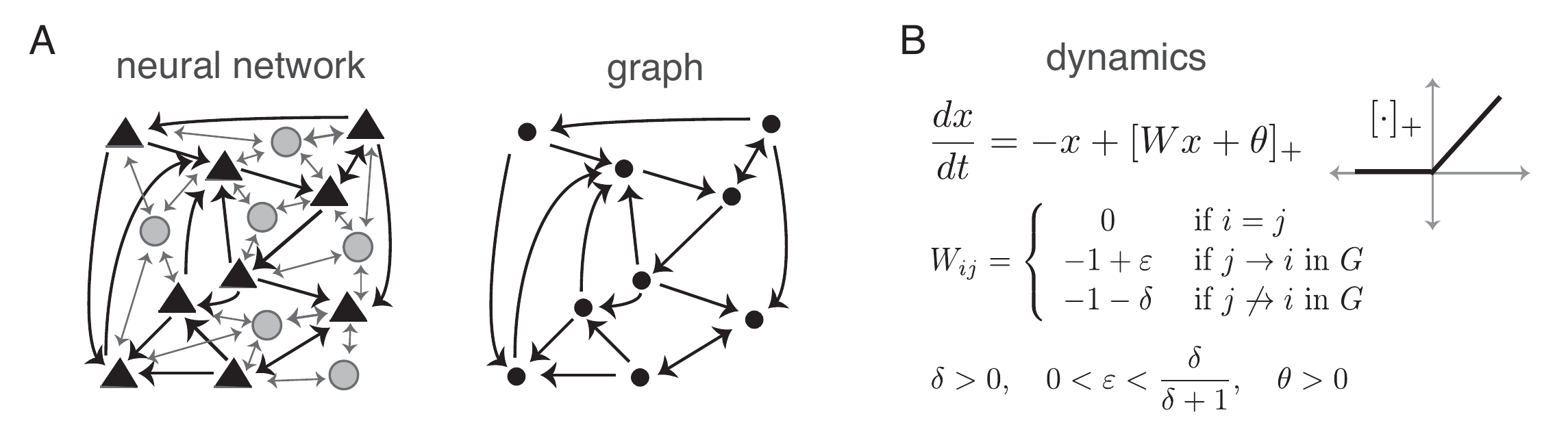}
\caption{(A) A neural network with excitatory pyramidal neurons (triangles) and a background network of inhibitory interneurons (gray circles) that produce a global inhibition. The corresponding graph (right) retains only the excitatory neurons and their connections.  (B) Equations for the CTLN model.}
\label{fig:network-setup-and-3cycle}
\end{center}
\vspace{-.1in}
\end{figure}

Despite the additional constraints, CTLNs display the full range of nonlinear dynamics that are observed in general TLNs \cite{CTLN-preprint, book-chapter}, and thus provide a rich but simplified setting in which to study TLNs. Here, we focus on the theory of fixed point supports $\FP(G)$, and how they can be inferred from the underlying connectivity graph $G$. In a companion paper \cite{CTLN-paper}, we study the relationship between fixed points and dynamic attractors.
 
First, we use our general results about fixed points of TLNs to develop some stronger tools that are specialized to CTLNs. These technical results then allow us to prove a series of {\it graph rules}, connecting $\FP(G)$ directly to the structure of $G$.
In particular, any conclusions about the fixed points of a CTLN that can be obtained from graph rules are automatically independent of the choice of parameters $\varepsilon, \delta,$ and $\theta$, provided these fall within the legal range. Additionally, we study larger networks composed of smaller ``building block'' subnetworks, and prove several theorems relating the fixed points and permitted motifs of the full network to those of its components. Our results demonstrate that CTLNs are surprisingly tractable, and provide the foundation for a ``graphical calculus"  to infer features of a network's dynamics directly from its underlying connectivity. 

The organization of this paper is as follows.  In Section~\ref{sec:fixed-pts-TLNs}, we begin by reviewing some essential background on fixed points of general TLNs, including prior results on their index and stability.  Then in Section~\ref{sec:technical-results} we present a new characterization of $\FP(W,b)$ in terms of sign conditions (Theorem~\ref{thm:sgn-condition}).  In Section~\ref{sec:fixed-pts-CTLNs}, we specialize to CTLNs and show how the sign conditions can allow us to determine fixed point supports directly from the underlying graph $G$. We also introduce additional tools, such as graphical domination and simply-added splits, that are later used to prove various graph rules. Sections~\ref{sec:graph-rules} and~\ref{sec:building-blocks}, on graph rules and their extensions to composite graphs, are in some sense the heart of the paper. In particular, we identify various families of graphs and structures that yield (parameter-independent) permitted and forbidden motifs. Finally, in Section~\ref{sec:domination} we introduce a more general form of domination which is applicable to the broader family of TLNs, and provides
a second characterization of $\FP(W,b)$. This is then used to prove several of our earlier results, including Theorem~\ref{thm:graph-domination} on graphical domination.

\section{Fixed points of TLNs}\label{sec:fixed-pts-TLNs}

\subsection{Some general background}\label{sec:general-background}

We refer to a network with dynamics~\eqref{eq:network2} as a TLN $(W,b)$ on $n$ neurons.
A {\it fixed point} $x^*$ of the network is a point in the state space satisfying 
$\left.\dfrac{dx_i}{dt}\right|_{x = x^*} = 0$ for each neuron $i \in [n]$, where $[n] \od \{1,\ldots,n\}$.
In other words,
\begin{equation}\label{eq:x*}
x_i^* = \left[\sum_{j=1}^n W_{ij}x_j^*+b_i \right]_+,  \quad i = 1,\ldots,n.
\end{equation}
The {\it support} of a fixed point is the subset of active neurons,
$$\supp(x^*) \od \{i \in [n] \mid x_i^*>0\}.$$
We typically refer to supports as subsets $\sigma \subseteq [n]$.  

\begin{definition}
We say that a TLN $(W,b)$ is {\it competitive} if $W_{ij} \leq 0$, $W_{ii} = 0$, and $b_i \geq 0$ for all $i,j \in [n]$.
\end{definition}

Notice from equation~\eqref{eq:x*} that, because of the threshold nonlinearity, we must have $x_i^* \geq 0$ for each $i \in [n]$. For competitive $(W,b)$, it follows that whenever $x_i^* > 0$ we must have $b_i > 0$.
On the other hand, if $b_i = 0$ for all $i \in [n]$, the activity of a competitive TLN will always decay to the fixed point $x = 0$. To rule this out, we will additionally require that $b_i>0$ for at least one neuron, ensuring that all fixed points are nontrivial (see Definition~\ref{def:nondegenerate}). Furthermore, it is straightforward to check that the activity of a competitive TLN is always bounded. In particular, if $x(0) \in \prod_{i=1}^n [0,b_i]$, then $x(t) \in \prod_{i=1}^n [0,b_i]$ for all $t>0$. 

We will often restrict matrices and vectors to a particular subset of neurons $\sigma$.
We use the notation $A_\sigma$ and $b_\sigma$ to denote a matrix $A$ and a vector $b$ that have been truncated to include only entries with indices in $\sigma$. Furthermore, we use the notation $(A_i;b)$ to denote a matrix $A$ whose $i$th column has been replaced by the vector $b$, as in Cramer's rule (see Lemma~\ref{lemma:cramer}).\footnote{The use of the subscript $i$ inside $(A_i;b)$ and $((A_\sigma)_i;b_\sigma)$ has a different meaning than the subscript $\sigma$ in $A_\sigma$, because it refers to replacing the $i$th column by $b$ (or $b_\sigma$), as opposed to restricting the entries of $A$ to the index set $\{i\}$. This is an abuse of notation, but the meaning should always be clear from the context. We will only use the vector replacement meaning inside expressions for Cramer's determinants, such as $\det((I-W_\sigma)_i;b_\sigma)$.} In the case of a restricted matrix, $((A_\sigma)_i;b_\sigma)$ denotes the matrix $A_\sigma$ where the column corresponding to the index $i \in \sigma$ has been replaced by $b_\sigma$ (note that this is not typically the $i$th column of $A_\sigma$).

\begin{definition} \label{def:nondegenerate}
We say that a TLN $(W,b)$ is {\it nondegenerate} if 
\begin{itemize}
\item $b_i > 0$ for at least one $i \in [n]$,
\item $\det(I-W_\sigma) \neq 0$ for each $\sigma \subseteq [n]$, and 
\item for each $\sigma \subseteq [n]$ such that $b_i >0$ for all $i \in \sigma$,  the corresponding Cramer's determinant is nonzero: $\det((I-W_\sigma)_i;b_\sigma) \neq 0$. 
\end{itemize}
\end{definition}
\noindent {\em Unless otherwise specified, we will assume all TLNs are both {\bf competitive} and {\bf nondegenerate}.} Note that almost all networks of the form~\eqref{eq:network2} are nondegenerate, since having a zero determinant is a highly fine-tuned condition.  

If $(W,b)$ is nondegenerate, there can be at most one fixed point per support. To see why, denote
\begin{equation}\label{eq:xsigma}
x^\sigma \od (I-W_\sigma)^{-1} b_\sigma.
\end{equation}
If there exists a fixed point $x^*$ with support $\sigma$, then $x_i^* = x_i^\sigma$ for each $i \in \sigma$ (and is zero otherwise).  It follows from the definition that $\sigma$ is the support of a fixed point if and only if:
\begin{itemize}
\item[(i)] $x_i^\sigma > 0$ for all $i \in \sigma$, and
\item[(ii)]  $\sum_{i\in\sigma} W_{ki}x_i^\sigma+b_k \leq 0$ for all $k \notin \sigma$. 
\end{itemize}
(This is straightforward, but see \cite{pattern-completion} for more details.) 

We denote the set of all fixed point supports of a TLN $(W,b)$ as
$$\FP(W,b) \od \{\sigma \subseteq [n] \mid \sigma \text{ is the support of a fixed point}\}.$$ 
Finding the fixed points of a nondegenerate TLN thus reduces to finding the set of supports $\FP(W,b)$.
Note that because we require $b_i\geq 0,$ and $b_i>0$ for at least one $i$, the empty set (the support of $x = 0$) is never an element of $\FP(W,b)$.  

Since condition (i) above  only depends on $(W_\sigma,b_\sigma)$, a necessary condition 
 for $\sigma \in \FP(W,b)$ is that $\sigma \in \FP(W_\sigma,b_\sigma).$ Such a fixed point survives
 the addition of other nodes $k \notin \sigma$ precisely when condition (ii) is satisfied.
Note that only the {\it survival} of the fixed point depends on the rest of the network; the actual {\it values} of the $x_i^\sigma$ (for $i \in \sigma$) cannot change. For this reason it makes sense to distinguish the subsets $\sigma$ that support a fixed point on the restricted networks $(W_\sigma, b_\sigma)$, irrespective of whether or not these fixed points survive to the full network. In particular, if $\sigma \notin \FP(W_\sigma,b_\sigma)$, then we are guaranteed that $\sigma \notin \FP(W,b)$.

\begin{definition}[permitted and forbidden motifs]\label{def:permitted-forbidden}
 Let $(W,b)$ be a TLN on $n$ neurons. We say that $\sigma \subseteq [n]$ is a {\it permitted motif} of the network if $\sigma \in \FP(W_\sigma,b_\sigma)$. Otherwise, we say that $\sigma$ is a {\it forbidden motif.} 
\end{definition}

\subsection{Index and parity}

For each TLN fixed point, labeled by its support $\sigma \in \FP(W,b)$, we define the {\em index} as
$$\idx(\sigma) \od \sgn \det(I-W_\sigma).$$
Since we assume our TLNs are nondegenerate, $\det(I-W_\sigma) \neq 0$ and thus $\idx(\sigma) \in \{\pm 1\}$.  Moreover, if $\sigma$ is the support of a {\em stable} fixed point, then the eigenvalues of $-I+W_\sigma$ must all have negative real part, and so those of $I-W_\sigma$ all have positive real part. This implies that $\idx(\sigma) = +1$ for all stable fixed points.

The following theorem, given in \cite{CTLN-paper}, indicates that fixed points with index $+1$ and $-1$ are almost perfectly balanced. It also tells us that the parity of the total number of fixed points is always odd.

\begin{theorem}[parity \cite{CTLN-paper}]\label{thm:parity}
Let $(W,b)$ be a TLN. Then,
$$\sum_{\sigma \in \FP(W,b)} \idx(\sigma) = + 1.$$
In particular, the total number of fixed points $|\FP(W,b)|$ is always odd.
\end{theorem}

As an immediate corollary, we obtain an upper bound on the number of stable fixed points, which all have index $+1$.  Here we also use the fact that $|\FP(W,b)| \leq 2^n-1$, which is the number of nonempty subsets of $[n]$.

\begin{corollary}
The number of stable fixed points in a TLN on $n$ neurons is at most $2^{n-1}$.
\end{corollary}

\subsection{Sign conditions}\label{sec:technical-results}
In this section, we provide a new characterization of fixed point supports of TLNs via the signs of particular Cramer's determinants. Recall Cramer's rule:

\begin{lemma}[Cramer's rule] \label{lemma:cramer}
Let $A$ be an $n \times n$ matrix with $\det A \neq 0$, and consider the linear system $Ax=b$.  This has a unique solution, $x = A^{-1}b$.  The entries of $x$ can be expressed as
$$x_i = \dfrac{\det(A_i;b)}{\det A},$$
where $(A_i;b)$ is the matrix obtained from $A$ by replacing the $i$th column with $b$.
\end{lemma}

\noindent Indeed, it follows directly from Cramer's rule that a fixed point of a TLN $(W,b)$ with support $\sigma$ has values
\begin{equation}\label{eq:x_i-s_i}
x_i^\sigma = \dfrac{\det((I-W_{\sigma})_i;b_{\sigma}),}{\det(I-W_\sigma)}, \;\; \text{ for } \; i \in \sigma.
\end{equation}
For any $\sigma \subseteq [n],$ we are thus motivated to define
\begin{equation}\label{eq:s_i}
s_i^\sigma \od \det((I-W_{\sigma\cup\{i\}})_i;b_{\sigma\cup\{i\}}), \;\; \text{for each} \;\; i \in [n].
\end{equation} 
Note that because we only consider TLNs that are nondegenerate (see Definition~\ref{def:nondegenerate}), we can assume all Cramer's determinants are nonzero, and thus $s_i^\sigma \neq 0$ for all $i \in [n]$ and $\sigma \subseteq [n]$.  Moreover, it follows directly from the definition that, for any $i \in [n]$,
\begin{equation}\label{eq:s_k}
s_i^\sigma = s_i^{\sigma \cup \{i\}}. 
\end{equation}
Note that for the empty set we have $s_i^{\emptyset} = s_i^{\{i\}} = b_i$.

It turns out that fixed point supports of (competitive, nondegenerate) TLNs can be fully determined from the signs of the $s_i^\sigma$,  yielding our first characterization of $\FP(W,b)$. Recall that $\sigma$ is a permitted motif of $(W,b)$ if $\sigma \in \FP(W_\sigma,b_\sigma)$.

\begin{theorem}[sign conditions] \label{thm:sgn-condition}
Let $(W,b)$ be a TLN on $n$ neurons.  For any nonempty $\sigma~\subseteq~[n]$,
$$\sigma \text{ is a permitted motif }  \;\; \Leftrightarrow \;\; \sgn s_i^\sigma = \sgn s_j^\sigma
\text{ for all } i,j \in \sigma.$$
When $\sigma$ is permitted, 
$\sgn s_i^\sigma = \sgn\det(I-W_\sigma) = \idx(\sigma)$ for all $i \in \sigma$.
Furthermore,
$$\sigma \in \FP(W,b)  \;\; \Leftrightarrow \;\; \sgn s_i^\sigma = \sgn s_j^\sigma = -\sgn s_k^\sigma
\text{ for all } i,j \in \sigma,\; k \not\in \sigma.$$
 \end{theorem} 

Before we prove Theorem~\ref{thm:sgn-condition}, we need the following lemma, which gives a useful identity for computing $s_k^\sigma$ values.

\begin{lemma} Let $(W,b)$ be a TLN on $n$ neurons, and let $\sigma \subseteq [n]$.  Then
\begin{equation}\label{eq:s_k-identity}
s_k^\sigma = \sum_{i\in\sigma} W_{ki}s^{\sigma}_i + b_k\det(I-W_{\sigma}) \;\; \text{ for any } \; k \in [n].
\end{equation}
\end{lemma}

\begin{proof} For $k \in \sigma$, it follows from the definition of $x^\sigma$, equation~\eqref{eq:xsigma}, that 
$x^\sigma = W_\sigma x^\sigma+b_\sigma$, and thus
 $$x_k^\sigma = \sum_{i\in\sigma} W_{ki}x^{\sigma}_i + b_k.$$
Using equation~\eqref{eq:x_i-s_i}, this yields $s_k^\sigma = \sum_{i\in\sigma} W_{ki}s^{\sigma}_i + b_k\det(I-W_{\sigma})$, as desired.

Next, we consider $k \notin \sigma$ and compute:
$$s_k^{\sigma} = \det((I-W_{\sigma\cup \{k\}})_k;b_{\sigma\cup \{k\}})
 = \det\left(\begin{array}{c|c} I-W_{\sigma} & b_{\sigma} \\ \hline  -W_{k1} \cdots -W_{k,k-1} & b_k \end{array}\right).$$
Applying the Laplace expansion for the determinant along the $k$th row, we obtain
\begin{eqnarray*}
s_k^{\sigma}  &=& 
\sum_{i \in \sigma} (-1)^{i+k}(-W_{ki})(-1)^{(k-1)-i}\det((I-W_{\sigma})_i;b_{\sigma}) 
+ b_k\det(I-W_{\sigma})\\
&=&  \sum_{i\in\sigma} W_{ki}s^{\sigma}_i + b_k\det(I-W_{\sigma}),
\end{eqnarray*}
which completes the proof.
\end{proof}

We are now ready to prove Theorem~\ref{thm:sgn-condition}.

\begin{proof}[Proof of Theorem~\ref{thm:sgn-condition} (sign conditions)]
Recall that $\sigma \in \FP(W_\sigma, b_\sigma)$ if and only if $x_i^\sigma > 0$ for each $i \in \sigma$ (see equation~\eqref{eq:xsigma} and subsequent remarks). 
By Cramer's rule, we have $x_i^\sigma = \dfrac{s_i^\sigma}{\det(I-W_\sigma)}$ (see equation~\eqref{eq:x_i-s_i}). Now suppose $\sigma \in \FP(W_\sigma, b_\sigma)$. Since $x_i^\sigma>0$ for each $i \in \sigma$, we must have $\sgn s_i^\sigma = \sgn s_j^\sigma = \sgn \det(I-W_\sigma) = \idx(\sigma)$ for all $i, j \in \sigma$.  

For the reverse implication, suppose $\sgn s_i^\sigma = \sgn s_j^\sigma$ for all $i, j \in \sigma$.  This immediately implies that all the $x_i^\sigma$ for $i \in \sigma$ have the same sign, but we must show this sign is positive. First, we show that $b_i>0$ for all $i \in \sigma$ (in competitive TLNs, we always have $b_i \geq 0$). To see this, suppose there exists a $j \in \sigma$ such that $b_j = 0$. Then by equation~\eqref{eq:s_k-identity} we would have 
$$s_j^\sigma = \sum_{i\in\sigma} W_{ji}s^{\sigma}_i + b_j\det(I-W_{\sigma}) = \sum_{i\in\sigma \setminus j} W_{ji}s^{\sigma}_i, $$
using the fact that $W_{jj} = 0$. Since $W_{ji} < 0$ for all $i \neq j$, the above equality contradicts the assumption that all the signs of the $s_i^\sigma$ are equal for $i \in \sigma$. We can thus conclude that $b_i > 0$ for all $i \in \sigma$. This in turn ensures that $x_i^\sigma > 0$ for at least one (and thus all) $i \in \sigma$, because by definition $(I-W_\sigma) x_i^\sigma = b_\sigma$, and all entries of $(I-W_\sigma)$ are nonnegative. Hence $\sigma \in \FP(W_\sigma, b_\sigma)$, completing the proof of the first part of the theorem.

To prove the second part, recall that $\sigma \in \FP(W,b)$ precisely when $\sigma \in \FP(W_\sigma, b_\sigma)$ and $\sigma$ satisfies fixed-point condition (ii): $\sum_{i\in\sigma} W_{ki}x_i^\sigma+b_k \leq 0$ for all $k \notin \sigma$.  Using equation~\eqref{eq:s_k-identity} again,
and dividing by $\det(I-W_\sigma)$, we find that 
\begin{equation}
\dfrac{s_k^\sigma}{\det(I-W_\sigma)} = \sum_{i\in\sigma} W_{ki}x_i^\sigma+b_k.
\end{equation}
Since the network is nondegenerate, $\det(I-W_\sigma) \neq 0$ and $s^\sigma_k \neq 0$, and so condition (ii) is equivalent to the sign condition: $\sgn s_k^\sigma = -\sgn \det(I-W_\sigma)$.  Putting this together with the above sign conditions for $\sigma \in \FP(W_\sigma, b_\sigma)$, we see that $\sigma \in \FP(W,b)$ if and only if 
$\sgn s_i^\sigma = \sgn s_j^\sigma = -\sgn s_k^\sigma
\text{ for all } i,j \in \sigma,\; k \not\in \sigma.$
\end{proof}

In the following example, we show how to use Theorem~\ref{thm:sgn-condition} to find $\FP(W, b)$ for a TLN of size 2 with a uniform external input $b$.  
\begin{example}
Consider a TLN with $W = \left(\begin{array}{cc} 0 & W_{12} \\ W_{21} & 0\end{array}\right)$, for some $W_{12}, W_{21} <0$, and external input $b = \one$.  Recall that the empty set is never a fixed point support of a competitive TLN with a positive input, so we restrict to considering nonempty subsets of $\{1,2\}$.  For $\sigma = \{1\}$, we see that $s_1^\sigma = b_1 = 1$, while 
$$s_2^\sigma = \det((I-W_{\sigma \cup 2})_2; b) = \det \left(\begin{array}{cc} 1 & 1 \\ -W_{21} & 1\end{array}\right) = 1+W_{21}.$$
In particular, $\sgn s_2^\sigma = -\sgn s_1^\sigma$ precisely when $W_{21} <-1$.  Thus, by Theorem~\ref{thm:sgn-condition},
$$\{1\} \in \FP(W,b) \; \; \Leftrightarrow \; \; W_{21} <-1.$$
By a similar argument, we have
$$\{2\} \in \FP(W,b) \; \; \Leftrightarrow \; \; W_{12} <-1.$$

Finally, consider $\sigma = \{1,2\}$.  We obtain $s_1^\sigma = 1+W_{12}$ and $s_2^\sigma = 1+W_{21}$.  By Theorem~\ref{thm:sgn-condition}, $\sigma \in \FP(W,b)$ precisely when $\sgn(1+W_{12}) = \sgn(1+W_{21})$, and so 
$$\{1,2\} \in \FP(W,b) \; \; \Leftrightarrow \; \; \left\{\begin{array}{l} W_{12}>-1 \text{ and }\; W_{21} >-1, \text{ or} \\ W_{12}<-1 \text{ and }\; W_{21} <-1 \end{array}\right.$$
Note that in the case where $W$ is a CTLN, the condition $W_{21} <-1$ corresponds to $1 \not\to 2$ in the  associated graph $G$, while $W_{12} <-1$ indicates $2 \not\to 1$ in $G$. So the last condition reduces to $\{1,2\} \in \FP(G|_\sigma)$ if and only if nodes $1$ and $2$ are either bidirectionally connected or disconnected in $G$. 
\end{example}

Since the values of $s_i^\sigma$ for $i \in \sigma$ depend only on $(W_\sigma,b_\sigma)$, while the values $s_k^\sigma$ for $k \notin \sigma$ depend only on $(W_{\sigma \cup \{k\}}, b_{\sigma \cup \{k\}})$, we immediately have the following useful corollary:

\begin{corollary}\label{cor:on-off-conds} Let $(W,b)$ be a TLN on $n$ neurons, and let $\sigma \subseteq [n]$.
The following are equivalent:

\begin{enumerate}
\item $\sigma \in \FP(W,b)$
\item $\sigma \in \FP(W_\tau, b_\tau)$ for all $\sigma \subseteq \tau \subseteq [n]$.
\item $\sigma \in \FP(W_\sigma,b_\sigma)$ and $\sigma \in \FP(W_{\sigma \cup \{k\}}, b_{\sigma \cup \{k\}})$ for all $k \notin \sigma$
\item $\sigma \in \FP(W_{\sigma \cup \{k\}}, b_{\sigma \cup \{k\}})$ for all $k \notin \sigma$
\end{enumerate}
\end{corollary}

Theorem~\ref{thm:sgn-condition} also gives a relationship between the indices of fixed points whose supports differ by only one neuron.

\begin{lemma}[alternation]\label{lemma:alternation}
Let $(W,b)$ be a TLN.  If $\sigma, \sigma \cup \{k\} \in \FP(W,b),$ for $k \notin \sigma,$ are both fixed point supports, then 
$$\idx(\sigma\cup\{k\}) = - \idx(\sigma).$$
\end{lemma}

\begin{proof}
If $\sigma, \sigma \cup \{k\} \in \FP(W,b)$, then by Theorem~\ref{thm:sgn-condition} we have
$\idx(\sigma) = \sgn s_i^\sigma = -\sgn s_k^\sigma$ for any $i \in \sigma$.  Recalling that $s_k^{\sigma \cup \{k\}} = s_k^\sigma$, we see that $\idx(\sigma \cup \{k\}) = -\idx(\sigma)$, as desired. 
\end{proof}

\begin{corollary}\label{cor:support-gap} 
If $\sigma \in \FP(W,b)$ is the support of a stable fixed point, then there is no other stable fixed point with support $\sigma$, $\sigma \setminus k $, or $\sigma \cup \{k\}$ for any $k \in [n]$.
\end{corollary}

\section{Fixed points of CTLNs}\label{sec:fixed-pts-CTLNs}

When $(W,b)$ comes from a CTLN with graph $G$ and parameters $\varepsilon, \delta,$ and $\theta$, so that $W = W(G,\varepsilon,\delta)$ and $b = \theta\bf{1}$, we use the notation 
$$\FP(G) = \FP(G,\varepsilon,\delta) \od \FP(W(G,\varepsilon,\delta), \theta\bf{1}).$$
Note that we always suppress $\theta$ from the notation, since it can easily be seen that the value of $\theta$ does not affect the set of fixed point supports, so long as $\theta >0$. (Different $\theta$ values merely rescale the fixed point values.) In addition, we will typically suppress the $\varepsilon$ and $\delta$ dependence as well, using the simpler notation $\FP(G)$ to denote the set of fixed point supports when $\varepsilon$ and $\delta$ are understood to be fixed. 

For some graphs,  $\FP(G) = \FP(G,\varepsilon,\delta)$ does indeed depend on the choice of $\varepsilon$ and $\delta$, but our theoretical results almost always deal with aspects of $\FP(G)$ that are independent of the choice of parameters, provided these lie within the legal range.\footnote{For any result that does depend on $\varepsilon$ and $\delta$, we will indicate this explicitly.}  For example, a graph rule could tell us about certain fixed point supports that can be ruled out, and are thus not contained in $\FP(G)$. Such a conclusion is parameter independent, even if there are other supports in $\FP(G)$ whose presence depends on parameters.

Recall that $\sigma$ is a permitted motif if $\sigma \in \FP(W_\sigma, b_\sigma)$, and is forbidden otherwise.  We use the same language for CTLNs: $\sigma$ is a permitted motif if $\sigma \in \FP(G|_\sigma)$, where $G|_\sigma$ is the induced subgraph obtained by restricting vertices and edges to the vertex set $\sigma$. If $\sigma \notin \FP(G|_\sigma)$, we say that $\sigma$ is a forbidden motif. Note that this may depend on the choice of parameters $\varepsilon,\delta$.

The results in this section provide a technical foundation for the graph rules and building block rules that are presented in Sections~\ref{sec:graph-rules} and~\ref{sec:building-blocks}, respectively.
Before moving on to the main content, we pause briefly to provide some simple bounds on the total activity at fixed points of CTLNs, irrespective of the support. These bounds, like the actual fixed point values, do depend on $\varepsilon, \delta$ and $\theta$, even when the supports $\FP(G)$ do not.

\begin{lemma}\label{lemma:total-bounds}
If $x^*$ is a fixed point of a CTLN on $n$ nodes, with parameters $\varepsilon, \delta$, and $\theta$, then
$$\dfrac{\theta}{1+\delta} < \sum_{i =1}^n x_i^* < \dfrac{\theta}{1-\varepsilon}.$$
\end{lemma}

\begin{proof}
Let $\sigma$ be the support of $x^*$.  Then $(I-W_\sigma)x_\sigma^* = \theta \one_\sigma$, and so $\theta = R_j \cdot x_\sigma^*$, where $R_j$ is the $j$th row vector of $I-W_\sigma$.  Since all entries of $R_j$ and $x_\sigma^*$ are positive, and the off-diagonal entries of $I-W_\sigma$ are all at least $1-\varepsilon$ and at most $1+\delta$, it follows that
$$(1-\varepsilon) \sum_{i \in \sigma} x_i^* + \varepsilon x_j^* \leq \theta \leq (1+\delta)\sum_{i \in \sigma} x_i^* - \delta x_j^*$$
for any $j \in \sigma$.  Since $x_j^* > 0$, we have
$(1-\varepsilon) \sum_{i \in \sigma} x_i^*< \theta < (1+\delta)\sum_{i \in \sigma} x_i^*,$
which implies the desired result.
\end{proof}

\subsection{Sign rules for CTLNs}\label{sec:sgn-condition-CTLN}
Recall from Section~\ref{sec:technical-results}, that Theorem~\ref{thm:sgn-condition} gives sign conditions for when a subset $\sigma$ is a permitted motif and/or a fixed point support in a general TLN $(W, b)$.  
We illustrate this result for the special case of CTLNs in the following example.

\begin{figure}[!h]
\begin{center}
\includegraphics[width=6.5in]{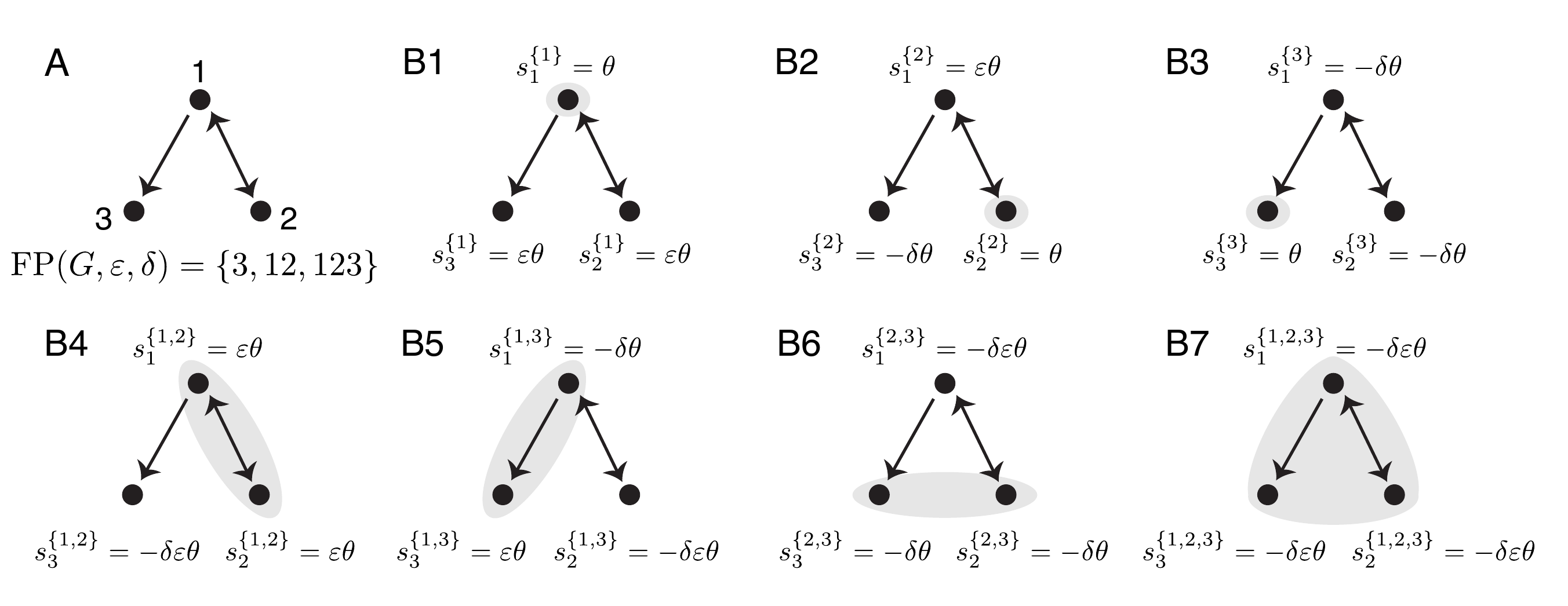}
\vspace{-.15in}
\end{center}
\caption{\textbf{Graph with $s_i^\sigma$ values for all nonempty $\sigma$}.  (A) Graph with $\FP(G, \varepsilon, \delta)$. (B1 - B7) In each panel, the nodes in the gray shaded region comprise $\sigma$, and each node is labeled with its $s_i^\sigma$ value.}
\label{fig:n3-single-example}
\end{figure}

\begin{example}\label{ex:s_i-values}
Let $G$ be the graph in Figure~\ref{fig:n3-single-example}A.  We will use the sign conditions to find the permitted motifs and $\FP(G) = \FP(G, \varepsilon, \delta)$.  Recall from Section~\ref{sec:general-background} that the empty set is never the support of a fixed point, and so we restrict to analyzing the nonempty subsets of $\{1,2,3\}$.  Each panel in Figure~\ref{fig:n3-single-example} shows a choice of $\sigma$ (gray shaded regions) and the values of $s_i^\sigma = \det((I-W_{\sigma \cup i})_i; \theta \one)$ for all $i \in \{1,2,3\}$.  For example, in B1, $s_1^{\{1\}} = \det(\theta) = \theta$ and $s_2^{\{1\}} = \det\left(\begin{array}{cc} 1 & \theta\\ 1-\varepsilon & \theta \end{array}\right) = \varepsilon \theta$.  

Observe that every singleton $\{i\}$ is a permitted motif, as are $\{1,2\}, \{2,3\}$, and $\{1,2,3\}$ since these subsets satisfy $\sgn s_i^\sigma = \sgn s_j^\sigma$ for all $i, j \in \sigma$.  A permitted motif survives as a fixed point of $G$ precisely when the external nodes all have opposite sign for $s_k^\sigma$.  Thus the only singleton fixed point support is $\{3\}$ since it satisfies $\sgn s_i^{\{3\}} = -\sgn s_3^{\{3\}}$ for $i =1,2$ (see B3).  Continuing this analysis, we find that $\FP(G) = \{3, 12, 123\}$.  Furthermore, since the signs of $s_i^\sigma$ are constant across $\varepsilon, \delta >0$ for each $i \in \{1,2,3\}$ and $\sigma \subseteq \{1,2,3\}$, we see that $\FP(G)$ is in fact parameter independent.
\end{example}

Figure~\ref{fig:n3-graphs-s_i} in Appendix Section~\ref{appendixA}, shows all directed graphs of size $n \leq 3$ labeled with the full support $s_i^{[n]}$ values for each node $i \in [n]$.   From this, we see that 6 of the 16 directed graphs of size 3 are permitted motifs, and the remainder are forbidden, independent of the choice of parameters.

\subsection{Simply-added splits}
In this section, we introduce the concept of \emph{simply-added splits}, whereby the vertices of a graph are partitioned in a special way that allows us to easily compute certain $s_i^\sigma$ values.  This kind of structure will play a prominent role in Section~\ref{sec:building-blocks}, where we build larger graphs from component subgraphs.  In particular, simply-added splits are essential to our ability to relate $\FP(G)$ for such a composite graph to the $\FP(G_i)$ of its components.  They also provide an additional tool for proving some of the graph rules in Section~\ref{sec:graph-rules}.

\begin{definition}[simply-added]
Let $G$ be a graph on $n$ nodes, $\tau \subset [n]$ nonempty, and $k \notin \tau$. We say that $k$ is a {\it projector onto $\tau$} if $k \to i$ for all $i \in \tau$. We say that $k$ is a {\it non-projector onto $\tau$} if $k \not\to i$ for all $i \in \tau$. 
For any nonempty $\tau, \omega \subset [n]$ such that $\tau \cap \omega= \emptyset$, we say $\omega$ is {\it simply-added to $\tau$} if for each $k \in \omega$, $k$ is either a projector or a non-projector onto $\tau$ (see Figure~\ref{fig:simply-added}).  In this case, we say that the $(\tau,\omega)$ is a {\it simply-added split} of the subgraph $G|_\sigma$, for $\sigma = \tau \cup \omega$. 
\end{definition}

\begin{figure}[!h]
\begin{center}
\includegraphics[height=1.5in]{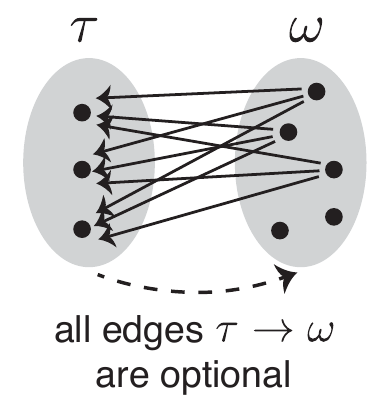}
\vspace{-.15in}
\end{center}
\caption{In this graph, $\omega$ is simply-added to $\tau$ and thus each $k \in \omega$ either sends all possible edges to $\tau$, or no edges. There is no constraint on the edges within $\tau$, within $\omega$, or from $\tau$ to $\omega$.}
\label{fig:simply-added}
\end{figure}

The following theorem shows that the $s_i^\sigma$ values factor nicely when a set of neurons is simply-added to a subgraph. 

\begin{theorem}[simply-added]  \label{thm:simply-added}
Let $G$ be a graph on $n$ nodes, and let $\tau,\omega \subset [n]$ such that $\omega$ is simply-added to $\tau$.  For $\sigma = \tau \; {\cup} \:\omega$, 
$$s_i^\sigma = \dfrac{1}{\theta} s_i^{\omega} s_i^{\tau} = \alpha s_i^{\tau}  \;\; \text{ for each } \; i \in \tau,$$
where $\alpha = \frac{1}{\theta} s_i^{\omega}$ has the same value for every $i \in \tau$.

More generally, for any $\sigma \subseteq \tau\; {\cup} \:\omega$, 
$$s_i^\sigma = \dfrac{1}{\theta} s_i^{\sigma \cap \omega} s_i^{\sigma \cap \tau} = \alpha s_i^{\sigma \cap \tau} \;\; \text{ for each } \; i \in \tau,$$
where $\alpha = \frac{1}{\theta} s_i^{\sigma\cap\omega}$. 
\end{theorem}

Recall, using equation~\eqref{eq:s_k}, that $s_i^{\emptyset} = \theta$, since for CTLNs $b_i = \theta$ for each $i \in [n]$. With this convention, the above theorem holds even if $\sigma \cap \tau$ or $\sigma \cap \omega$ is empty.

To prove Theorem~\ref{thm:simply-added} we use the standard formula for the determinant of a $2 \times 2$ block matrix,
\begin{equation}\label{eq:2x2}
\det \begin{bmatrix} A & B \\ C & D\end{bmatrix} = \det(A)\det(D-C A^{-1}B),
\end{equation}
which applies as long as $A$ is invertible.

\begin{proof}[Proof of Theorem~\ref{thm:simply-added}]
Since $\omega$ is simply-added to $\tau$, each $k \in \omega$ is either a projector or non-projector onto $\tau$.  Thus, for each $k \in \omega$, $W_{ik} = -1+c_k$ for all $i \in \tau$, where $c_k=\varepsilon$ if $k$ is a projector or $c_k=-\delta$ if $k$ is non-projector.  

First consider $\sigma = \tau \cup \omega$. For any $i \in \tau$, we compute:
\begin{eqnarray*}
s_i^\sigma &=& \theta\det((I-W_\sigma)_i;\one) = 
\theta\det\left(\begin{array}{c | c | c c c} 
 & 1 & 1-c_{k_1} & \cdots & 1-c_{k_\ell}\\
I-W_{\tau\setminus \{i\}} & \vdots & \vdots & & \vdots\\
 & 1 & 1-c_{k_1} & \cdots & 1-c_{k_\ell}\\
\hline
* & 1 & 1-c_{k_1} & \cdots & 1-c_{k_\ell}\\
 \hline
 * & \one & & I-W_\omega &
\end{array}\right)\\
&=& \theta\det\left(\begin{array}{c | c | c c c} 
 & 1 & 0 & \cdots & 0\\
I-W_{\tau\setminus \{i\}}  & \vdots & \vdots & & \vdots\\
 & 1 & 0 & \cdots & 0\\
 \hline
  * & 1 & 0 & \cdots & 0\\
  \hline
 * & \one & & I-W_\omega - A &
\end{array}\right) 
= \theta \det((I-W_\tau)_i;\one)\det(I-W_\omega-A),
\end{eqnarray*}
where $A$ consists of the modifications to the $I-W_\omega$ block that arise from clearing out the top right block using the column of 1s.

Now observe that $\theta\det((I-W_\tau)_i;\one) = s_i^\tau$, while
\begin{eqnarray*}
\det(I-W_\omega-A) &=& 
\det\left(\begin{array}{c | c c c} 
 1 & 0 & \cdots & 0\\
  \hline
 \one & & I-W_\omega - A &
\end{array}\right) 
= \det\left(\begin{array}{c | c c c} 
 1 & 1-c_{k_1} & \cdots & 1-c_{k_\ell}\\
  \hline
 \one & & I-W_\omega &
\end{array}\right) \\
&=& \det((I-W_{\omega\cup\{i\}})_i;\one) = \frac{1}{\theta}s_i^{\omega\cup\{i\}} =\frac{1}{\theta}s_i^\omega.
\end{eqnarray*}
We therefore have $s_i^\sigma = \frac{1}{\theta}s_i^\tau s_i^\omega$, and see that $s_i^\omega$ has the same value for each $i \in \tau$.

Now consider $\sigma \subseteq \tau \cup \omega$. Note that $\sigma = (\sigma \cap \tau) \cup (\sigma \cap \omega)$, where $\sigma \cap \omega$ is simply-added to $\sigma \cap \tau$ (since $\omega$ is simply-added to $\tau$). To compute $s_i^\sigma$ for $i \in \tau$, we consider two cases. First, suppose $i \in \sigma \cap \tau$. In this case, applying the previous formula to $\sigma = (\sigma \cap \tau) \cup (\sigma \cap \omega)$ immediately yields the desired result. Now consider $i \in \tau \setminus \sigma$. In this case, we can apply the previous formula to $\sigma \cup i = ((\sigma \cap \tau) \cup i) \cup (\sigma \cap \omega)$, and so
$$s_i^\sigma = s_i^{\sigma \cup i} = \dfrac{1}{\theta} s_i^{\sigma \cap \omega} s_i^{(\sigma \cap \tau) \cup i} = \dfrac{1}{\theta} s_i^{\sigma \cap \omega} s_i^{\sigma \cap \tau},$$
where the first and last equalities use equation~\eqref{eq:s_k}. Thus, the desired formula for $s_i^\sigma$ holds for all $i \in \tau$.
\end{proof}

Note that if $\omega = \{k\}$, where $k$ is a projector onto $\tau$, then $s_i^\omega = \varepsilon\theta$ for each $i \in \tau$. If, on the other hand, $k$ is a non-projector onto $\tau$, then $s_i^\omega = -\delta\theta$ for each $i \in \tau$. To see this, observe that $s_i^{\{k\}} = s_i^{\{i,k\}} = (1+W_{ik})\theta$, which evaluates to $\varepsilon\theta$ if $k\to i$, and $-\delta\theta$ for $k \not\to i$. This gives the following useful corollary.

\begin{corollary}\label{cor:simply-added}
Suppose $\sigma = \tau \cup \{k\}$, where $k$ is simply-added onto $\tau$. If $k$ is a projector, then
$s_i^\sigma = \varepsilon s_i^\tau$ for each $i \in \tau$. If $k$ is a non-projector, then $s_i^\sigma = -\delta s_i^\tau$
for each $i \in \tau$.
\end{corollary}

\begin{example}
Let $G$ be the graph in Figure~\ref{fig:n3-single-example}A.  For $\sigma= \{1,2,3\}$, we see that $\sigma$ has a simply-added split $\sigma = \tau \cup \omega$ where $\tau = \{1,2\}$ and $\omega = \{3\}$ is a non-projector onto $\tau$.  Applying Corollary~\ref{cor:simply-added}, we have $s_i^\sigma = -\delta s_i^{\tau}$ for $i=1,2$.  We see from panel B4 of Figure~\ref{fig:n3-single-example} that $s_i^{\tau} = \varepsilon \theta$, and (B7) shows $s_i^\sigma = -\delta \varepsilon \theta$ for $i=1, 2$.  Furthermore, $\sigma$ actually has a second simply-added split: $\sigma = \tau \cup \omega$, where $\tau= \{2,3\}$ and $\omega=\{1\}$ is a projector onto $\tau$.  In this case, Corollary~\ref{cor:simply-added} guarantees that $s_i^\sigma = \varepsilon s_i^{\tau}$ for $i=2,3$.  Since $s_i^\tau = -\delta\theta$ for $i=2, 3$ (see panel B6 of Figure~\ref{fig:n3-single-example}), we see that this second simply-added split gives a consistent value for $s_2^\sigma$.
\end{example}

\subsection{Graphical domination}\label{sec:graph-domination}

One of the most important tools we will use to prove the graph rules in Section~\ref{sec:graph-rules} is the concept of {\it graphical domination}. This refers to certain ``domination'' relationships between vertices in the graph of a CTLN that can be identified by examining the graph alone (e.g., without computing $s_i^\sigma$ values). The presence of such a relationship within a subgraph $G|_\sigma$ is sufficient to guarantee that $\sigma$ is a forbidden motif. Graphical domination can also be used to determine whether or not a permitted motif survives as a fixed point support in a larger network. These facts are collected in Theorem~\ref{thm:graph-domination}. Because they rely only on the graph structure, any results about a CTLN obtained from graphical domination are automatically parameter independent, within the legal range. 

Although we will make frequent use of Theorem~\ref{thm:graph-domination} in Sections~\ref{sec:graph-rules} and~\ref{sec:building-blocks}, we postpone the proof until Section~\ref{sec:graph-dom-proof}. This is because graphical domination is a special case of a more general notion of domination, which we introduce in Section~\ref{sec:domination}.

\begin{definition}\label{def:graph-domination}
We say that $k$ \emph{graphically dominates} $j$ with respect to $\sigma$ if $\sigma \cap \{j, k\} \neq \emptyset$ and the following three conditions all hold:
\begin{itemize}
\item[(1)] for each $i \in \sigma \setminus \{j, k\}$, if $i \to j$ then $i \to k$,
\item[(2)] if $j \in \sigma$, then $j \to k$, and 
\item[(3)] if $k \in \sigma$, then $k \not\to j$.
\end{itemize}
\end{definition}

\begin{figure}[!h]
\begin{center}
\includegraphics[width=5.75in]{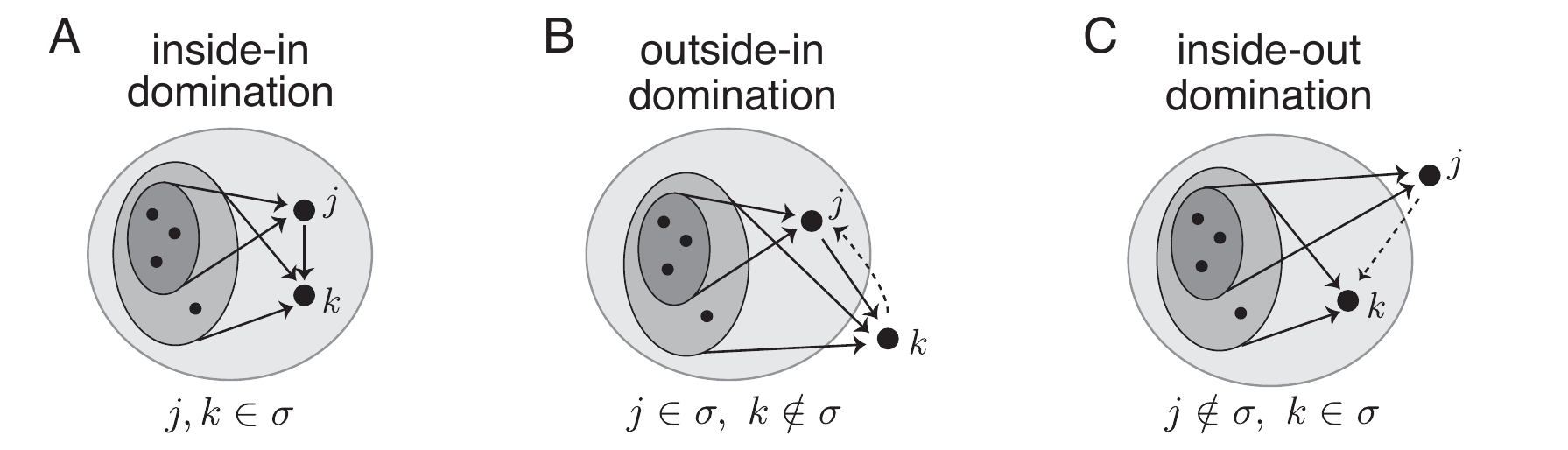}
\vspace{-.15in}
\end{center}
\caption{{\bf The three cases of graphical domination in Theorem~\ref{thm:graph-domination}.} 
In each panel, $k$ graphically dominates $j$ with respect to $\sigma$ (the outermost shaded region). The inner shaded regions illustrate the subsets of nodes that send edges to $j$ and $k$. Note that the vertices sending edges to $j$ are a subset of those sending edges to $k$, but this
containment need not be strict. Dashed arrows indicate optional edges between $j$ and $k$.}
\label{fig:domination}
\end{figure}

\begin{theorem}[graphical domination]\label{thm:graph-domination}
Suppose $k$ graphically dominates $j$ with respect to $\sigma$.  Then the following statements all hold:
\begin{itemize}
\item[(a)] If $j,k \in \sigma$, then $\sigma \notin \FP(G|_\sigma)$, and so $\sigma \notin \FP(G).$ (inside-in domination)
\item[(b)] If $j \in \sigma$ and $k \not\in \sigma$, then $\sigma \notin \FP(G|_{\sigma\cup\{k\}})$, and so $\sigma \notin \FP(G)$. (outside-in domination) 
\item[(c)] If $j \notin \sigma$, $k \in \sigma$, and $\sigma \in \FP(G|_\sigma)$, then $\sigma \in  \FP(G|_{\sigma \cup \{j\}}).$ (inside-out domination)
\end{itemize}
\end{theorem}

Graphical domination is a special case of a more general notion of domination, which applies to any TLN. We will discuss general domination in Section~\ref{sec:domination}, and show in Theorem~\ref{thm:domination} that it gives us a complete characterization of fixed point supports, similar to the sign conditions in Theorem~\ref{thm:sgn-condition}.  In Section~\ref{sec:graph-dom-proof}, we will use general domination to prove Theorem~\ref{thm:graph-domination}.

Part (a) of Theorem~\ref{thm:graph-domination} tells us that if there is any graphical domination {\it inside} of $\sigma$, then $\sigma$ is not a permitted motif, and therefore cannot be a fixed point support in $\FP(G)$. Part (b) tells us that if there is any node $j \in \sigma$ that is dominated by a $k$ {\it outside} of $\sigma$, then $\sigma$ does not survive as a fixed point support of $G$, irrespective of whether or not $\sigma$ is a permitted motif. Finally, part (c) implies that if for each $j \notin \sigma$ there exists a $k$ {\it inside} $\sigma$ that dominates $j$, then $\sigma$ is guaranteed to survive as a fixed point support in $\FP(G)$, provided $\sigma$ is a permitted motif.  

It is important to note that any results obtained via graphical domination hold for all values of $\varepsilon, \delta$ within the legal range.  Thus, if $\sigma$ is forbidden by graphical domination, it is forbidden independent of parameters; similarly, if $\sigma$ survives (or does not survive) due to graphical domination, this fact is also parameter independent. This is not necessarily the case for forbidden motifs that do {\it not} have graphical domination (see Appendix Section~\ref{appendixC} for examples). 

We will see in Section~\ref{sec:building-blocks} that when a motif is forbidden by graphical domination, this feature gets inherited to larger graphs that have this motif as their ``skeleton.''  This follows from the fact that graphical domination interplays nicely with simply-added splits.\footnote{Suppose $\omega$ is simply-added to $\tau$. If $j,k \in \tau$, and $k$ graphically dominates $j$ with respect to $\tau$, then $k$ also graphically dominates $j$ with respect to $\tau \cup \omega$. This is because all edges from $\omega$ to $j$ have corresponding edges from $\omega$ to $k$.} For this reason, it is useful to define a stronger notion of forbidden motif:

\begin{definition}[strongly forbidden] Consider a CTLN on $n$ nodes. We say that $\sigma \subseteq [n]$ is {\em strongly forbidden} if there exist $j,k \in \sigma$ such that $k$ graphically dominates $j$ with respect to $\sigma$. 
\end{definition}

\subsection{Uniform in-degree}
Here we introduce {\it uniform in-degree} graphs, which comprise a large and interesting family of permitted motifs. This family is particularly nice because the survival rules are parameter independent, and can be easily checked directly from the graph (see Theorem~\ref{thm:uniform-in-degree}).

\begin{definition}
We say that $\sigma$ has {\it uniform in-degree} $d$ if every $i \in \sigma$ has in-degree $d_i^{\mathrm{in}} = d$ within $G|_\sigma$.     
\end{definition}

\noindent Note that the subgraph $G|_\sigma$ could have uniform in-degree, but the nodes of $\sigma$ may have different degrees with respect to the full graph $G$.

There exist graphs with uniform in-degree $d$ for any $0 \leq d \leq |\sigma|-1$.  When $d = 0,$ the graph is necessarily an \emph{independent set}, i.e.\ a collection of nodes with no edges between them; when $d = |\sigma|-1$, it is a \emph{clique}, i.e.\  every pair of nodes is connected by a bidirectional edge.  For in-between values of $d$, however, there are several distinct possibilities.  For example, if $|\sigma|=4$ and $d=1$, the subgraph $G|_\sigma$ could be a $4$-cycle, a pair of $2$-cliques, or a $3$-cycle with a single out-going edge to a fourth node (we call this last graph the ``tadpole'').  Notice that the tadpole graph illustrates that a uniform in-degree subgraph need not be cyclically symmetric, and need not have uniform out-degree.  
\medskip

\begin{figure}[!h]
\begin{center}
\includegraphics[width=5in]{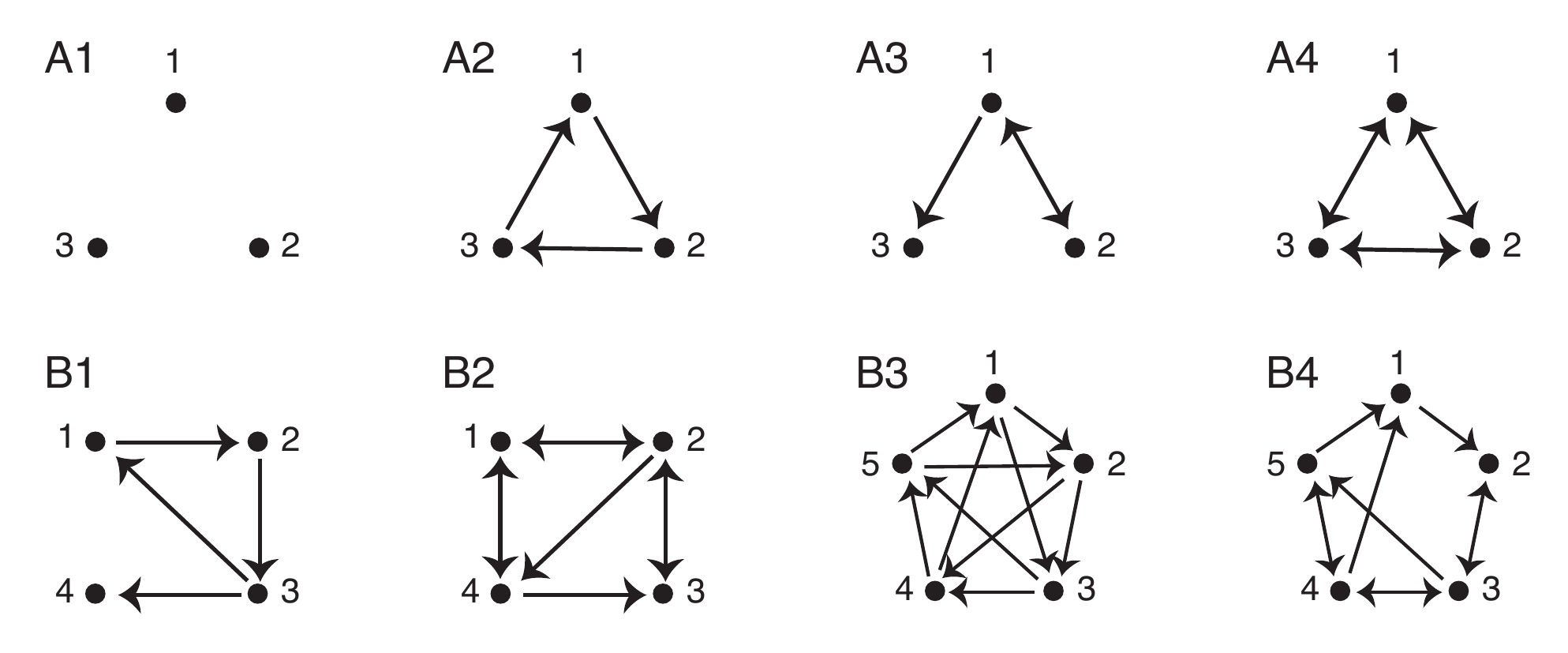}
\end{center}
\vspace{-.075in}
\caption{Examples of uniform in-degree graphs. (A1-A4) All graphs with uniform in-degree on $n=3$ vertices. (B1-B4) Some examples with $n=4$ and $n=5$ vertices.}
\end{figure}

\begin{theorem}[uniform in-degree] \label{thm:uniform-in-degree}
Let $\sigma$ have uniform in-degree $d$. For $k \notin \sigma$, let $d_k \od |\{i \in \sigma \mid i \to k\}|$ be the number of edges $k$ receives from $\sigma$.
 Then 
 $$\sigma \in \FP(G|_{\sigma \cup \{k\}}) \;\; \Leftrightarrow \;\; d_k \leq d.$$ 
Furthermore, if $|\sigma|>1$ and $d <|\sigma|/2$, then the fixed point is unstable.  If $d=|\sigma|-1$, i.e.\ if $\sigma$ is a clique, then the fixed point is stable.
\end{theorem}

\noindent The proof of Theorem~\ref{thm:uniform-in-degree} requires general domination, and is thus postponed to Section~\ref{sec:proofs-uniform-in-degree}. 
Note that this result, like graphical domination, is parameter independent.

Next we will show how the eigenvalues of a uniform in-degree subnetwork $I- W_\tau$ are inherited when a subset $\omega$ is simply-added to $\tau$.  Recall that in the proof of Theorem~\ref{thm:simply-added} (simply-added), the determinant calculation was simplified by the insertion of the all-ones vector, which occurred in the expressions $s_i^\sigma = \theta\det((I-W_\sigma)_i;\one).$  When $\tau$ has uniform in-degree, we can apply a similar trick to compute the eigenvalues of $I-W_\sigma$, where $\sigma = \tau \cup \omega$.

\begin{lemma}\label{lemma:simply-added-evals}
Let $\sigma$ have a simply-added split $\tau \cup \omega$, where $\tau$ is uniform in-degree. Let $R_\tau$ be the row sum of $I-W_\tau$. Note that this is the top (Perron-Frobenius) eigenvalue of $I-W_\tau$. Then 
$$\eig(I-W_\sigma) \supset \eig(I-W_\tau) \setminus R_\tau.$$ 
So all the eigenvalues of $\tau$ get inherited, except possibly the top one $R_\tau$. In particular, if $\tau$ is unstable then $\sigma$ is also unstable.
\end{lemma}

\begin{proof}
To calculate $\eig(I-W_\sigma)$, observe that 
$$\det(I-W_\sigma - \lambda I)= \det\left(\begin{array}{c|ccc} 
& \alpha_1 & \ldots & \alpha_{|\omega|}\\
I-W_\tau-\lambda I & \vdots & & \vdots \\
& \alpha_1 & \ldots & \alpha_{|\omega|}\\
\hline
 * & & I-W_\omega - \lambda I
 \end{array} \right)
 $$
 where $\alpha_i = 1-\varepsilon$ if $i$ is a projector onto $\tau$, and $\alpha_i = 1+\delta$ otherwise.  Let 
 $$A= ((I - W_\tau - \lambda I)_{|\tau|} ; (R_\tau - \lambda)\one)$$
 be the matrix obtained from $I - W_\tau - \lambda I$ by replacing the last column with a column of $R_\tau - \lambda$.  Note that $\det(I - W_\tau - \lambda I) = \det(A)$ since $A$ can be obtained from $I - W_\tau - \lambda I$ by adding all the columns of the matrix to its last column.  Performing these same column operations on the full matrix yields
 $$\det(I-W_\sigma - \lambda I)= \det\left(\begin{array}{c|ccc} 
& \alpha_1 & \ldots & \alpha_{|\omega|}\\
A & \vdots & & \vdots \\
& \alpha_1 & \ldots & \alpha_{|\omega|}\\
\hline
 * & & I-W_\omega - \lambda I
 \end{array} \right).
 $$
Using the last column of $A$ (scaled by $\frac{\alpha_i}{R_\tau- \lambda}$), we can clear out the top right of the matrix to obtain 
 $$\det(I-W_\sigma - \lambda I)= \det\left(\begin{array}{c|c} 
 A & 0\\
 \hline
 * & B
  \end{array} \right),
 $$
 where $B = \frac{1}{R_\tau- \lambda} C$ for some matrix $C$ whose entries are all polynomial in $\lambda$.  Note that $\det B = \dfrac{1}{(R_\tau- \lambda)^{|\omega|}} \det C$, and so 
 $$\det(I-W_\sigma - \lambda I) = \det(A) \det(B) = \det(I-W_\tau - \lambda I) \dfrac{\det(C)} {(R_\tau- \lambda)^{|\omega|}} .$$
Thus all the roots of $\det(I-W_\tau - \lambda I)$ must also be roots of $\det(I-W_\sigma - \lambda I)$ except possibly $\lambda = R_\tau$, and so 
 $\eig(I-W_\sigma) \supset \eig(I-W_\tau) \setminus R_\tau$ as desired.
\end{proof}

This result will be useful in Section~\ref{sec:composite-graphs}, in the context of composite graphs (see Proposition~\ref{prop:composite-unstable}).

\section{Graph rules for CTLNs}\label{sec:graph-rules}

In this section, we prove a variety of graph rules characterizing fixed point supports of CTLNs in terms of the underlying graph $G$.  These are truly ``graph rules" in that they depend only on $G,$ and
are thus independent of the choice of parameters $\varepsilon, \delta$ and $\theta$ (provided these fall within the legal range).  We will again use the streamlined notation $\FP(G)$ to denote the set of fixed point supports, as in Section~\ref{sec:fixed-pts-CTLNs}.

A few of our results express the relationship between $\FP(G)$ and $\FP(G')$ for a pair of related graphs $G$ and $G'$ (see, e.g., Rule~\ref{rule:sources}), or have the form $\sigma \cup \{k\} \in \FP(G)$ if and only if $\sigma \in \FP(G)$ (e.g., Rule~\ref{rule:added-sink}). 
Such {\it relationships} are also parameter independent, even when the sets $\FP(G)$ and $\FP(G')$ 
are themselves parameter dependent. In particular, for a different set of parameters, the fixed point supports $\FP(G)$ and $\FP(G')$ may change, but the given relationship between them stays the same. These statements should thus be understood as applying to CTLNs with fixed and matching parameters.

\subsection{Uniform in-degree and parity}
We begin with a very simple rule that follows directly from Theorem~\ref{thm:parity} (parity).

\begin{rules}[parity]\label{rule:parity}
For any graph $G$, the total number of fixed points $|\FP(G)|$ is odd.
\end{rules}

Next, we digest Theorem~\ref{thm:uniform-in-degree} to obtain graph rules for some important special cases of uniform in-degree graphs: independent sets, cliques, and cycles.  First we repackage Theorem~\ref{thm:uniform-in-degree} as a rule to include it here for completeness.

\begin{rules}[uniform in-degree]\label{rule:uniform-in-deg}
Let $\sigma$ have uniform in-degree $d$. For $k \notin \sigma$, let $d_k \od |\{i \in \sigma \mid i \to k\}|$ be the number of edges $k$ receives from $\sigma$. Then 
$$\sigma \in \FP(G) \;\; \Leftrightarrow \;\; d_k \leq d \;\text{ for all }\; k \not\in \sigma.$$
\end{rules}

An \emph{independent set} is a collection of nodes with no edges between them within the restricted subgraph; such a set is uniform in-degree with $d=0$.  Recall a \emph{sink} is a node with no outgoing edges.

\begin{rules}[independent sets]\label{rule:indep-set}
Suppose $\sigma \subseteq [n]$ is an independent set.  Then 
$$\sigma \in \FP(G) \; \; \Leftrightarrow \; \; \sigma \text{ is a union of sinks.}$$
\noindent Furthermore, when $\sigma \in \FP(G)$, the fixed point is stable if and only if $|\sigma| = 1$. Additionally, $\idx(\sigma) = (-1)^{|\sigma|-1}$.
\end{rules}

To see the index formula in Rule~\ref{rule:indep-set}, observe that a singleton has index $1$, and by Lemma~\ref{lemma:alternation} (alternation), the indices must alternate their signs according to the size of the independent set.

A subset $\sigma$ is a \emph{clique} if every pair of nodes has a bidirectional edge between them in $G|_\sigma$, and thus it has uniform in-degree $d = |\sigma|-1$.  We say that a node $k \notin \sigma$ is a \emph{target of $\sigma$} when $i \to k$ for all $i \in \sigma$. We say $\sigma$ is \emph{target-free} if there is no node $k \in [n] \setminus \sigma$ that is a target of $\sigma$.  

\begin{rules}[cliques]\label{rule:clique}
Suppose $\sigma \subseteq [n]$ is a clique.  Then 
$$\sigma \in \FP(G) \; \; \Leftrightarrow \; \; \sigma \text{ is target-free.}$$
\noindent Furthermore, when $\sigma \in \FP(G)$, the fixed point is always stable and has $\idx(\sigma) = +1$.
\end{rules}

Recall that a {\it cycle} is a cyclically symmetric subgraph with uniform in-degree 1. Applying Theorem~\ref{thm:uniform-in-degree} yields the following result.

\begin{rules}[cycles]\label{rule:cycle}
Suppose $\sigma \subseteq [n]$ is a cycle.  Then 
$$\sigma \in \FP(G) \; \; \Leftrightarrow \; \; \text{for all } k \notin \sigma, k \text{ receives at most 1 edge from } \sigma.$$
\noindent Furthermore, when $\sigma \in \FP(G)$, the fixed point is always unstable but $\idx(\sigma) = +1.$
\end{rules}

The index in Rules~\ref{rule:clique} and~\ref{rule:cycle} follows from the fact that the sum of the indices of fixed points is always 1 (Theorem~\ref{thm:parity}), together with the observation that no proper subset of a clique or a cycle can support a fixed point. In the case of cliques, every proper subset is a clique with a target, and thus does not survive. For a cycle,  it is easy to see that any proper subset either contains a domination relationship or is an independent set that does not survive by Rule~\ref{rule:indep-set}.

\subsection{Adding a single node}
Here we consider how fixed point supports are affected by the addition of a single node to a graph; specifically, we can fully characterize the effect on the fixed points when the added node is a source, sink, projector, or target (see Figure~\ref{fig:graph-vocab}). A \emph{source} is a node that has no incoming edges. We say a source is \emph{proper} if it has at least one outgoing edge.  A \emph{sink} is a node that has no outgoing edges. Recall that a node $k$ is a \emph{projector onto $\sigma$} if $k \to i$ for all $i \in \sigma$; note that there may or may not be edges back from $\sigma$ to $k$.  Finally, recall that a node $k$ is a \emph{target of $\sigma$}, if it receives an edge from every node in $\sigma$ (again there may or may not be back edges from $k$ to $\sigma$).

\begin{figure}[!h]
\begin{center}
\includegraphics[width = .9\textwidth]{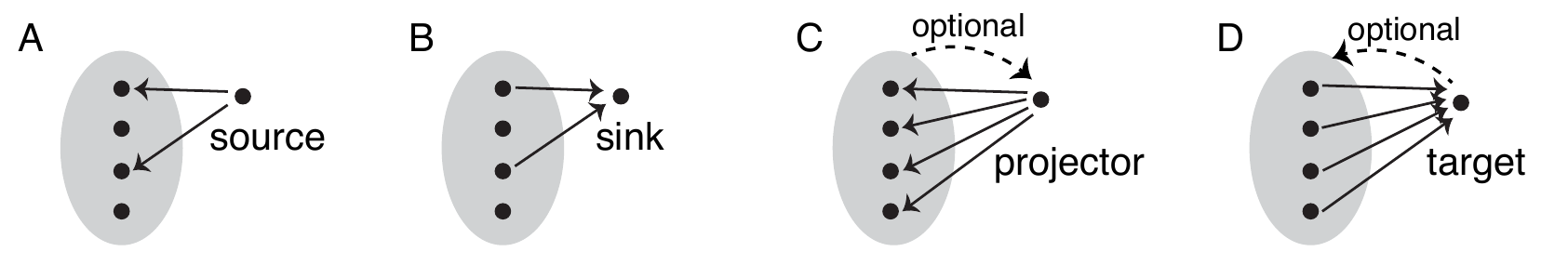}
\caption{(A) A (proper) source in the graph. (B) A sink in the graph. (C) A projector onto the gray region.  (D) A target of the gray region.  Dashed lines indicate that all back edges are optional.}
\vspace{-.1in}
\end{center}
\label{fig:graph-vocab}
\end{figure}

The proofs of the remaining rules will make use of the following technical results from the last three sections:
\begin{itemize}
\item Theorem~\ref{thm:sgn-condition} (sign conditions) characterizing $\sigma \in \FP(G)$ via the signs of the $s_i^\sigma$;
\item Corollary~\ref{cor:on-off-conds} which presents equivalent conditions for $\sigma \in \FP(G)$ in terms of survival of $\sigma$ in intermediate subgraphs $G|_{\tau}$ for $\tau \supseteq \sigma$;
\item Theorem~\ref{thm:simply-added} (simply-added) and Corollary~\ref{cor:simply-added}, which describe how $s_i^\sigma$ factors when $\sigma$ contains a simply-added split; and 
\item Theorem~\ref{thm:graph-domination} (graphical domination) guaranteeing that if $\sigma$ has inside-in or outside-in domination, then $\sigma \notin \FP(G)$, while the presence of inside-out domination of a node $k \notin \sigma$ ensures the survival of a permitted motif $\sigma$ to $\FP(G|_{\sigma \cup k})$.
\end{itemize}

The first two rules show that no fixed point support can contain a node that is a proper source in the restricted subgraph, and that any node that is a proper source in all of $G$ can be removed without any effect on $\FP(G)$. 

\begin{rules}[sources in $\sigma$]\label{rule:source-in-sigma}
Let $G$ be a graph on $n$ nodes and $\sigma \subseteq [n]$.  If there exists a $k \in \sigma$ such that $k$ is a proper source in $G|_{\sigma}$, then $\sigma \notin \FP(G)$.  More generally, if there exists an $\ell \in [n]$ such that $\sigma$ contains a proper source in $G|_{\sigma \cup \{\ell\}}$, then $\sigma \notin \FP(G)$.
\end{rules}
\begin{proof}
Suppose there exists an $\ell \in [n]$ such that $k \in \sigma$ is a proper source in $G|_{\sigma \cup \{\ell\}}$ with $k \to \ell$. Then $\ell$ graphically dominates $k$ since $k$ has no inputs in $G|_{\sigma \cup \{\ell\}}$ and $k \to \ell$.  Hence, $\sigma \notin \FP(G|_{\sigma \cup \{\ell\}})$ by Theorem~\ref{thm:graph-domination}, and so $\sigma \notin \FP(G)$ by Corollary~\ref{cor:on-off-conds}.  
\end{proof}

\begin{rules}[sources in $G$]\label{rule:sources}
Let $G$ be a graph on $n$ nodes and $k \in [n]$.  If $k$ is a proper source in $G$, then 
$$\FP(G) = \FP(G|_{[n]\setminus\{k\}}).$$
\end{rules}
\begin{proof}
Since $k$ is a proper source in $G$, there exists $\ell \in [n]$ such that $k \to \ell$.  Thus, for any $\sigma\subseteq [n]$ with $k \in \sigma$, $\sigma$ contains a proper source in $G|_{\sigma \cup \{\ell\}}$, and so $\sigma \notin \FP(G)$ by Rule~\ref{rule:source-in-sigma}.  On the other hand, for $\sigma \subseteq [n] \setminus \{k\}$, $\sigma \in \FP(G)$ implies $\sigma \in \FP(G|_{[n]\setminus\{k\}})$, by Corollary~\ref{cor:on-off-conds}.  Hence, $\FP(G) \subseteq \FP(G|_{[n]\setminus\{k\}}).$
To see the reverse inclusion, let $\sigma \in \FP(G|_{[n]\setminus\{k\}})$.  For any $i \in \sigma$, $i$ inside-out dominates $k$ with respect to $\sigma$, and so by Theorem~\ref{thm:graph-domination}, $\sigma$ survives and thus $\sigma \in \FP(G)$.  Therefore, $\FP(G) \supseteq \FP(G|_{[n]\setminus\{k\}})$.
\end{proof}

As an illustration of Rule~\ref{rule:sources}, consider the graph in Figure~\ref{fig:graph-series}B, which is obtained from the graph in A by adding  a proper source, node 4.  Observe that $\FP(G)$ is identical for these two graphs, since the fixed points of $G$ are preserved upon removal of the proper source.

\begin{figure}[!h]
\begin{center}
\includegraphics[width = .8\textwidth]{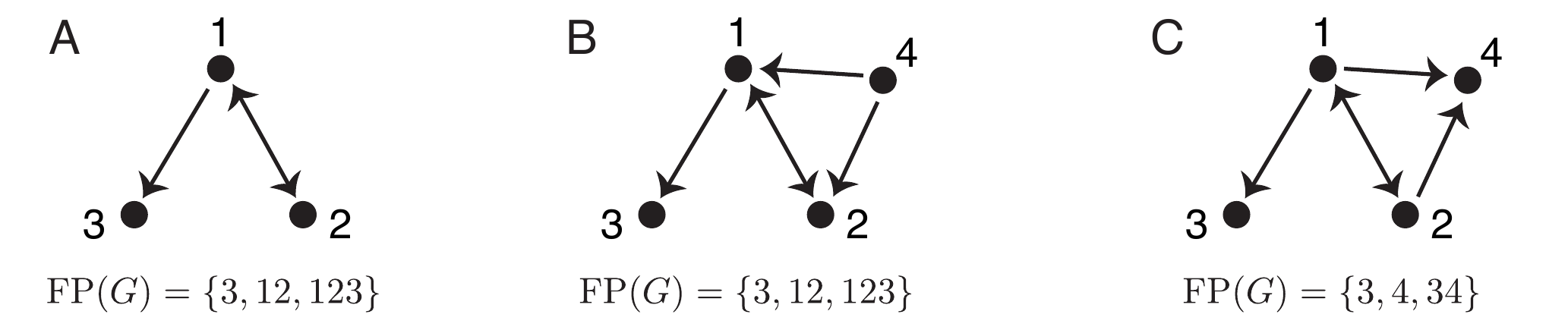}
\caption{Three example graphs with their respective $\FP(G)$. }
\vspace{-.1in}
\label{fig:graph-series}
\end{center}
\end{figure}

 Rule~\ref{rule:added-sink} shows that the addition of a sink $k$ creates a new fixed point support $\sigma \cup \{k\}$ precisely when $\sigma$ is a surviving fixed point of the original graph.

\begin{rules}[sinks]\label{rule:added-sink}
Let $G$ be a graph on $n$ nodes and $k \in [n]$. If $k$ is a sink in $G$, then for any nonempty $\sigma \subseteq [n] \setminus \{k\}$,
$$  \sigma \cup \{k\} \in \FP(G) \;\; \Leftrightarrow \;\;\sigma \in \FP(G),$$
and $\idx(\sigma \cup \{k\}) = - \idx(\sigma).$
\end{rules}
\begin{proof}
Since $k$ is a sink in $G$, $k$ is a non-projector onto $[n] \setminus \{k\}$ and for all $i \in [n] \setminus \{k\}$, we have $s_i^{\sigma \cup \{k\}} = -\delta s_i^{\sigma}$ by Corollary~\ref{cor:simply-added}.  Thus, for all $i, j \in \sigma$, we have $\sgn s_i^{\sigma \cup \{k\}} = \sgn s_j^{\sigma \cup \{k\}}$ if and only if $\sgn s_i^{\sigma} = \sgn s_j^{\sigma}$.  Similarly, for all $\ell \in [n] \setminus (\sigma \cup\{k\})$, $\sgn s_\ell^{\sigma \cup \{k\}} = -\sgn s_i^{\sigma \cup \{k\}}$ if and only if $\sgn s_\ell^{\sigma} = -\sgn s_i^{\sigma}$.  Hence all the signs for all $i \in [n] \setminus \{k\}$ have the proper relationship for $\sigma \cup \{k\}$ to be a fixed point support precisely when $\sigma$ is a fixed point support.  Thus, the only remaining sign to check is that of $s_k^{\sigma \cup \{k\}}= s_k^{\sigma}$.  Observe that $\sigma \cup \{k\} \in \FP(G)$ precisely when, for all $i\in\sigma$, 
$$\sgn s_k^{\sigma \cup \{k\}} = \sgn s_i^{\sigma \cup \{k\}} = \sgn (-\delta s_i^{\sigma}) = - \sgn s_i^{\sigma}, $$
in which case, $\sgn s_k^{\sigma} = - \sgn s_i^{\sigma}$ and so $\sigma \in \FP(G).$  Thus the result holds and $\idx(\sigma \cup \{k\}) = - \idx(\sigma),$ as desired.
\end{proof}

As an example, consider the graph in Figure~\ref{fig:graph-series}A, which consists of a clique $\{1,2\}$ and a sink node 3.  By Rule~\ref{rule:added-sink} (sinks), $\{1,2,3\} \in \FP(G)$ if and only if $\{1,2\} \in \FP(G)$.  Since $\{1,2\}$ is a target-free clique, it survives as a fixed point by Rule~\ref{rule:clique}, and thus $\{1,2,3\} \in \FP(G)$ as well.  In contrast, in Figure~\ref{fig:graph-series}C, the addition of the sink node 4 eliminates the fixed point supports $\{1,2\}$ and $\{1,2,3\}$ by Rule~\ref{rule:uniform-in-deg} (uniform in-degree).  Thus we know that $\{1,2,4\}, \{1,2,3,4\} \notin \FP(G)$ by Rule~\ref{rule:added-sink}.  

\smallskip 

For our next rule, we consider when the added node $k$ is a projector onto $\sigma$.  In contrast to when $k$ is a sink, $\sigma \cup \{k\}$ will be a fixed point support precisely when $\sigma$ does \underline{not} survive the addition of $k$.

\begin{rules}[projectors]\label{rule:added-projector}
Let $G$ be a graph on $n$ nodes and $k \in [n]$.  If $k$ is a projector onto $[n] \setminus \{k\}$, then for any nonempty $\sigma \subseteq [n]\setminus \{k\}$, 
$$  \sigma \cup \{k\} \in \FP(G) \;\; \Leftrightarrow \;\;\sigma \not\in \FP(G) \text{ and } \sigma \in \FP(G|_{[n] \setminus \{k\}}),$$
and $\idx(\sigma \cup \{k\}) = \idx(\sigma).$
\end{rules}
\begin{proof}
Since $k$ is a projector onto $[n] \setminus \{k\}$, we have that $s_i^{\sigma \cup \{k\}} = \varepsilon s_i^{\sigma}$ for all $i \in [n] \setminus \{k\}$ (Corollary~\ref{cor:simply-added}).  By the same logic as in the proof of Rule~\ref{rule:added-sink}, all the signs for $i \in [n] \setminus \{k\}$ have the proper relationship for $\sigma \cup \{k\}$ to be a fixed point support precisely when $\sigma$ is a fixed point support in $G|_{[n] \setminus \{k\}}$. Thus, the only remaining sign to check is that of $s_k^{\sigma \cup \{k\}}= s_k^{\sigma}$. Observe that $\sigma \cup \{k\} \in \FP(G)$ precisely when $\sigma \in \FP(G|_{[n] \setminus \{k\}})$, and for all $i\in\sigma$,  
$$\sgn s_k^{\sigma \cup \{k\}} = \sgn s_i^{\sigma \cup \{k\}} = \sgn s_i^{\sigma}.$$
In this case, $\sgn s_k^{\sigma} = \sgn s_i^{\sigma}$ and so $\sigma \notin \FP(G)$. Finally, $\idx(\sigma \cup \{k\}) = \idx(\sigma).$
\end{proof}

The next rule shows that if $\sigma$ has a target, it can never support a fixed point.

\begin{rules}[targets]\label{rule:target}
Let $G$ be a graph on $n$ nodes and $\sigma \subseteq [n]$.  If there exists $k \in [n] \setminus \sigma$ such that $k$ is a target of $\sigma$, then $\sigma \notin \FP(G)$.  
\end{rules}
\begin{proof}
For every $i \in \sigma$, $k$ outside-in dominates $i$ with respect to $\sigma$, and so $\sigma \notin \FP(G)$.
\end{proof}
As an illustration of Rule~\ref{rule:target}, notice that in Figure~\ref{fig:graph-series}C, node 4 is a target of $\{1,2\}$, and thus $\{1,2\} \notin \FP(G)$.

The last two rules characterize the collection of fixed point supports when a graph contains a node that is either fully bidirectionally connected to all other nodes, or at the other extreme, is isolated and has no connections to other nodes.  

\begin{rules}[bidirectionally-connected nodes]\label{rule:bidirectionally-connected-node}
Let $G$ be a graph on $n$ nodes and $k \in [n]$.
If $k \leftrightarrow i$ for all $i \in [n] \setminus \{k\}$, then
$$\FP(G) = \{\sigma \cup \{k\}~|~ \sigma \in \FP(G|_{[n]\setminus\{k\}})\}.$$
Thus, $|\FP(G)|= |\FP(G|_{[n]\setminus\{k\}})|.$
\end{rules}
\begin{proof}
Observe that $k$ is a target of $[n] \setminus \{k\}$, and so for any $\sigma$ with $k \notin \sigma$, $\sigma \notin \FP(G)$ by Rule~\ref{rule:target}.  Since $k$ is also a projector onto $[n] \setminus \{k\}$, Rule~\ref{rule:added-projector} applies, and so $\sigma \cup \{k\} \in \FP(G)$ if and only if $\sigma \in \FP(G|_{[n]\setminus\{k\}})$ since $\sigma \notin \FP(G)$ for all such $\sigma$.  
\end{proof}

Recall that a node is \emph{isolated} if it has no incoming and no outgoing edges.

\begin{rules}[isolated nodes]\label{rule:isolated-node}
Let $G$ be a graph on $n$ nodes and $k \in [n]$.  If $k$ is an isolated node in $G$, then
$$\FP(G) = \FP(G|_{[n]\setminus\{k\}}) \cup \{\sigma \cup \{k\}~|~ \sigma \in \FP(G|_{[n]\setminus\{k\}})\} \cup \{k\}.$$
Thus, $|\FP(G)|= 2|\FP(G|_{[n]\setminus\{k\}})|+1.$
\end{rules}
\begin{proof}
Let $\sigma \in \FP(G|_{[n]\setminus\{k\}})$.  For any $i \in \sigma$, we have $i$ inside-out dominates $k$ with respect to $\sigma$, and so $\sigma \in \FP(G)$.  Since $k$ is isolated, it is also a sink, and thus by Rule~\ref{rule:added-sink}, $\sigma \cup \{k\} \in \FP(G)$ for all $\sigma \in \FP(G)$.  
Finally, $\{k\} \in \FP(G)$ by Rule~\ref{rule:indep-set} since it is trivially an independent set that is a sink.
\end{proof}

\subsection{Parameter independence in $\FP(G)$}\label{sec:parameter-independence}
All the CTLN graph rules presented thus far, including graphical domination, are parameter independent. Thus, any arguments relying on these results are guaranteed to depend solely on the graph structure.  In light of this, it is natural to ask to what extent $\FP(G)$ is parameter independent. The proof of Theorem~\ref{thm:n4-param-independent}, below, shows that graph rules are sufficient to characterize the full set $\FP(G)$ for all graphs on $n \leq 4$ nodes, and thus $\FP(G)$ is parameter independent for all such graphs.

\begin{theorem}\label{thm:n4-param-independent}
Let $G$ be a graph on $n \leq 4$ nodes.  Then $\FP(G)$ is constant across all values of $\varepsilon$ and $\delta$ in the legal range.  In other words, $\FP(G)$ is parameter independent.
\end{theorem}
\begin{proof}
Figure~\ref{fig:n3-graphs-s_i} in Appendix Section~\ref{appendixA} gives all directed graphs of size $n \leq 3$ together with the $s_i^\sigma$ values, for $\sigma = [n]$.  In particular, it can be seen that the signs of the $s_i^\sigma$ for these small graphs are all parameter independent. Applying Theorem~\ref{thm:sgn-condition} (sign conditions), we see that there are nine permitted motifs of size $|\sigma| \leq 3$.  
Equivalently, all nine permitted motifs can be identified via graph rules, and all the other motifs can be shown to be forbidden via inside-in graphical domination. The survival rules for the permitted motifs are also parameter independent: seven of them are uniform in-degree, and thus their survival is parameter independent by Rule~\ref{rule:uniform-in-deg}; while for the remaining two, either inside-out or outside-in graphical domination applies to every possible embedding of the motifs into larger graphs. It follows that all fixed point supports that are proper subsets of $[n]$ for $n \leq 4$ are parameter independent. Moreover, by Rule~\ref{rule:parity} (parity) we can always determine whether the full support $[n]$ is in $\FP(G)$ from knowledge of which proper subsets are in $\FP(G)$. We conclude that for $n \leq 4$, $\FP(G)$ is parameter independent.
\end{proof}

It turns out that among permitted motifs of size $|\sigma|=4$, there are a few that have parameter-dependent survival rules when embedded in graphs of size $n \geq 5$.  Figure~\ref{fig:param-depend-survival} shows all three permitted motifs of size 4 whose survival is parameter dependent, together with the embeddings where this dependence occurs. (For all other embeddings, the survival rules are parameter independent and can be derived from graphical domination.) 
In fact, the precise dependence on $\varepsilon$ and $\delta$ can be easily computed (see Example~\ref{ex:parameter-dependence}), and is also shown in Figure~\ref{fig:param-depend-survival}.
In Appendix Section~\ref{appendixC}, we give some example graphs on $n=5$ that have parameter-dependent $\FP(G)$ as a result of having one or more of these motifs embedded in a parameter-dependent manner. 

\begin{figure}[!h]
\begin{center}
\includegraphics[width=.8\textwidth]{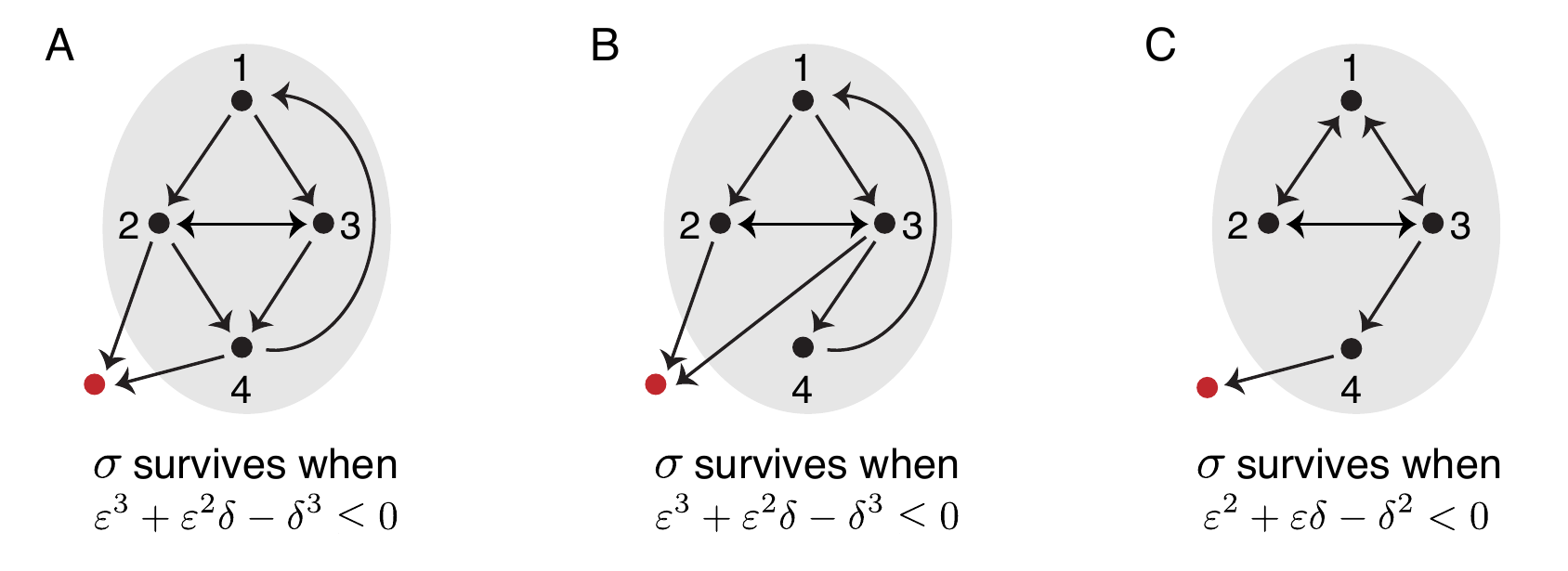}
\vspace{-.15in}
\end{center}
\caption{\textbf{Parameter-dependent survival.} The three permitted motifs with $|\sigma|=4$ that have parameter-dependent survival rules.  In each panel, the graph in the gray shaded region corresponds to $\sigma$ and the red node is an external node impacting the survival of $\sigma$.  Note that all edges from the red node back to $\sigma$ do not impact the survival of $\sigma$, and are thus omitted for clarity. In graph (A), an equivalent embedding has edges from vertices $3$ and $4$ to the red node, by symmetry. In each case, the polynomial inequality gives the precise condition under which $\sigma$ survives.}
\label{fig:param-depend-survival}
\end{figure}

\begin{example}\label{ex:parameter-dependence}
For the three graphs in Figure~\ref{fig:param-depend-survival}, we will use Theorem~\ref{thm:sgn-condition} (sign conditions) to work out the parameter dependence of the survival of these motifs.  In each case, let $\sigma = \{1,2,3,4\}$ and let node 5 be the one in red.

(A) Let $G$ be the graph in Figure~\ref{fig:param-depend-survival}A. Computing $s_i^\sigma = \det((I-W_\sigma)_i; \theta \one)$ for all $i \in \sigma$, we obtain
$$s_1^\sigma = \varepsilon\theta(\varepsilon^2+2\varepsilon\delta + 2\delta^2),\quad s_2^\sigma = s_3^\sigma = \varepsilon\theta(\varepsilon^2+\varepsilon\delta + \delta^2), \quad s_4^\sigma = \varepsilon\theta(2\varepsilon^2+3\varepsilon\delta+2\delta^2).$$ 
Since $\sgn s_i^\sigma = \sgn s_j^\sigma$ for all $i, j \in \sigma$, we see that $\sigma$ is permitted by Theorem~\ref{thm:sgn-condition} (sign conditions).  Additionally, we have
$s_5^\sigma = \varepsilon\theta(\varepsilon^3+\varepsilon^2\delta - \delta^3).$
Since $\sigma$ survives if and only if $\sgn s_5^\sigma = - \sgn s_i^\sigma$ for all $i \in \sigma$, we see that $\sigma$ survives precisely when
$\varepsilon^3+\varepsilon^2\delta - \delta^3 <0.$
\smallskip

(B) Let $G$ be the graph in Figure~\ref{fig:param-depend-survival}B.  Here we have $$s_1^\sigma = \varepsilon\delta^2\theta,\quad s_2^\sigma = s_3^\sigma = \varepsilon\theta(\varepsilon^2+\varepsilon\delta + \delta^2), \quad s_4^\sigma = \varepsilon\theta(\varepsilon^2+2\varepsilon\delta+2\delta^2),$$
and since the signs of the $s_i^\sigma$ all agree, we see that $\sigma$ is a permitted motif.  Then since 
$s_5^\sigma = \varepsilon\theta(\varepsilon^3+\varepsilon^2\delta - \delta^3),$
we again see that $\sigma$ survives precisely when
$\varepsilon^3+\varepsilon^2\delta - \delta^3 <0.$
\smallskip

(C) Let $G$ be the graph in Figure~\ref{fig:param-depend-survival}C.  For this graph, we have
$$s_1^\sigma = s_2^\sigma = s_3^\sigma = -\varepsilon^2\delta\theta, \quad s_4^\sigma = -\varepsilon^2\theta(\varepsilon+2\delta),$$
showing that $\sigma$ is a permitted motif.  Additionally, $s_5^\sigma = -\varepsilon^2\theta(\varepsilon^2+\varepsilon\delta - \delta^2),$
yielding the condition for survival of $\sigma$:
$\varepsilon^2+\varepsilon\delta - \delta^2 <0.$
\end{example}

\medskip

For certain classes of graphs, we can get parameter independence of $\FP(G)$ through at least $n=5$.  A directed graph is called \emph{oriented} if it has no bidirectional edges. (Note that all the graphs in Figure~\ref{fig:param-depend-survival} have the bidirectional edge  $2 \leftrightarrow 3$, and are thus {\it not} oriented.)

\begin{theorem}\label{thm:oriented-param-independent}
Let $G$ be an oriented graph on $n \leq 5$ nodes. Then $\FP(G)$ is parameter independent.
\end{theorem}
\begin{proof}
In Appendix Section~\ref{appendixB}, we examine all permitted motifs with $|\sigma|\leq 4$ that can arise as subgraphs of an oriented graph, and show that the corresponding survival rules are all parameter independent.  Combining this with Rule~\ref{rule:parity} (parity), we see that the full $\FP(G)$ is parameter independent for oriented graphs of size $n \leq 5$.
\end{proof}

We have already seen that parameter independence of $\FP(G)$ cannot extend to all $n=5$ graphs (see Appendix Section~\ref{appendixC} for some explicit examples).  It is possible, however, that Theorem~\ref{thm:oriented-param-independent} can be extended to oriented graphs with $n> 5$.

\section{Building block graph rules for CTLNs}\label{sec:building-blocks}

Thus far, in Sections~\ref{sec:fixed-pts-CTLNs} and~\ref{sec:graph-rules}, we have presented a variety of graph rules that enable us to determine when certain subgraphs are permitted or forbidden motifs, as well as their survival rules. In this section, we consider larger networks corresponding to {\it composite graphs} that are built from component subgraphs, and address the question: if we know $\FP(G_i)$ for each component $G_i$, what can we say about $\FP(G)$ for the full network? Our main results are collected in Section~\ref{sec:composite-graphs}. These results can be thought of as graph rules for networks that can be decomposed into smaller ``building blocks'' whose structure is better understood.  In the remaining subsections, we prove these results and provide additional details on the fixed point supports of composite graphs. 

As in Sections~\ref{sec:fixed-pts-CTLNs} and~\ref{sec:graph-rules}, we simplify notation and write $\FP(G)$ instead of $\FP(G,\varepsilon, \delta)$. Note that whenever such expressions occur in the same claim, such as $\FP(G)$ and $\FP(G_i)$, the same choices for $\varepsilon$ and $\delta$ must be assumed in each instance. All the results, however, are parameter independent.

\subsection{Composite graphs}\label{sec:composite-graphs}

\begin{definition}[composite graph] Given a set of graphs $G_1,\ldots,G_N$, and a graph $\widehat{G}$ on $N$ nodes, the {\it composite graph} with {\it components} $G_i$ and {\it skeleton} $\widehat{G}$ is the graph $G$ constructed by taking the union of all component graphs, and adding edges between components according to the following rule: if $u \in G_i$ and $v \in G_j$, then $u \to v$ in $G$ if and only if $i \to j$ in $\widehat{G}$. (See Figure~\ref{fig:general-composite}.)
\end{definition}

\begin{figure}[!ht]
\begin{center}
\includegraphics[width=.9\textwidth]{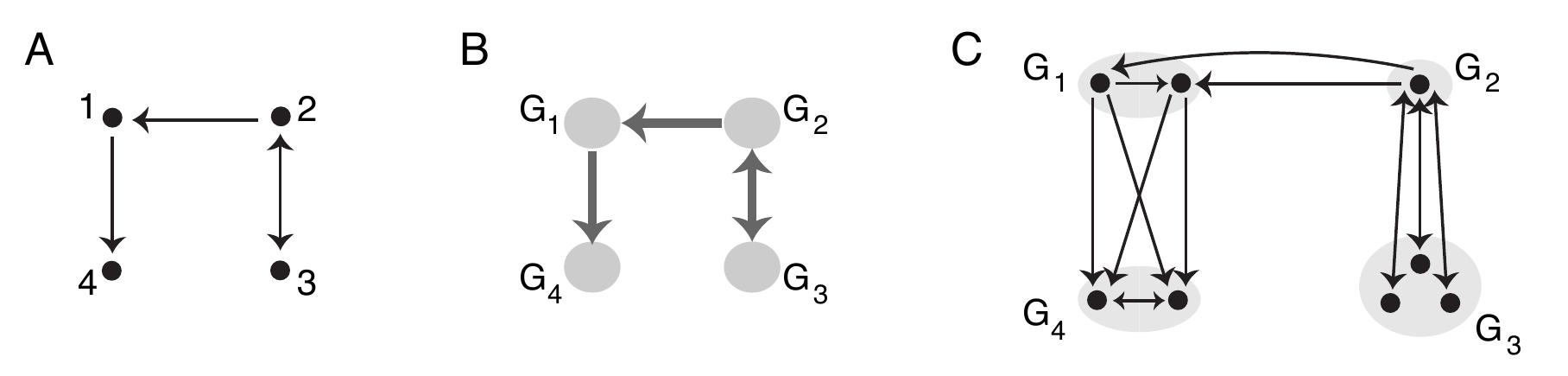}
\caption{(A) A skeleton graph $\widehat G$.  (B) An arbitrary composite graph with skeleton $\widehat G$ from A.  Each node $i$ in the skeleton is replaced with a component graph $G_i$ whose connections to the rest of the graph are prescribed by the connections of node $i$ in $\widehat G$.  (C) An example composite graph with skeleton $\widehat G$ from A.
}
\label{fig:general-composite}
\end{center}
\vspace{-.1in}
\end{figure}

Note that any graph $G$ can trivially be thought of as a composite graph, with each component a single 
vertex. The results in this section are more interesting when at least one component in the graph is larger.

The first result says that if $G|_\sigma$ is a composite graph, for some $\sigma \subseteq [n]$, and any 
component $\sigma_i$ is a forbidden motif, then $\sigma$ is also a forbidden motif. In other words, one bad apple spoils the bunch.
The proof is given in the next section. 

\begin{theorem}\label{thm:one-bad-apple}
Let $G|_\sigma$ be any composite graph with components $G|_{\sigma_1}, \ldots, G|_{\sigma_N}$.  If $\sigma_i$ is a forbidden motif for any $i \in [N]$, then $\sigma$ is also a forbidden motif.
\end{theorem}

Next we show that if $G|_\sigma$ is a composite graph that has an unstable uniform in-degree component, this is sufficient to guarantee that $\sigma$ is also unstable. 

\begin{proposition}\label{prop:composite-unstable}
Let $G|_\sigma$ be any composite graph with components $G|_{\sigma_1}, \ldots, G|_{\sigma_N}$.  If $\sigma_i$ is a uniform in-degree permitted motif that is unstable for any $i \in [N]$, then $\sigma$ is unstable (or rather $I-W_\sigma$ is unstable).
\end{proposition}
\begin{proof}
By definition of a composite graph, $\sigma$ has a simply-added split where $\sigma \setminus \sigma_i$ is simply-added to $\sigma_i$.  Since there exists a $\sigma_i$ that is uniform in-degree, we can apply Lemma~\ref{lemma:simply-added-evals} to show that all the eigenvalues of $I-W_{\sigma_i}$ are inherited to $I-W_\sigma$ except for possibly the top eigenvalue (which is the row sum of $I-W_{\sigma_i}$).  Since $\sigma_i$ is unstable, $I-W_{\sigma_i}$ has at least one negative eigenvalue, and this cannot be the top eigenvalue, since the row sum is always positive.  Thus, $I-W_\sigma$ inherits this negative eigenvalue, and so is unstable.
\end{proof}

For the following results, $G$ is always a composite graph on $n$ vertices, with skeleton $\widehat{G}$ and components $G_1,\ldots,G_N$. For any $\sigma \subseteq [n]$, we define $\sigma_i \od \sigma \cap \tau_i$, where $\tau_i$ is the set of vertices in component $G_i$. Note that $[n] = \cup_i \tau_i$, and $G_i = G|_{\tau_i}$.

Recall that a graph is {\it strongly forbidden} if it has inside-in graphical domination.  The next theorem shows that when the skeleton is strongly forbidden, then the full graph is as well.  

\begin{theorem}\label{thm:strongly-forbidden}
Let $G$ be a composite graph with skeleton $\widehat{G}$.  If $\widehat{G}$ is strongly forbidden, then $G$ is strongly forbidden.
\end{theorem}

A strongly forbidden skeleton is thus sufficient to guarantee that $G$ is strongly forbidden even if all the components $G_i$ are permitted.  Graphical domination within the skeleton is essential to this result: the following example shows that Theorem~\ref{thm:strongly-forbidden} cannot be extended to skeletons that are forbidden, but not strongly forbidden.

\begin{figure}[!ht]
\begin{center}
\includegraphics[width=\textwidth]{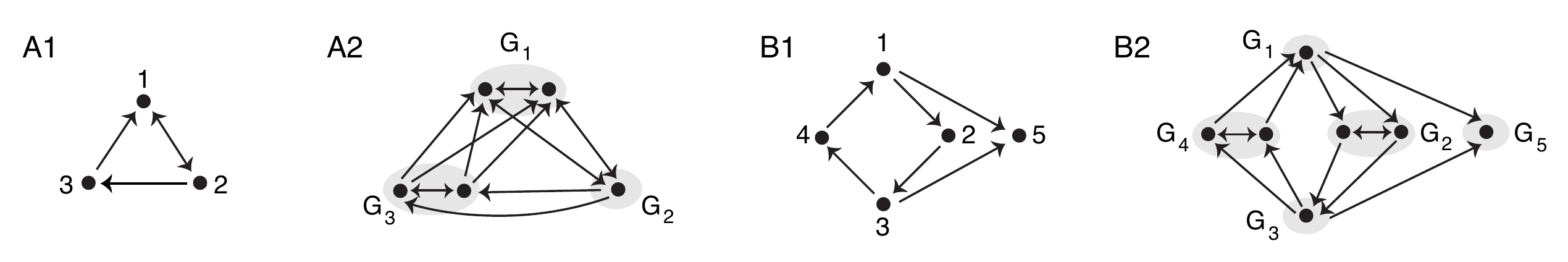}
\caption{(A1) A strongly forbidden motif, where node 1 graphically dominates node 3 with respect to $[3]$.  (A2) A composite graph whose skeleton is the graph from A1.  This graph is a strongly forbidden motif since every node in component $G_1$ graphically dominates every node in $G_3$. (B1) A forbidden motif that has no graphical domination.  (B2) A composite graph whose skeleton is the graph from B1.  This graph is a (uniform in-degree) permitted motif despite having a forbidden skeleton.}
\label{fig:counter-example-composite}
\end{center}
\vspace{-.1in}
\end{figure}

\begin{example}
Consider the graph in panel A1 of Figure~\ref{fig:counter-example-composite}.  This graph is strongly forbidden since node 1 graphically dominates node 3 with respect to the full support $\{1,2,3\}$.  The graph in panel A2 of Figure~\ref{fig:counter-example-composite} is a composite graph obtained from the skeleton in A1 by inserting 2-cliques at two of the nodes.  Theorem~\ref{thm:strongly-forbidden} guarantees that the graph in A2 is strongly forbidden since it has a strongly forbidden skeleton.  In particular, we see that every node in $G_1$ graphically dominates each node in $G_3$.  

In contrast, consider the graph in panel B1 of Figure~\ref{fig:counter-example-composite}, which has no graphical domination and thus is not strongly forbidden.  It turns out that this graph is still a forbidden motif: it has the form $\sigma \cup \{k\}$ where $\sigma=\{1,2,3,4\}$ and $k=5$ is a sink.  Since $\sigma$ has uniform in-degree $d=1$, and $5$ receives two edges from $\sigma$, we see $\sigma \notin \FP(G)$ by Rule~\ref{rule:uniform-in-deg}, and thus $\sigma \cup \{5\}\notin \FP(G)$ by Rule~\ref{rule:added-sink} (sinks).  The composite graph in panel B2 of Figure~\ref{fig:counter-example-composite} has B1 as its skeleton, again with 2-cliques inserted at two of the nodes.  However this graph is \underline{not} forbidden, despite having a forbidden skeleton, since it has uniform in-degree 2 (Rule~\ref{rule:uniform-in-deg}).  This example shows that Theorem~\ref{thm:strongly-forbidden} cannot be generalized from skeletons that are strongly forbidden to those that are merely forbidden.
\end{example}

Next we consider some special classes of composite graphs where we can guarantee that the full support is permitted precisely when all the components are permitted.

\begin{definition}
Let $G$ be a composite graph with components $G_1,\ldots, G_N$ and skeleton $\widehat G$. If $\widehat{G}$ is an independent set, then we say that $G$ is the {\it disjoint union} of the components $G_i$. If $\widehat{G}$ is a clique, then we say $G$ is the {\it clique union} of the $G_i$. Finally, if $\widehat{G}$ is a cycle, then $G$ is a {\it cyclic union} of its components. (See Figure~\ref{fig:special-composite}.)
\end{definition}

\begin{figure}[!ht]
\begin{center}
\includegraphics[width=.85\textwidth]{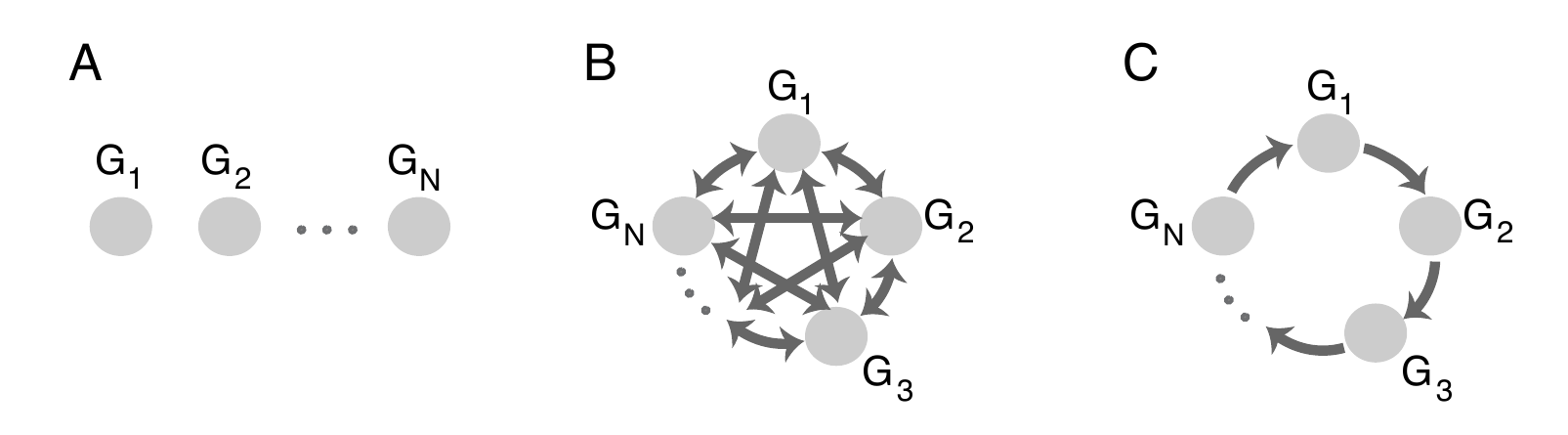}
\caption{(A) A disjoint union.  (B) A clique union.  (C) A cyclic union.
}
\label{fig:special-composite}
\end{center}
\vspace{-.1in}
\end{figure}
\FloatBarrier

\begin{theorem}\label{thm:composite-permitted}
Let $G|_\sigma$ be either a disjoint union, a clique union, or a cyclic union with nonempty components $G|_{\sigma_1}, \ldots, G|_{\sigma_N}$.  Then $\sigma$ is a permitted motif if and only if $\sigma_i$ is a permitted motif for every $i \in [N]$.  Moreover, when $\sigma$ is permitted, the index factors as:
$$\idx(\sigma) = \left\{\begin{array}{rl}  (-1)^{N-1}\displaystyle\prod_{i=1}^N \idx(\sigma_i), & \text{ for disjoint unions}, \vspace{.1in}\\  \displaystyle\prod_{i=1}^N \idx(\sigma_i), & \text{ for clique or cyclic unions.}\end{array}\right.$$
\end{theorem}

Note that the expressions for the indices, $\idx(\sigma)$, are identical or nearly identical between clique unions, cyclic unions, and disjoint unions. A simple rule of thumb is that the index of $\sigma$ in these unions is given by the product of the indices of all components $\sigma_i$ with the index of the skeleton. In the case of clique unions and cyclic unions, the skeletons have a unique (full support) fixed point, whose index is thus $+1$. For disjoint unions, the skeleton is an independent set of size $N$ with index $(-1)^{N-1}$ (see Rule~\ref{rule:indep-set}).

Disjoint unions, clique unions, and cyclic unions are all examples of composite graphs whose skeletons are permitted motifs. It is thus natural ask whether Theorem~\ref{thm:composite-permitted} holds for other composite graphs with permitted skeletons. In fact, we conjecture that this theorem holds more generally.

\begin{conjecture}\label{conjecture}
Let $G$ be a composite graph whose skeleton $\widehat G$ is a permitted motif. Then $G$ is a permitted motif if and only if every component $G_i$ is permitted.
\end{conjecture}

Note that the forward direction is a direct consequence of Theorem~\ref{thm:one-bad-apple}. It is also worth noting that even if true, this conjecture does not cover all permitted motifs that are composite graphs.
For example, the graph in panel B2 of Figure~\ref{fig:counter-example-composite} is a permitted motif, but its skeleton is forbidden.

The proof of Theorem~\ref{thm:composite-permitted} is given in Section~\ref{sec:composite-permitted-proof}. There we provide separate proofs for the disjoint, clique, and cyclic unions, and also work out the full sets $\FP(G)$ in the case where all of $G$ has one of these special composite structures. Furthermore, these results imply that $\FP(G)$ is parameter independent whenever all the $\FP(G_i)$ are parameter independent (see Corollary~\ref{cor:param-independent}).  In Section~\ref{sec:composite-survival}, we consider survival rules for disjoint unions and clique unions that are embedded in larger graphs. 

But first, in the next section, we present some lemmas which help us to prove the theorems on composite graphs stated above. There we also prove Theorems~\ref{thm:one-bad-apple} and~\ref{thm:strongly-forbidden}.

\subsection{Some lemmas for composite graphs}

\begin{lemma}\label{lemma:composite-signs}
Let $G$ be a composite graph, and consider $\sigma \subseteq [n]$. For any component $G_i$ such that $\sigma_i = \sigma \cap \tau_i \neq \emptyset$,
$$\sgn s_j^{\sigma} = \sgn s_k^{\sigma}\; \Leftrightarrow \; \sgn s_j^{\sigma_i} = \sgn s_k^{\sigma_i}, \;\; \text{for all}\;\; j,k \in \tau_i.$$
\end{lemma}

\begin{proof}
Observe that $[n] \setminus \tau_i$ is simply-added to $\tau_i$ for each $i \in [N]$, by the nature of the composite graph construction, and thus $\sigma \setminus \sigma_i$ is simply-added to $\tau_i$.  By Theorem~\ref{thm:simply-added}, $s_j^\sigma = \alpha s_j^{\sigma_i}$, where $\alpha = \frac{1}{\theta} s_j^{\sigma \setminus \sigma_i}$  is identical for all $j\in \tau_i$.  
\end{proof}

With this lemma, we can immediately prove Theorem~\ref{thm:one-bad-apple}.

\begin{proof}[Proof of Theorem~\ref{thm:one-bad-apple}] Suppose $G|_\sigma$ is a composite graph with components $G|_{\sigma_1}, \ldots, G|_{\sigma_N}$, and one of the $\sigma_i$ is a forbidden motif. Then, by Theorem~\ref{thm:sgn-condition} (sign conditions), there exists $j,k \in \sigma_i$ such that $\sgn s_j^{\sigma_i} \neq \sgn s_k^{\sigma_i}$, and therefore by Lemma~\ref{lemma:composite-signs} we have
$\sgn s_j^{\sigma} \neq \sgn s_k^{\sigma}$. Thus, $\sigma$ is a forbidden motif.
\end{proof}

\begin{lemma}\label{lemma:composite1}
If $\sigma \in \FP(G)$, then $\sigma_i \in \FP(G_i)$ for each $i \in [N]$ such that $\sigma_i \neq \emptyset$.
\end{lemma}

\begin{proof}
If $\sigma \in \FP(G)$, then for any $j, k \in \sigma_i$ and $\ell \in \tau_i \setminus \sigma_i$, we have $\sgn s_j^\sigma = \sgn s_k^\sigma = - \sgn s_\ell^\sigma$ by Theorem~\ref{thm:sgn-condition} (sign conditions).  By Lemma~\ref{lemma:composite-signs}, we infer that $\sgn s_j^{\sigma_i} = \sgn s_k^{\sigma_i} = - \sgn s_\ell^{\sigma_i}$, and so $\sigma_i$ satisfies the sign conditions within $G_i$.  Thus, $\sigma_i \in \FP(G_i).$
\end{proof}

As an immediate consequence of Lemma~\ref{lemma:composite1}, we see that if any component $\sigma_i$ of $\sigma$ is forbidden (or non-surviving) in $G_i$, then $\sigma$ is forbidden (or non-surviving) in $G$.  The next lemma shows that if $\sigma$ is a permitted motif, then the survival of each of its components in $G$ is sufficient to guarantee the survival of $\sigma$.  

\begin{lemma}\label{lemma:composite2}
If $\sigma \in \FP(G|_\sigma)$ and $\sigma_i \in \FP(G_i)$ for each $i \in [N]$, then $\sigma \in \FP(G)$.
\end{lemma}

\begin{proof}
First, note that $\sigma_i \in \FP(G_i)$ implies $\sigma_i \neq \emptyset$ for each $i \in [N]$, and $\sgn s_j^{\sigma_i} = \sgn s_k^{\sigma_i} = -\sgn s_{\ell}^{\sigma_i}$ for all $j,k \in \sigma_i$ and $\ell \in \tau_i \setminus \sigma_i$. On the other hand,
$\sigma \in \FP(G|_\sigma)$ implies $\sgn s_j^\sigma = \sgn s_k^\sigma$ for all $j,k \in \sigma$. Now applying Lemma~\ref{lemma:composite-signs}, we see that $\sgn s_\ell^\sigma = - \sgn s_j^\sigma$ for $\ell \in [n]\setminus \sigma$ and $j \in \sigma$. It follows that the sign conditions are satisfied, and so $\sigma \in \FP(G)$.
\end{proof}

The next lemma tells us that graphical domination in the skeleton graph $\widehat{G}$ gets inherited to the composite graph $G$.

\begin{lemma}\label{lemma:composite-domination}
Let $G$ be a composite graph with skeleton $\widehat{G}$. If $\hat{k}$ graphically dominates $\hat \jmath$ with respect to $\widehat\sigma$ in $\widehat{G}$, then for any $k \in G_{\hat{k}}$ and any $j \in G_{\hat\jmath}$, we have that $k$ graphically dominates~$j$ with respect to $\sigma = \displaystyle \bigcup_{i \in \widehat\sigma} G_i$.
\end{lemma}

\begin{proof}
Let $[n]$ denote the vertices of $G$, and $[N]$ the vertices of $\widehat{G}$.  Consider $\widehat \sigma \subseteq [N]$ and suppose there exists $\hat \jmath, \hat k \in [N]$ such that $\hat k$ graphically dominates $\hat \jmath$ with respect to $\widehat\sigma$.  By definition of graphical domination, this means (1) for any $\hat \imath \in \widehat\sigma$, $\hat \imath \to \hat \jmath$ implies $\hat \imath \to \hat k$; (2) if $\hat\jmath \in \widehat \sigma$, then $\hat \jmath \to \hat k$; and (3) if $\hat k \in \widehat \sigma$, then $\hat k \not\to \hat \jmath$.  We will show that in $G$, any $j \in G_{\hat \jmath}$ and $k \in G_{\hat k}$ satisfy these same properties, and so $k$ graphically dominates $j$ with respect to $\sigma$.  

By definition of the composite graph, $\hat \jmath \to \hat k$ and $\hat k \not\to \hat \jmath$ in the skeleton immediately imply $j \to k$ and $k \not\to j$ in $G|_\sigma$, and so satisfying conditions (2) and (3) in the skeleton guarantees they are satisfied in $G$.  Now we turn to condition (1).  Consider $i \in \sigma$ such that $i \to j$ in $G|_\sigma$.  If $i$ and $j$ are both in component $\sigma_{\hat \jmath}$, then $j \in \sigma$ and so $\hat\jmath \in \widehat \sigma$.  Thus condition (2) holds and $\hat \jmath \to \hat k$ in $\widehat G|_{\widehat \sigma}$, which implies that $i \to k$ in $G|_\sigma$.  If they are in different components, i.e.\ $i \in \sigma_{\hat\imath}$ for some $\hat\imath \neq \hat \jmath$, then by definition of a composite graph, $i \to j$ implies $\hat \imath \to \hat \jmath$ in $\widehat G|_{\widehat \sigma}$, and so by (1), we must have $\hat \imath \to \hat k$ and thus $i \to k$ in $G|_\sigma$.  Thus, for any $i \in \sigma$, $i \to j$ implies $i \to k$ and condition (1) is satisfied.  Hence $k$ graphically dominates $j$ with respect to $\sigma$, and so $G|_\sigma$ is strongly forbidden.
\end{proof}

As an immediate consequence, we obtain Theorem~\ref{thm:strongly-forbidden}, since inside-in graphical domination within any graph implies that it is strongly forbidden.

\subsection{Proof of Theorem~\ref{thm:composite-permitted} (disjoint, clique, and cyclic unions)}\label{sec:composite-permitted-proof}

To prove Theorem~\ref{thm:composite-permitted}, we separately consider composite graphs that are disjoint unions, clique unions, and cyclic unions.  We characterize the full $\FP(G)$ of these graphs.  In the next section, we provide partial survival rules when these motifs are embedded in larger graphs. 

As before, in this section we assume $G$ is a composite graph on $n$ vertices, with skeleton $\widehat{G}$ and components $G_1,\ldots,G_N$. For any $\sigma \subseteq [n]$, we denote $\sigma_i = \sigma \cap \tau_i$, where $\tau_i$ is the set of vertices in $G_i$. Note that $[n] = \cup_i \tau_i$, and $G_i = G|_{\tau_i}$.

\begin{theorem}[disjoint union]\label{thm:disjoint-unions}
 Let $G$ be a disjoint union of components $G_1,\ldots,G_N$. For any nonempty $\sigma \subseteq [n]$, 
$$\sigma \in \FP(G) \quad \Leftrightarrow \quad \sigma_i  \in \FP(G_i) \cup \emptyset ~~\text{ for all } i \in [N].$$
Moreover, for any $\sigma \in \FP(G)$, the index is given by
$$\idx(\sigma) = (-1)^{|\widehat \sigma|-1}\prod_{i\in \widehat \sigma} \idx(\sigma_i), \;\; \text{where}\;\;  \widehat \sigma\od \{i \in [N]~|~ \sigma_i \neq \emptyset\}.$$
\end{theorem}

In particular, Theorem~\ref{thm:disjoint-unions} tells us that
$$\FP(G) = \{\sigma \subseteq [n] \mid \sigma_i \in  \FP(G_i) \cup \emptyset \text{ for each } i \in [N]\},$$
and so,
$$|\FP(G)| = \prod_{i=1}^N (|\FP(G_i)|+1) -1.$$
For example, if $|\FP(G_i)| = 1$ for each $i \in [N]$, as in the case of the independent set (where every component is a single node), then $|\FP(G)| = 2^N-1.$

\begin{proof}
$(\Rightarrow)$ This follows immediately from Lemma~\ref{lemma:composite1}.

\noindent $(\Leftarrow)$  To simplify notation, we fix $\theta = 1$.\footnote{Note that $\FP(G)$ never depends on the value of $\theta$, provided $\theta>0$, so this choice cannot affect the results.} Suppose $\sigma$ has the property that $\sigma_i \in \FP(G_i)$ for each $i \in \widehat{\sigma}$, where $\widehat \sigma\od \{i \in [N]~|~ \sigma_i \neq \emptyset\}$. Consider $j \in \sigma_i$, and observe that for any $k \neq i$, $\sigma_k$ is simply-added to $\sigma_i$. Using Theorem~\ref{thm:simply-added}, we thus obtain
$s_j^{\sigma_i \cup \sigma_k} = s_j^{\sigma_i}s_j^{\sigma_k}.$
Moreover, for any other $\ell \in \widehat{\sigma}\setminus\{i,k\}$, we also have that $\sigma_\ell$ is simply-added to $\sigma_i \cup \sigma_k$, and so
$$s_j^{\sigma_i \cup \sigma_k \cup \sigma_\ell } = s_j^{\sigma_i}s_j^{\sigma_k}s_j^{\sigma_\ell}.$$
Continuing in this fashion yields
$$s_j^\sigma = \prod_{\ell \in \widehat \sigma} s_j^{\sigma_\ell}.$$
Furthermore, the above formula also holds for $j \notin \sigma$. This is because $j \in \tau_k$ for some $k$, and all other (nonempty) $\sigma_i$ are simply-added to $\tau_k$, not just to $\sigma_k$.

Note that for each $\sigma_i \in \FP(G_i)$, $\sigma_i$ survives as a fixed point in $\FP(G)$ by inside-out domination. Therefore, for any $j \in [n]$, $\sgn s_j^{\sigma_i} = \idx(\sigma_i)$ if and only if $j \in \sigma_i$.
For $j \in \sigma$, we have $j \in \sigma_i$ for exactly one $i \in \widehat{\sigma}$, and thus
$$\sgn s_j^\sigma = \sgn s_j^{\sigma_i} \prod_{\ell \in \widehat \sigma \setminus i} \sgn s_j^{\sigma_\ell} = \idx(\sigma_i) \prod_{\ell \in \widehat \sigma \setminus i} (-\idx(\sigma_\ell)) = (-1)^{|\widehat \sigma|-1} \prod_{\ell \in \widehat \sigma} \idx(\sigma_\ell).$$
Since the sign of $s_j^\sigma$ is the same for all $j\in\sigma$, we see $\sigma$ is permitted and the index $\idx(\sigma)$ must match the sign of $s_j^\sigma$ (by Theorem~\ref{thm:sgn-condition}).
On the other hand, for $j \notin \sigma$, 
$$\sgn s_j^\sigma = \prod_{\ell \in \widehat \sigma} \sgn s_j^{\sigma_\ell} =  \prod_{\ell \in \widehat \sigma} (- \idx(\sigma_\ell))= (-1)^{|\widehat \sigma|} \prod_{\ell \in \widehat \sigma} \idx(\sigma_\ell) = -\idx(\sigma).$$
By Theorem~\ref{thm:sgn-condition}, we conclude that $\sigma \in \FP(G)$ with  $\idx(\sigma) = (-1)^{|\widehat \sigma|-1} \prod_{\ell \in \widehat \sigma} \idx(\sigma_\ell)$. 
\end{proof}

The next theorem characterizes the full $\FP(G)$ whenever $G$ is a clique union.

\begin{theorem}[clique union]\label{thm:clique-unions}
 Let $G$ be a graph on $n$ vertices that is a clique union of $N$ subgraphs $G_1, \ldots, G_N$.  For any $\sigma \subseteq [n]$,
$$\sigma \in \FP(G) \quad \Leftrightarrow \quad \sigma_i \in \FP(G_i)~~\text{ for all } i \in [N].$$
Moreover, if $\sigma \in \FP(G)$ then $\idx(\sigma) = \prod_{i=1}^N \idx(\sigma_i)$.
\end{theorem}

Note that the theorem implies that the total number of fixed points satisfies:
$$|\FP(G)| = \prod_{i=1}^N |\FP(G_i)|.$$
In particular, if each component $G_i$ has a unique fixed point, then $G$ has a unique fixed point.

\begin{proof} $(\Rightarrow)$ Suppose $\sigma \in \FP(G)$. Lemma~\ref{lemma:composite1} guarantees that for any $i \in [N]$, $\sigma_i \in \FP(G_i)$ or $\sigma_i = \emptyset$.  If we assume there exists an $i \in [N]$ such that $\sigma_i = \emptyset$, then any $k \in \tau_i$ is a target of $\sigma$, and thus by Rule~\ref{rule:target} we cannot have $\sigma \in \FP(G),$ a contradiction. We thus conclude that $\sigma_i \neq \emptyset$, and so $\sigma_i \in \FP(G_i),$ for all $i \in [N]$.

$(\Leftarrow)$ To simplify notation, we fix $\theta = 1$.   Suppose $\sigma_i \in \FP(G_i)$ for all $i \in [N]$. In particular, each $\sigma_i$ is nonempty.
As with disjoint unions, clique unions have the property that each $\sigma_\ell$ is simply-added to any union of the other components. We thus have the same formula as in the proof of Theorem~\ref{thm:disjoint-unions}, where
for any $j \in [n]$ and $\sigma \subseteq [n]$ the value $s_j^\sigma$ factors as:
$$s_j^\sigma = \prod_{\ell \in [N]} s_j^{\sigma_\ell}.$$
Note that here the product is necessarily over all $\ell \in [N],$ since each $\sigma_\ell$ is nonempty.  Moreover, for any $\ell \in [N]$ such that $j \notin \tau_\ell$,  $j$ is a target of $\sigma_\ell,$ and so $\sigma_\ell \notin \FP(G|_{\sigma_\ell \cup j})$. This implies, by Theorem~\ref{thm:sgn-condition}, that $\sgn s_j^{\sigma_\ell} = \idx(\sigma_\ell)$ for all $j \notin \tau_\ell$.

To show that $\sigma \in \FP(G)$, we must show that $\sgn s_j^\sigma = \idx(\sigma)$ for each $j \in \sigma$, and has opposite sign if $j \notin \sigma$ (see Theorem~\ref{thm:sgn-condition}). Consider $j \in \sigma$, and observe that $j \in \sigma_i$ for some $i$. Clearly, $\sgn s_j^{\sigma_i} = \idx(\sigma_i)$. By the argument above, we also have $\sgn s_j^{\sigma_\ell} = \idx(\sigma_\ell)$ for each $\ell \in [N]\setminus i$.
Therefore, using the above product formula for $s_j^\sigma$, 
 we have $\sgn s_j^\sigma = \prod_{\ell=1}^N \idx(\sigma_\ell)$.  Since the sign matches for all $j \in \sigma$, we have that $\sigma$ is permitted and also that $\idx(\sigma) = \prod_{\ell=1}^N \idx(\sigma_\ell),$ as desired. 
 
 On the other hand, for any $j \notin \sigma$, there exists an $i$ such that $j \in \tau_i \setminus \sigma_i$.  Since $\sigma_i \in \FP(G_i)$, $\sgn s_j^{\sigma_i} = - \idx(\sigma_i)$, while for all other $\ell \in [N]\setminus i$, we have $j \notin \tau_\ell$ and so $\sgn s_j^{\sigma_\ell} = \idx(\sigma_\ell)$.  Thus, 
$$\sgn s_j^\sigma = -\idx(\sigma_i) \prod_{\ell \in [N] \setminus i} \idx(\sigma_\ell) = - \prod_{\ell=1}^N \idx(\sigma_\ell)= - \idx(\sigma).$$
By Theorem~\ref{thm:sgn-condition}, we conclude that $\sigma \in \FP(G)$.
\end{proof}

We conclude this section with Theorem~\ref{thm:cyclic-union} and its proof, characterizing $\FP(G)$ when $G$ is a cyclic union.

\begin{theorem}[cyclic union]\label{thm:cyclic-union}
Let $G$ be a cyclic union of components $G_{1},\ldots,G_{N}.$ Then
$$\sigma \in \FP(G) \;\; \Leftrightarrow \;\; \sigma_i \in \FP(G_{i}) \;\; \text{for all}\; i \in [N].$$
Moreover, if $\sigma \in \FP(G)$, then $\idx(\sigma) = \prod_{i=1}^N \idx(\sigma_i)$.
\end{theorem}
Note that, just as with the clique union, the theorem implies that for cyclic unions:
$$|\FP(G)| = \prod_{i=1}^N |\FP(G_i)|.$$
To prove Theorem~\ref{thm:cyclic-union}, we'll need the following lemmas.

\begin{lemma}\label{lemma:cyclic-union}
Let $G$ be a cyclic union of components $G_{1},\ldots,G_{N}$.   If $\sigma \in \FP(G)$, then $\sigma_i \neq \emptyset$ for all $i \in [N]$.
\end{lemma}
\begin{proof}
Let $\sigma \in \FP(G)$ and $\widehat{\sigma} = \{i \in [N]~|~ \sigma_i \neq \emptyset\}$.  Suppose for the sake of contradiction that $\widehat \sigma \neq [N]$.  Then the skeleton $\widehat G|_{\widehat\sigma}$ of $\sigma$ is a proper subset of a cycle, and so the skeleton must either contain a proper source or be an independent set.  If $\widehat G|_{\widehat\sigma}$ contains a proper source, then $\widehat \sigma$ has inside-in graphical domination, and so by Lemma~\ref{lemma:composite-domination}, that same domination lifts to inside-in graphical domination in $\sigma$.  On the other hand, if $\widehat \sigma$ is an independent set, then there is outside-in graphical domination of $\widehat \sigma$ by any external node that receives an edge from $\widehat \sigma$.  Thus again by Lemma~\ref{lemma:composite-domination} that domination relationship lifts to outside-in graphical domination of $\sigma$.  But then $\sigma \notin \FP(G)$ by Theorem~\ref{thm:graph-domination}, yielding a contradiction.  Thus, we must have $\widehat \sigma = [N]$, i.e. $\sigma_i \neq \emptyset$ for all $i \in [N]$.
\end{proof}

The next lemma is an immediate consequence of Lemmas~\ref{lemma:composite1} and~\ref{lemma:composite2}.
\begin{lemma}\label{lemma:component-survival-iff}
Let $\sigma \subseteq [n]$ have nontrivial overlap with each component of $G$, so that $\sigma_i \neq \emptyset$ for each $i \in [N]$. Then 
$$\sigma \in \FP(G) \Leftrightarrow \sigma \in \FP(G|_\sigma) \text{ and } \sigma_i \in \FP(G_i) \text{ for each } i \in [N].$$
\end{lemma}

We are now ready to prove the cyclic union theorem.

\begin{proof}[Proof of Theorem~\ref{thm:cyclic-union} (cyclic union)]
We prove this by complete induction on $n$, the total number of vertices of $G$, and for fixed number of components $N$. The base case is $n = N$, so that $G$ is an $n$-cycle. In this case there is a unique fixed point of full support (by the remarks following Rule~\ref{rule:cycle}), and every component is a single vertex with index $1$. The result trivially holds. 

Now assume the inductive hypothesis that the theorem statement holds for all cyclic unions with $N$ components and $m$ vertices, where $N \leq m < n$. By Lemma~\ref{lemma:cyclic-union} and Lemma~\ref{lemma:component-survival-iff}, 
$$\sigma \in \FP(G) \Leftrightarrow \sigma \in \FP(G|_\sigma) \text{ and } \sigma_i \in \FP(G_{i}) \text{ for all } i \in [N].$$
So to prove the first theorem statement for a graph of $n$ vertices, it suffices to show that if $\sigma_i \in \FP(G_{i})$ for all $i \in [N],$ then $\sigma \in \FP(G|_\sigma)$. 

Suppose $|\sigma| < n$, and suppose $\sigma_i \in \FP(G_{i})$ for all $i \in [N]$. Recall that if $\sigma_i \in \FP(G_{i})$, then $\sigma_i \in \FP(G|_{\sigma_i})$. Applying the inductive hypothesis to $G|_\sigma$, we see that this in turn implies $\sigma \in \FP(G|_\sigma)$, as desired. Moreover, the index formula holds because $\idx(\sigma)$ is the same as what it was in the smaller graph $G|_\sigma$, and is thus given immediately by the inductive hypothesis.

Now assume $|\sigma| = n$. This means $G|_\sigma = G$, and $\sigma_i=\tau_i$ for each $i \in [N]$.
By Lemma~\ref{lemma:component-survival-iff}, $\sigma \in \FP(G)$ implies $\sigma_i \in \FP(G_{i})$ for all $i \in [N]$. So what remains is to show the converse direction, and that the index formula holds. Suppose $\tau_i \in \FP(G_{i})$ for each $i \in [N]$. If $\sigma = [n] \notin \FP(G)$, then $|\FP(G)| = \prod_{i \in [N]}|\FP(G_{i})| -1$, since all the smaller elements of $\FP(G)$ are indeed given by picking a fixed point support from each component graph $G_{i}$. By parity, each $|\FP(G_{i})|$ is odd, and hence the product of these terms is odd. It follows that $|\FP(G)|$ is even, contradicting parity for $G$. We conclude, then, that we must have $\sigma \in \FP(G)$.

Finally, we show that the index formula holds for $\sigma = [n]$ (by assumption, it holds for $|\sigma|<n$). Using Theorem~\ref{thm:parity} (the index theorem) for $\FP(G)$, we compute
\begin{eqnarray*}
1 &=& \sum_{\sigma \in \FP(G)} \idx(\sigma) = \sum_{\sigma \in \FP(G)\setminus [n]} \idx(\sigma) + \idx([n])\\
&=& \sum_{\sigma \in \FP(G)\setminus [n]} \left(\prod_{i \in [N]} \idx(\sigma_i)\right) + \idx([n])\\
&=&  \prod_{i \in [N]}\left(\sum_{\omega \in \FP(G_{i})} \idx(\omega)\right) - \prod_{i \in [N]} \idx(\tau_i) + \idx([n])\\
&=& 1- \prod_{i \in [N]} \idx(\tau_i) + \idx([n]).
\end{eqnarray*}
Note that in the last equality we have used the index theorem to see that the indices sum to $1$ for each $\FP(G_{i})$. It follows that $\idx([n]) = \prod_{i \in [N]} \idx(\tau_i)$, as desired.
\end{proof}

It is worth thinking about what properties of cycles made the induction argument in Theorem~\ref{thm:cyclic-union} work. Namely, we needed $\widehat G$ to be a permitted motif and for every proper subset of its vertices to be either strongly forbidden or guaranteed not to survive as a fixed point support due to outside-in graphical domination, as in Lemma~\ref{lemma:cyclic-union}. In other words, $\FP(\widehat G)$ must have a unique fixed point that has full support, with all proper subsets excluded via graphical domination. Note that cliques also have this property, so Theorem~\ref{thm:clique-unions} could have been proven with a similar inductive argument. 

The following lemma provides an example of how the inductive argument can generalize beyond clique unions and cyclic unions. 

\begin{figure}[!ht]
\begin{center}
\includegraphics[width=3in]{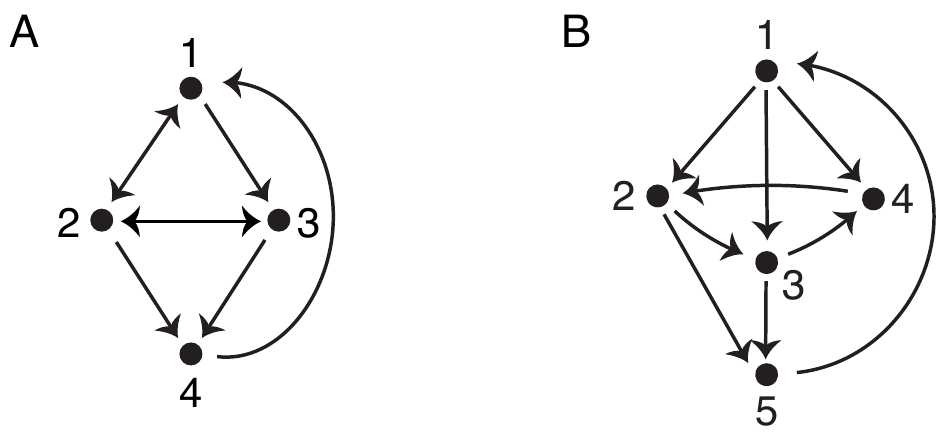}
\caption{The two skeletons $\widehat G$ considered in Lemma~\ref{lemma:N4-skeleton}.}
\label{fig:N4-skeleton}
\end{center}
\vspace{-.15in}
\end{figure}

\begin{lemma}\label{lemma:N4-skeleton}
Let $G$ be a composite graph with one of the skeletons $\widehat G$ depicted in Figure~\ref{fig:N4-skeleton}, with $N = 4$ or $N = 5$. Then 
$$\sigma \in \FP(G) \;\; \Leftrightarrow \;\; \sigma_i \in \FP(G_{i}) \;\; \text{for all } i \in [N].$$
Moreover, if $\sigma \in \FP(G)$, then $\idx(\sigma) = \prod_{i=1}^N \idx(\sigma_i)$.
\end{lemma}

\begin{proof}[Proof sketch]
First, consider the $N=4$ skeleton in Figure~\ref{fig:N4-skeleton}A. In this case, $\widehat G$ has uniform in-degree $2$, and is thus a permitted motif. It is straightforward to check that all proper subsets of $\widehat G$ are either strongly forbidden or do not survive as fixed point supports due to outside-in graphical domination.  Thus, an analogue of Lemma~\ref{lemma:cyclic-union} holds, and Lemma~\ref{lemma:component-survival-iff} applies. The induction proof of Theorem~\ref{thm:cyclic-union} can therefore be easily adapted to this case, yielding the desired result.

For the $N=5$ skeleton (Figure~\ref{fig:N4-skeleton}B), we can again check that all proper subsets of $\widehat G$ are either strongly forbidden or do not survive as fixed point supports due to outside-in graphical domination. By Rule~\ref{rule:parity} (parity), the full graph must be a permitted motif. The remaining arguments then follow exactly as in the previous case.
\end{proof}

We are now ready to prove Theorem~\ref{thm:composite-permitted}.
Putting together Theorems~\ref{thm:disjoint-unions}, \ref{thm:clique-unions} and \ref{thm:cyclic-union}, we immediately obtain the proof.

\begin{proof}[Proof of Theorem~\ref{thm:composite-permitted}]
To see the first statement, apply Theorems~\ref{thm:disjoint-unions}, \ref{thm:clique-unions} and \ref{thm:cyclic-union} for $G = G|_\sigma$ a disjoint union, clique union, and cyclic union, respectively. Since each $\sigma_i$ is nonempty, it follows in every case that $\sigma \in \FP(G|_\sigma)$ if and only if $\sigma_i \in \FP(G|_{\sigma_i})$ for all $i \in [N]$. Hence, $\sigma$ is a permitted motif if and only if each $\sigma_i$ is a permitted motif. The index formulas also follow directly from Theorems~\ref{thm:disjoint-unions}, \ref{thm:clique-unions} and \ref{thm:cyclic-union}; note that in the case of the disjoint union, $\widehat{\sigma} = [N]$ and $|\widehat{\sigma}| = N$ because we assume here that each $\sigma_i$ is nonempty.
\end{proof}

Observe that Theorems~\ref{thm:disjoint-unions}, \ref{thm:clique-unions} and \ref{thm:cyclic-union} show that for any disjoint, clique, or cyclic union, $\FP(G)$ can only depend on the parameters $\varepsilon$ and $\delta$ via the $\FP(G_i)$.  Thus, if $\FP(G_i)$ is parameter independent for all $i$, then $\FP(G)$ is also parameter independent.  

\begin{corollary}\label{cor:param-independent}
 Let $G$ be a disjoint, clique, or cyclic union of components $G_1,\ldots,G_N$ with vertex sets $\tau_1, \ldots, \tau_N$.  Then
 $$\FP(G) \text{ is parameter independent}\; \;  \Leftrightarrow \; \; \FP(G_i) \text{ is parameter independent for all }i \in [N].$$
 In particular, if $|\tau_i| \leq 4$ for all $i \in [N]$, then $\FP(G)$ is parameter independent.
\end{corollary}

Note that the last statement follows from Theorem~\ref{thm:n4-param-independent}.

\subsection{Survival rules for disjoint unions and clique unions}\label{sec:composite-survival}

Next we consider disjoint, clique, and cyclic unions that are embedded in a larger graph, $G$, which is not assumed to have any special composite structure. Proposition~\ref{prop:disjoint-union-survival} shows that a {\it necessary} condition for a disjoint union $\sigma$ to survive in the larger graph is that {\it every} component $\sigma_i$ is itself a surviving fixed point support.  In contrast, Proposition~\ref{prop:clique-union-survival} shows that a {\it sufficient} condition for a clique union fixed point to survive in $G$ is that at least {\it one} component $\sigma_i$ survives the addition of the other vertices in $G$.
The proofs of these results rely on general domination, and so we save them for Section~\ref{sec:domination-simply-added}.

\begin{proposition}[survival of disjoint union]\label{prop:disjoint-union-survival}
Let $G|_\sigma$ be a subgraph of $G$ that is a disjoint union of components $G|_{\sigma_1}, \ldots, G|_{\sigma_N}$.  If $\sigma_i \notin \FP(G)$ for any $i \in [N]$, then $\sigma \notin \FP(G)$. 
\end{proposition}

\begin{proposition}[survival of clique union]\label{prop:clique-union-survival}
Let $\sigma$ be a permitted motif that is a clique union of components
 $G|_{\sigma_1}, \ldots, G|_{\sigma_N}$ inside a larger graph $G$.  If for each $k \in [n] \setminus \sigma$, there
exists an $i \in [N]$ such that $\sigma_{i} \in \FP(G|_{\sigma_{i} \cup k})$, then $\sigma \in \FP(G)$. 
\end{proposition}

\begin{figure}[!ht]
\begin{center}
\includegraphics[width=.7\textwidth]{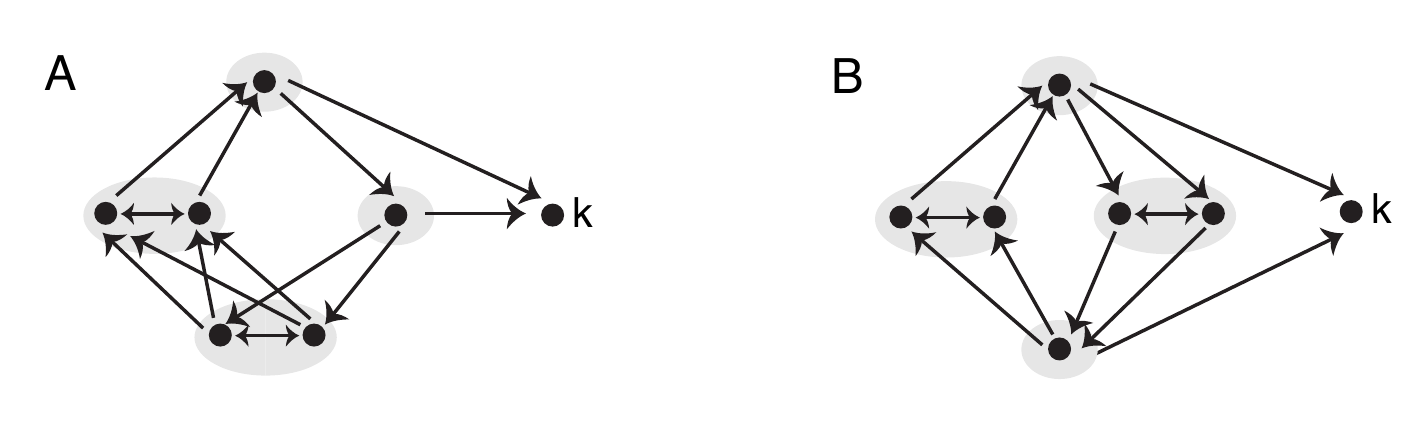}
\caption{(A) A cyclic union that does not survive the addition of node $k$, by graphical domination. (B) A cyclic union that does survive the addition of node $k$, by uniform in-degree.
}
\label{fig:cyclic-union-survival}
\end{center}
\vspace{-.1in}
\end{figure}

The previous two propositions gave simple survival conditions for the composite graph solely in terms of the survival of the individual components.  Unfortunately there is no analogous result for cyclic unions because the order of the surviving/non-surviving components within the cyclic union can impact the survival of the composite graph.  Specifically consider the graphs in Figure~\ref{fig:cyclic-union-survival}, which are cyclic unions with outgoing edges to an external node $k$.  For both of these graphs, the components of the cyclic union are two 2-cliques and two singletons, and in both, the singleton components are permitted motifs that do not survive in the full graph since they have outgoing edges to node $k$.  The only difference between the two graphs is the order in which those components were inserted in the cyclic union.  However, the cyclic union does \underline{not} survive in the full graph in A, while it does survive in the full graph in B.  To see this, note that in graph A, node $k$ outside-in graphically dominates the node to its left, and thus the cyclic union does not survive.  In B, the cyclic union has uniform in-degree 2 and node $k$ receives only 2 edges from it; thus by Rule~\ref{rule:uniform-in-deg}, the cyclic union survives.  This contrast shows that we cannot hope for a survival rule for cyclic unions that relies only on knowing the survival of individual components.

\subsection{Bidirectional simply-added splits}

A key feature of composite graphs that allowed us to prove the previous results is that for each component, $G_i$, the rest of the graph is simply-added onto it. Furthermore, in the case of disjoint unions and clique unions, every subset of components has the property that the rest of the graph is simply-added onto it, and vice versa. This is also true of all composite graphs with only two components, but is not a general feature of composite graphs. For example, cyclic unions do not have this property: no single component is simply-added onto the rest of the graph.

Here we consider graphs $G$ that have a {\it bidirectional simply-added split}, meaning that $[n] = \tau \cup \omega$, where $\omega$ is simply-added onto $\tau$ and $\tau$ is simply-added onto $\omega$. Note that although every composite graph with $N = 2$ components satisfies this property, $G$ need not be such a graph (see Figure~\ref{fig:doubly-simply-added}). For example, $\tau$ could send both projectors and non-projectors to $\omega$, and vice versa, so that $G|_\tau$ and $G|_\omega$ are not valid components of a composite graph.

\begin{figure}[!h]
\begin{center}
\includegraphics[height=1.15in]{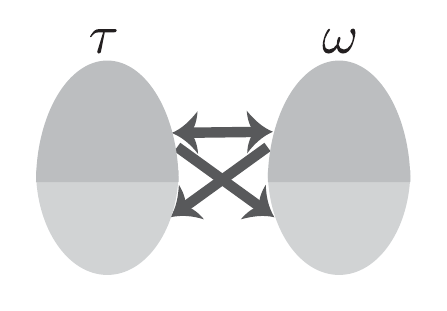}
\vspace{-.15in}
\end{center}
\caption{\textbf{Bidirectional simply-added split.} In this graph $\tau$ is simply-added to $\omega$ and vice versa.  Thus $\tau$ is composed of two classes of nodes: projectors onto $\omega$ (top dark gray region) and non-projectors onto $\omega$ (bottom light gray region).  Similarly, $\omega$ can be decomposed into projectors and non-projectors onto $\tau$.  The thick arrows indicate that every node of a given region sends an edge to every node in the other region. The edges within $\tau$ and $\omega$ can be arbitrary. }
\label{fig:doubly-simply-added}
\vspace{-.15in}
\end{figure}

The following theorem will use some new notation. Let $G$ be a graph on $n$ nodes. For any $\omega \subseteq [n]$, let $S_\omega$ denote the fixed point supports of $G|_\omega$ that survive to be fixed points of $G$, and let $D_\omega$ denote the non-surviving (dying) fixed points:
$$S_\omega \od \FP(G|_\omega) \cap \FP(G), \quad \text{and} \quad D_\omega \od \FP(G|_\omega)\setminus S_\omega.$$
Recall that by Corollary~\ref{cor:on-off-conds}, in order to check if $\sigma \in  \FP(G|_\omega)$ survives to $\FP(G)$ we need only to check that $\sigma \in \FP(G|_{\omega \cup k})$ for each $k \notin \omega$.
In the case where $G$ has a simply-added split, $[n] = \tau \cup \omega$, with $\omega$ simply-added onto $\tau$, each $k \notin \omega$ receives precisely the same edges from $\omega$, and so we need only check that $\sigma \in \FP(G|_{\omega \cup k})$ for a single $k \in \tau$. In this case, 
$$S_\omega = \{\sigma \subseteq \omega \mid \sigma \in \FP(G|_{\omega \cup k})\},$$
for any choice of $k \in \tau$.

\begin{theorem}\label{thm:doubly-simply-added}
Let $G$ be a graph with bidirectional simply-added split $[n] = \tau \cup \omega$. For any nonempty $\sigma \subseteq [n]$, let $\sigma = \sigma_\tau \cup \sigma_\omega$ where
$\sigma_\tau = \sigma \cap \tau$ and $\sigma_\omega = \sigma \cap \omega$. Then,
$\sigma \in \FP(G)$ if and only if one of the following holds:
\begin{itemize}
\item[(i)]$ \sigma_\tau \in S_\tau \cup \emptyset\;\; \text{and}\;\; \sigma_\omega \in S_\omega \cup \emptyset, \;\; \text{or}$
\item[(ii)] $\sigma_\tau \in D_\tau \;\; \text{and}\;\; \sigma_\omega \in D_\omega.$
\end{itemize}
Moreover, if $\sigma \in \FP(G)$ and both $\sigma_\tau$ and $\sigma_\omega$ are nonempty, then its index is given by
$$\idx(\sigma) = \left\{\begin{array}{rl} 
-\idx(\sigma_\tau)\idx(\sigma_\omega), & \text{if }\; \sigma_\tau \in S_\tau, \;\sigma_\omega \in S_\omega \\
\idx(\sigma_\tau)\idx(\sigma_\omega), & \text{if }\; \sigma_\tau \in D_\tau, \;\sigma_\omega \in D_\omega \end{array}\right.$$
Otherwise, $\sigma = \sigma_\tau$ or $\sigma_\omega$, and has the same index.
\end{theorem}

\begin{proof}
First we consider the case where $\sigma_\tau$ or $\sigma_\omega$ is empty. Without loss of generality, suppose $\sigma_\tau = \emptyset$. Then $\sigma = \sigma_\omega$, and clearly $\sigma \in \FP(G)$ if and only if $\sigma_\omega \in S_\omega$.

Now suppose $\sigma \subseteq [n]$ with both $\sigma_\tau$ and $\sigma_\omega$ nonempty. Since $G$ has a bidirectional simply-added split along $(\tau,\omega)$, using Theorem~\ref{thm:simply-added} we see that for any $i \in [n]$, we have
$$s_i^\sigma = \dfrac{1}{\theta} s_i^{\sigma_\tau}s_i^{\sigma_\omega}.$$
Furthermore, for any $i \in \tau$, all $s_i^{\sigma_\omega}$ have the same value and so
$s_i^\sigma = \alpha s_i^{\sigma_\tau}$. Similarly, for any $j \in \omega$, we have $s_j^\sigma = \beta s_j^{\sigma_\omega}$. Thus, the relative signs of $s_i^{\sigma_\tau}$ across $i \in \tau$ are preserved in the $s_i^{\sigma}$, and the same for the relative signs of $s_j^{\sigma_\omega}$ across $j \in \omega$. Hence, $\sigma \in \FP(G)$ if and only if $\sigma_\tau \in \FP(G|_\tau)$, $\sigma_\omega \in \FP(G|_\omega)$, and $\sgn(s_i^\sigma) = \sgn(s_j^\sigma)$ for any $i \in \sigma_\tau$ and $j \in \sigma_\omega$.

To see when the above signs agree, observe that $\sigma_\tau \in \FP(G|_\tau)$ implies
$\sgn(s_i^{\sigma_\tau}) = \idx(\sigma_\tau)$ for all $i \in \sigma_\tau$, and similarly 
$\sgn(s_j^{\sigma_\omega}) = \idx(\sigma_\omega)$ for all $j \in \sigma_\omega$.
Therefore, $\sgn(s_i^\sigma) = \sgn(s_j^\sigma)$ if and only if
\begin{equation}\label{eq:sgn-eqn}
\sgn(s_i^{\sigma_\omega})\idx(\sigma_\tau) = \sgn(s_j^{\sigma_\tau})\idx(\sigma_\omega).
\end{equation}
However, $\sgn(s_i^{\sigma_\omega})$ and $\sgn(s_j^{\sigma_\tau})$ depend on whether or not $\sigma_\omega$ and $\sigma_\tau$, respectively, survive to fixed points of $G$. To track this, we define
$\chi(\sigma_\tau) = 1$ if $\sigma_\tau \in S_\tau$, and $\chi(\sigma_\tau) = -1$ if $\sigma_\tau \in D_\tau$. (Note that $\sigma_\tau \in \FP(G|_\tau)$ implies $\sigma_\tau \in S_\tau \dot\cup D_\tau$.) In particular, since $j \notin \sigma_\tau$, we see that $\sgn(s_j^{\sigma_\tau})$ agrees with $\idx(\sigma_\tau)$ if and only if $\sigma_\tau \in D_\tau$.
Thus, we can write $\sgn(s_j^{\sigma_\tau}) = -\chi(\sigma_\tau)\idx(\sigma_\tau)$ and
$\sgn(s_i^{\sigma_\omega}) = -\chi(\sigma_\omega)\idx(\sigma_\omega)$. Plugging this into equation~\eqref{eq:sgn-eqn} we see that $\sgn(s_i^\sigma) = \sgn(s_j^\sigma)$ if and only if
$$-\chi(\sigma_\omega)\idx(\sigma_\omega)\idx(\sigma_\tau) = 
-\chi(\sigma_\tau)\idx(\sigma_\tau)\idx(\sigma_\omega),$$
which holds if and only if $\chi(\sigma_\tau) = \chi(\sigma_\omega).$ In other words, we can conclude that $\sigma \in \FP(G)$ if and only if both $\sigma_\tau \in S_\tau$ and $\sigma_\omega \in S_\omega$, or both $\sigma_\tau \in D_\tau$ and $\sigma_\omega \in D_\omega$, as desired. Finally, the index formulas follow from observing that  
if $\sigma \in \FP(G)$, then $\idx(\sigma) = -\chi(\sigma_\tau)\idx(\sigma_\tau)\idx(\sigma_\omega).$
\end{proof}

It is worth noting that Theorem~\ref{thm:doubly-simply-added} allows us to recover the results on disjoint unions and clique unions as a special case. If $\tau$ and $\omega$ only contain non-projectors, then $\tau \cup \omega$ is a disjoint union, where every fixed point of each subset survives, i.e.\ $S_\tau = \FP(G|_\tau)$, $S_\omega = \FP(G|_\omega)$, and $D_\tau=D_\omega=\emptyset$.   Thus, Theorem~\ref{thm:doubly-simply-added} shows that the fixed points of a disjoint union are all the fixed points of the individual components and unions of these.  On the other hand, if $\tau$ and $\omega$ only contain projectors, then $\tau \cup \omega$ is a clique union.  In this case, every fixed point of each subset does not survive (dies) because it has a target, and so $D_\tau = \FP(G|_\tau)$, $D_\omega = \FP(G|_\omega)$, and $S_\tau = S_\omega = \emptyset$.  Therefore, Theorem~\ref{thm:doubly-simply-added} shows that the fixed points of a clique union are solely the unions of fixed points from every component.

\section{Domination}\label{sec:domination}

In this section, we introduce a more general form of {\it domination}, which is broader than the concept of graphical domination first introduced in Section~\ref{sec:graph-domination}. Domination applies to all competitive and nondegenerate TLNs with uniform external inputs, so that $b_i = b_j = \theta$ for all $i,j \in [n]$.  Furthermore, while graphical domination is insufficient to determine all permitted and forbidden motifs of a CTLN (see Appendix Section~\ref{appendixB}), general domination precisely characterizes all fixed point supports not only for CTLNs but also for TLNs. In particular, our main result on domination, Theorem~\ref{thm:domination}, gives an alternative to Theorem~\ref{thm:sgn-condition} (sign conditions).

As we will see later in this section, this broader form of domination is not practical for explicit computations, but provides a useful technical tool for proving results about fixed points {\it without} appealing to the signs of the $s_i^\sigma$. In particular, domination will allow us to prove Theorems~\ref{thm:graph-domination} (graphical domination) and~\ref{thm:uniform-in-degree} (uniform in-degree), as well as Propositions~\ref{prop:disjoint-union-survival} and~\ref{prop:clique-union-survival} (survival rules for disjoint and clique unions).

\subsection{General domination}\label{sec:gen-domination}

Let $(W,\theta)$ be a TLN with uniform inputs $\theta$, and define $\Wtil = -I + W$.
The quantities of interest for general domination are the sums:
\begin{equation}\label{eq:w_i}
w_j^\sigma \od \sum_{i \in \sigma} \Wtil_{ji} |s_i^\sigma|.
\end{equation}
In contrast to Theorem~\ref{thm:sgn-condition} (sign conditions), where the signs of the $s_i^\sigma$ were essential to determining fixed point supports, here we completely discard the signs and use only the absolute values, $|s_i^\sigma|$. 

The definition of the $w_j^\sigma$ may seem mysterious at first. Before defining domination or stating our main theorem about it, we will work through an example that shows how the $w_j^\sigma$ values encode information about the fixed points supports of a TLN. There are two main observations that will emerge from this example. First, if $\sigma \in \FP(W,\theta)$ then $w_j^\sigma = w_k^\sigma$ for all $j,k \in \sigma$. In other words, the $w_j^\sigma$ precisely match for all nodes inside a permitted motif. Second,  if $\sigma \in \FP(W,\theta)$ and there is some $k \notin \sigma$, then for $j \in \sigma$ we have $w_j^\sigma > w_k^\sigma$. So the values of $w_j^\sigma$ {\it inside} the fixed point support are all equal to each other and greater than the values of $w_k^\sigma$ for nodes {\it outside}.

\begin{figure}[!ht]
\vspace{-.1in}
\begin{center}
\includegraphics[height=1.4in]{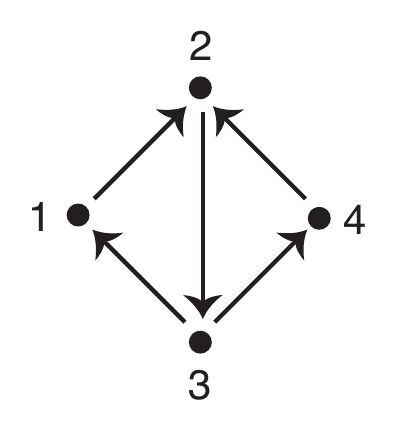}
\caption{The butterfly graph for Example~\ref{ex:butterfly-graph}}
\label{fig:butterfly-graph}
\end{center}
\vspace{-.1in}
\end{figure}

\begin{example}\label{ex:butterfly-graph}
Consider a CTLN whose graph $G$ is the butterfly graph in Figure~\ref{fig:butterfly-graph}, which is studied in detail in Appendix Section~\ref{appendixB}.  By appealing to earlier graph rules and parity, it is straightforward to see that $\FP(G) = \{123, 234, 1234\}$.  In this example, we explore the values of $w_j^\sigma$ for different subsets of vertices in this graph.  

First consider $\sigma = \{1,2,3,4\}$, which we already know is a permitted motif.  In Appendix Section~\ref{appendixB}, we computed the following values for the butterfly graph:
$$s_1^\sigma = s_4^\sigma = -\delta\theta(\varepsilon^2 + \varepsilon\delta + \delta^2), \quad
s_2^\sigma = -\delta\theta(2\varepsilon^2+2\varepsilon\delta + \delta^2), \quad
s_3^\sigma = -\delta\theta(2\varepsilon^2 + 3\varepsilon\delta + 2\delta^2)$$
Using these calculations, we can compute the $w_i^\sigma$ values. For example,
\begin{eqnarray*}
w_1^\sigma &=& \sum_{i \in \sigma} \Wtil_{1i} |s_i^\sigma| = (-1)|s_1^\sigma| + (-1-\delta)|s_2^\sigma| + (-1+\varepsilon)|s_3^\sigma|+(-1-\delta)|s_4^\sigma|\\
&=& \delta \theta (2\varepsilon^3 - 6\varepsilon^2 -7\varepsilon\delta -\varepsilon\delta^2 -5\delta^2 -2\delta^3).
\end{eqnarray*}
By symmetry, it is clear that $w_4^\sigma = w_1^\sigma$, but in fact, all the $w_i^\sigma$ values are equal across $i \in \sigma$.  Although it is not obvious that $w_2^\sigma$ should match $w_1^\sigma$, it does:
\begin{eqnarray*}
w_2^\sigma &=& \sum_{i \in \sigma} \Wtil_{2i} |s_i^\sigma| = (-1)|s_2^\sigma| + (-1+\varepsilon)(|s_1^\sigma|+|s_4^\sigma|)+ (-1-\delta)|s_3^\sigma|\\
&=& \delta \theta (2\varepsilon^3 - 6\varepsilon^2 -7\varepsilon\delta -\varepsilon\delta^2 -5\delta^2 -2\delta^3).
\end{eqnarray*}
Theorem~\ref{thm:domination} will show that this is the hallmark of permitted motifs.

Next consider $\tau = \{1,2,3\}$.  Since $\tau$ is a 3-cycle, the $s_i^\tau$ values can be obtained from graph 8 in Figure~\ref{fig:n3-graphs-s_i}, giving us $s_i^\tau = \theta(\varepsilon^2 + \varepsilon\delta + \delta^2)$ for all $i \in \tau$.  Then, for all $j \in \tau$, 
\begin{eqnarray*}
w_j^\tau &=&  \sum_{i \in \sigma} \Wtil_{ji} |s_i^\sigma| 
= \theta(\varepsilon^2 + \varepsilon\delta + \delta^2)(-1 + (-1+\varepsilon) + (-1 -\delta))\\
&=& \theta(\varepsilon^2 + \varepsilon\delta + \delta^2)(-3+\varepsilon-\delta).
\end{eqnarray*}
Additionally,
\begin{eqnarray*}
w_4^\tau &=&  \sum_{i \in \sigma} \Wtil_{4i} |s_i^\sigma| 
= \theta(\varepsilon^2 + \varepsilon\delta + \delta^2)((-1+\varepsilon) + 2(-1 -\delta))\\
&=& \theta(\varepsilon^2 + \varepsilon\delta + \delta^2)(-3+\varepsilon-2\delta).
\end{eqnarray*}
Notice that $w_j^\tau > w_4^\tau$ for all $j \in \tau$. In fact, the inequality that $w_j^\tau > w_k^\tau$ for all $j \in \tau$ and $k \notin \tau$ must be satisfied for any permitted motif $\tau$ to survive as a fixed point of a larger graph $G$.  This observation, together with the fact that within a permitted motif all the values $w_i^\sigma$ match, is captured below in Theorem~\ref{thm:domination}.
\end{example}

In order to compute the values $w_j^\sigma$ in the above example, we used our pre-computed values for $s_i^\sigma$.  However, once we know the $s_i^\sigma$, we can simply apply Theorem~\ref{thm:sgn-condition} (sign conditions) and be done.  In particular, it is not practical to explicitly compute the $w_j^\sigma$ in order to check whether or not $\sigma$ is a fixed point support. The true value of introducing the $w_j^\sigma$ is that arguments can be made about their relative values {\it without} knowing the signs of the $s_i^\sigma$, and this can in turn be used to determine whether or not $\sigma \in \FP(W, \theta)$.
With this motivation, we now define {\it domination}, which is a generalization of graphical domination (first introduced in Section~\ref{sec:graph-domination}).

\begin{definition}[domination]\label{def:general-domination}
Consider a TLN $(W,\theta)$ on $n$ neurons, let $\sigma \subseteq [n]$ be nonempty, and let $w_j^\sigma$ be defined as in~\eqref{eq:w_i}.  For any $j,k \in [n]$, we say that
\begin{itemize}
\item $k$ {\em dominates} $j$ with respect to $\sigma$, and write $k >_{\sigma} j$, whenever
$w_k^\sigma > w_j^\sigma$;
\item $k$ is {\em equivalent} to $j$ with respect to $\sigma$, and write
$k \sim_\sigma j$, whenever $w_k^\sigma = w_j^\sigma$;
\item $k$ is {\em not dominated} by $j$ with respect to $\sigma$, and write
$k \geq_\sigma j$, whenever $w_k^\sigma \geq w_j^\sigma$.
\end{itemize}
Moreover, if $j \sim_\sigma k$ for all $j,k \in \sigma$, then we say that $\sigma$ is {\em domination-free}.
\end{definition}

The domination relation satisfies nice properties.  Clearly, if $\ell >_\sigma k$ and $k >_\sigma j$, then $\ell >_\sigma j$.  Furthermore, if $k >_\sigma j$ then we cannot also have $j >_\sigma k$.  We thus see that $>_\sigma$ is transitive and antisymmetric, but not reflexive; while $\geq_\sigma$ is transitive, antisymmetric, and reflexive. This makes $>_\sigma$ a strict partial order and $\geq_\sigma$ a partial order.  Incomparable elements under $>_\sigma$ belong to equivalence classes of the equivalence relation $\sim_\sigma$.  It is easy to see that $>_\sigma, \geq_\sigma$ and $\sim_\sigma$ interact just as the usual ``$>, \geq$'' and ``$=$'' do. In particular,
if $\ell>_\sigma k$ and $k \sim_\sigma j$, then $\ell >_\sigma j$.

We can now give our second characterization of $\FP(W, \theta)$, for the case of uniform inputs $b_i = \theta > 0$ for each $i \in [n]$. Recall from Definition~\ref{def:permitted-forbidden} that $\sigma$ is a permitted motif if $\sigma \in \FP(W_\sigma,\theta)$, and $\sigma$ is a forbidden motif otherwise. In order to have $\sigma \in \FP(W,\theta)$, we must have that $\sigma$ is a permitted motif that survives as a fixed point support in the full network.

\begin{theorem}[general domination] \label{thm:domination}
Let $(W,\theta)$ be a TLN, and let $\sigma \subseteq [n]$. Then
$$\sigma \text{ is a permitted motif } \Leftrightarrow \sigma \text{ is domination-free.}$$
If $\sigma$ is a permitted motif, then $\sigma \in \FP(W,\theta)$ if and only if
for each $k \not\in \sigma$ there exists $j \in \sigma$ such that $j >_\sigma k$.
\end{theorem}

In addition to telling us that permitted motifs are domination-free, Theorem~\ref{thm:domination} states that permitted motifs survive precisely when every node outside the motif is (inside-out) dominated by some  node inside. This is precisely what we saw in Example~\ref{ex:butterfly-graph}. Namely, the $w_j^\tau$ and $w_j^\sigma$ all matched for $j \in \tau$ or $j \in \sigma$, respectively. Moreover, for $j \in \tau$ and $k \notin \tau$, we saw that $w_j^\tau > w_k^\tau$, and thus $j >_\tau k$, consistent with the fact that $\tau \in \FP(G)$.

Before presenting the proof of Theorem~\ref{thm:domination}, in the next section, we illustrate it with another example. Let $G$ be the graph in Figure~\ref{fig:counterexample-graph}, consider $\tau = \{1, 2, 3, 4\}$.  Since $\tau$ is uniform in-degree (with $d=1$) and node 5 receives more than $d$ edges from $\tau$, Theorem~\ref{thm:uniform-in-degree} guarantees that $\tau$ is a permitted motif that does not survive in the full graph, and so we know that $\tau \in \FP(G|_\tau)$ but $\tau \notin \FP(G)$. However, this result cannot be obtained from graphical domination. In fact, the proof of Theorem~\ref{thm:uniform-in-degree} relies on general domination. 
In the following example, we will show that $\tau \notin \FP(G)$ directly, using Theorem~\ref{thm:domination} (domination).  This previews the proof of Theorem~\ref{thm:uniform-in-degree} in Section~\ref{sec:proofs-uniform-in-degree}. 

\begin{figure}[!ht]
\begin{center}
\includegraphics[width=1.5in]{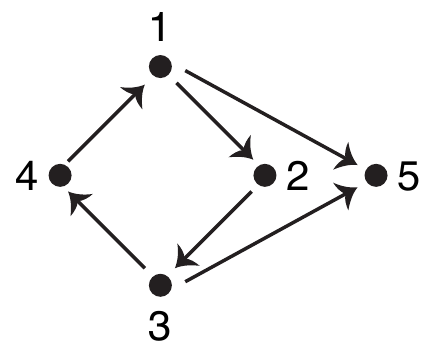}
\caption{Graph for Example~\ref{ex:counterexample-graph}.
}
\label{fig:counterexample-graph}
\end{center}
\vspace{-.2in}
\end{figure}

\begin{example}\label{ex:counterexample-graph}

Let $G$ be the graph in Figure~\ref{fig:counterexample-graph}, and let $\tau = \{1, 2, 3, 4\}$.
Observe that for each $i \in \tau$, we can compute $s_i^\tau = \det((I-W_\tau)_i; \theta \one) = \varepsilon \theta ( \varepsilon^2 + 2 \varepsilon\delta + 2 \delta^2)$.  Then for all $j \in \tau$, we have
\begin{eqnarray*}
 w_j^\tau &=& \sum_{i \in \tau} \Wtil_{ji}|s_i^\tau| = \varepsilon \theta ( \varepsilon^2 + 2 \varepsilon\delta + 2 \delta^2) \sum_{i \in \tau} \Wtil_{ji}\\
 &=& \varepsilon \theta ( \varepsilon^2 + 2 \varepsilon\delta + 2 \delta^2) (-4+\varepsilon - 2\delta).
 \end{eqnarray*}
Since the $w_j^\tau$ have the same value for all $j \in \tau$, we see that $\tau$ is domination-free, and hence is permitted.  Next, 
\begin{eqnarray*}
 w_5^\tau &=& \sum_{i \in \tau} \Wtil_{5i}|s_i^\tau| = \varepsilon \theta ( \varepsilon^2 + 2 \varepsilon\delta + 2 \delta^2) \sum_{i \in \tau} \Wtil_{5i}\\
 &=& \varepsilon \theta ( \varepsilon^2 + 2 \varepsilon\delta + 2 \delta^2) (-4+2\varepsilon - 2\delta) > w_j^\tau,
 \end{eqnarray*}
for $j \in \tau$. Thus, $5 >_\tau j$, and so $\tau \notin \FP(G)$ by Theorem~\ref{thm:domination}.  

Now consider $\sigma = \{1,2,3,4,5\}$.  Observe that node 5 is a non-projector onto $\tau$, and thus by Corollary~\ref{cor:simply-added}, $s_i^\sigma = -\delta s_i^\tau =  -\varepsilon \delta \theta ( \varepsilon^2 + 2 \varepsilon\delta + 2 \delta^2)$ for all $i \in \tau$.  Computing $s_5^\sigma = \det((I-W_\sigma)_5; \theta \one)$, we obtain $s_5^\sigma = \varepsilon^2 \theta ( \varepsilon^2 + 2 \varepsilon\delta + 2 \delta^2)$.  Then for all $j \in \sigma \setminus 5$, we have
$$w_j^\sigma = \sum_{i \in \sigma \setminus 5} \Wtil_{ji}|s_i^\sigma| + (-1-\delta)|s_5^\sigma| = \varepsilon \theta ( \varepsilon^2 + 2 \varepsilon\delta + 2 \delta^2) (-\varepsilon-4\delta -2\delta^2),$$
while
$$w_5^\sigma = \sum_{i \in \sigma \setminus 5} \Wtil_{5i}|s_i^\sigma| + (-1)|s_5^\sigma| = \varepsilon \theta ( \varepsilon^2 + 2 \varepsilon\delta + 2 \delta^2) (-\varepsilon + 2\varepsilon\delta-4\delta -2\delta^2).$$
Since $w_5^\sigma \neq w_j^\sigma$ for $j, 5 \in \sigma$, we see that $\sigma$ is \underline{not} domination-free, and hence $\sigma$ is a forbidden motif.\footnote{Note that we also could have concluded that $\sigma$ was forbidden by Rule~\ref{rule:added-sink} since it is the union of sink with a non-surviving fixed point.}
\end{example}

\subsection{Proof of Theorem~\ref{thm:domination} (general domination)}\label{sec:dom-proofs}

Here we assume $b_i = \theta $ for all $i \in [n]$.  For constant $\theta>0$, the collection of fixed point supports is independent of the value of $\theta$, and so we drop $\theta$ from the notation and denote the set of fixed point supports simply as $\FP(W)$.  

Note that 
$s_i^\sigma = \det((I-W_\sigma)_i;\theta \one)$ for each $i \in \sigma$,
and equation~\eqref{eq:s_k-identity} implies:
\begin{equation*}
- s_k^\sigma + \sum_{i \in \sigma\setminus\{k\}} W_{ki}s_i^\sigma  =  - s_j^\sigma  + \sum_{i \in \sigma\setminus\{j\}} W_{ji}s_i^\sigma \;\; \text{ for all } \; j,k \in [n],
\end{equation*}
where we have used the fact that $b_k = b_j$ and $W_{kk} = W_{jj} = 0$.
If we denote $\Wtil = -I + W,$ then the terms $s_j^\sigma, s_k^\sigma$ can be absorbed into the sum when $j,k \in \sigma$.  This allows us to write:
\begin{equation}\label{eq:sk-sj-3}
(\chi_\sigma(k)-1)s_k^\sigma + \sum_{i \in \sigma} \Wtil_{ki}s_i^\sigma  =  (\chi_\sigma(j)-1) s_j^\sigma  + \sum_{i \in \sigma} \Wtil_{ji}s_i^\sigma \;\; \text{ for all } \; j,k \in [n],
\end{equation}
where $\chi_\sigma(k) = 1$ if $k \in\sigma$ and $\chi_\sigma(k) = 0$ if $k \notin \sigma$.  In particular, if $j,k \in \sigma$ the above expression reduces to 
\begin{equation}\label{eq:sk-sj-4}
\sum_{i \in \sigma} \Wtil_{ki}s_i^\sigma = \sum_{i \in \sigma} \Wtil_{ji}s_i^\sigma,  \;\; \text{ for } \; j,k \in \sigma.
\end{equation}  
Equations~\eqref{eq:sk-sj-3} and~\eqref{eq:sk-sj-4} are true in general, irrespective of whether or not $\sigma$ is a fixed point support.

Now assume $\sigma$ is a permitted motif, so that $\sigma \in \FP(W_\sigma)$. In this case, all the $s_i^\sigma$ for $i \in \sigma$ must have the same sign, and so \eqref{eq:sk-sj-3} and \eqref{eq:sk-sj-4} 
continue to hold if each $s_i^\sigma$ is replaced by $|s_i^\sigma|$.  In fact, in the case of  \eqref{eq:sk-sj-4} the converse is also true: if equality holds after replacing $s_i^\sigma$ with $|s_i^\sigma|$, then $\sigma \in \FP(W_\sigma)$. Using the notation 
$w_j^\sigma = \sum_{i \in \sigma} \Wtil_{ji} |s_i^\sigma|,$ as in~\eqref{eq:w_i}, we have the following lemma.

\begin{lemma}\label{lem:equality}
$\sigma \in \FP(W_\sigma)\; \Leftrightarrow \;  w_j^\sigma = w_k^\sigma$  
for all $j,k\in\sigma$. (I.e., $\sigma \in \FP(W_\sigma)\; \Leftrightarrow j \sim_\sigma k$ for all $j,k\in\sigma$.)
\end{lemma}

\begin{proof}  ($\Rightarrow$) This direction was already established in the arguments above.  ($\Leftarrow$) To see the converse,
suppose $w_j^\sigma = w_k^\sigma$ for all $j,k \in \sigma$. Let 
$\alpha = - w_j^\sigma =  - \sum_{i \in \sigma} \Wtil_{ji} |s_i^\sigma| = \sum_{i \in \sigma} (I-W)_{ji} |s_i^\sigma|$.  Note that $\alpha > 0$, and let $v = (|s_i^\sigma|)_{i \in \sigma}$ be the column vector whose entries are $|s_i^\sigma|$ for each $i \in \sigma$.  It follows that $(I-W_\sigma)v = \alpha \one_\sigma$, and so for $x_\sigma = \frac{\theta}{\alpha} v$ we have $(I-W_\sigma)x_\sigma = \theta \one_\sigma$.  Since $x_\sigma$ has strictly positive entries, it is a fixed point of the network $(W_\sigma, \theta)$ with support $\sigma$.  Thus $\sigma \in \FP(W_\sigma)$.
\end{proof}

Lemma~\ref{lem:equality} can be restated as saying that $\sigma \in \FP(W_\sigma)$ if and only if $\sigma$ is domination-free.  Interestingly, the nondegeneracy condition on TLNs guarantees that if $\sigma$ is a permitted motif, then a node outside of $\sigma$ can never have the same $w_j^\sigma$ value as one inside $\sigma$, and so one must dominate the other.

\begin{lemma}\label{lem:either-or} Let $\sigma \in \FP(W_\sigma)$.
If $j \in \sigma$ and $k \notin \sigma$, then either $k >_\sigma j$ or $j >_\sigma k$. (I.e., we cannot have $j \sim_\sigma k$.)
\end{lemma}

\begin{proof}
It follows from equation~\eqref{eq:sk-sj-3} that if $j \in \sigma$ and $k \notin \sigma$, then 
$$s_k^\sigma = \sum_{i \in \sigma} \Wtil_{ki} s_i^\sigma - \sum_{i \in \sigma} \Wtil_{ji} s_i^\sigma,$$
and so
$\sgn s_k^\sigma = \sgn \left( \sum_{i \in \sigma} \Wtil_{ki} s_i^\sigma - \sum_{i \in \sigma} \Wtil_{ji} s_i^\sigma \right)$. If $\sigma \in \FP(W_\sigma)$, then all $s_i^\sigma$ for $i \in \sigma$ have the same sign. Replacing $s_i^\sigma$ with $|s_i^\sigma|$ yields
$$\sgn s_k^\sigma = \pm \sgn(w_k^\sigma - w_j^\sigma).$$
By the assumption of nondegeneracy of the TLN, we know that $\sgn s_k^\sigma \neq 0$, and so either $w_k^\sigma > w_j^\sigma$ or $w_j^\sigma > w_k^\sigma$. 
\end{proof}

The next lemma tells us when fixed points survive the addition of a single node.

\begin{lemma}\label{lem:survival}
Suppose $\sigma \in \FP(W_\sigma)$, and $k \notin \sigma$.  Then $\sigma \in \FP(W_{\sigma \cup \{k\}})$ if and only if $j >_\sigma k$ for some $j \in \sigma$.  If $\sigma \notin \FP(W_{\sigma \cup \{k\}})$, then $k>_\sigma j$ for all $j \in \sigma$.
\end{lemma}

\begin{proof} Let $\sigma \in \FP(W_\sigma)$, $k \notin \sigma$, and $j \in \sigma$. Recall from the proof of Lemma~\ref{lem:either-or} that 
$\sgn s_k^\sigma = \sgn \left( \sum_{i \in \sigma} \Wtil_{ki} s_i^\sigma - \sum_{i \in \sigma} \Wtil_{ji} s_i^\sigma \right)$. By Theorem~\ref{thm:sgn-condition} and Corollary~\ref{cor:on-off-conds}, $\sigma \in \FP(W_{\sigma\cup\{k\}})$ if and only if
 $\sgn s_k^\sigma = - \sgn s_j^\sigma$. If $\sgn s_i^\sigma = +1$ for each $i \in \sigma$, then replacing $s_i^\sigma$ with $|s_i^\sigma|$ in the sums reveals that $\sigma \in \FP(W_{\sigma \cup \{k\}})$ if and only if $j >_\sigma k$.  Similarly, if $\sgn s_i^\sigma = -1,$ we also have that $\sigma \in \FP(W_{\sigma \cup \{k\}})$ if and only if $j >_\sigma k$.  On the other hand, if $\sigma \not\in \FP(W_{\sigma \cup \{k\}})$, then we must have $k >_\sigma j$ by Lemma~\ref{lem:either-or}.
\end{proof}

Combining these results with Corollary~\ref{cor:on-off-conds}, we obtain the proof of Theorem~\ref{thm:domination}.

\begin{proof}[Proof of Theorem~\ref{thm:domination} (domination)]
The first statement, that $\sigma \in \FP(W_\sigma)$ if and only if $\sigma$ is domination-free, follows directly from Lemma~\ref{lem:equality}. Next, recall that by Corollary~\ref{cor:on-off-conds} we have $\sigma \in \FP(W)$ if and only if $\sigma \in \FP(W_\sigma)$ and $\sigma \in \FP(W_{\sigma \cup \{k\}})$ for all $k \notin \sigma$.  Applying Lemma~\ref{lem:survival}, we can conclude that for any permitted motif $\sigma$, we have $\sigma \in \FP(W)$ if and only if for each $k \notin \sigma$ there exists a $j \in \sigma$ such that $j>_\sigma k$.
\end{proof}

We end this section with a lemma that collects some key facts relating domination to fixed point supports, and parallels Theorem~\ref{thm:graph-domination} (graphical domination). In fact, we will use this in the next section to prove Theorem~\ref{thm:graph-domination}. 

\begin{lemma}[domination]\label{lem:domination}  Let $(W,\theta)$ be a TLN.  Suppose $k >_\sigma j$ for some $j,k \in [n]$.
\begin{itemize}
\item[(a)] If $j,k \in \sigma$, then $\sigma \notin \FP(W_\sigma)$, and thus $\sigma \notin \FP(W)$.
\item[(b)] If $j \in \sigma$, $k \not\in \sigma$, then $\sigma \notin \FP(W_{\sigma\cup\{k\}})$, and thus $\sigma \notin \FP(W)$. 
\item[(c)] If $j \notin \sigma$, $k \in \sigma$, and $\sigma \in  \FP(W_\sigma)$, then 
$\sigma \in \FP(W_{\sigma\cup\{j\}})$.
\end{itemize}
\end{lemma}

\begin{proof}
Part (a) is a direct corollary of Theorem~\ref{thm:domination}, while parts (b) and (c) are direct consequences of Lemma~\ref{lem:survival}.
\end{proof}

\subsection{Proof of Theorem~\ref{thm:graph-domination} (graphical domination)}\label{sec:graph-dom-proof}

Using Lemma~\ref{lem:domination}, it is now straightforward to prove Theorem~\ref{thm:graph-domination}. But first, we need to show that graphical domination (as defined in Section~\ref{sec:graph-domination}) is indeed a special case of the more general domination. 

\begin{lemma}\label{lem:comb-dom}
Consider a CTLN on $n$ nodes, and let $j,k \in [n]$ and $\sigma \subseteq [n]$. Suppose $k$ graphically dominates $j$ with respect to $\sigma$. Then $k >_\sigma j$.
\end{lemma}

\begin{proof}
Recall that if $k$ graphically dominates $j$ with respect to $\sigma$, then $\sigma \cap \{j,k\} \neq \emptyset$. Moreover, three conditions hold: (1) for each $i \in \sigma \setminus\{j,k\},$ if $i \to j$ then $i \to k$; (2) if $j \in \sigma$ then $j \to k$; and (3) if $k \in \sigma$, then $k \not\to j$.

Next, recall that in a CTLN we have $\Wtil_{ji} = -1+\varepsilon$ if $i \to j$, $\Wtil_{ji} = -1-\delta$ if $i \not\to j$, and $\Wtil_{ii} = -1$.  Condition 1 thus implies that $\Wtil_{ji} \leq \Wtil_{ki}$ for each $i \in \sigma\setminus\{j,k\}$.  If $j \in \sigma$, condition 2 gives $\Wtil_{jj} < \Wtil_{kj}$, while if $k \in \sigma$ condition 3 implies $\Wtil_{jk} < \Wtil_{kk}.$  Putting these together we see that 
$$w_j^\sigma = \sum_{i \in \sigma} \Wtil_{ji}|s_i^\sigma| < \sum_{i \in \sigma} \Wtil_{ki}|s_i^\sigma| = w_k^\sigma,$$ 
where the inequality is strict because at least one of $j$ or $k$ is in $\sigma$, and $\Wtil_{ji} < \Wtil_{ki}$ for $i=j$ or $i=k$. Since $w_k^\sigma> w_j^\sigma$, it follows from the definition that $k >_\sigma j$.
\end{proof}

We can now prove Theorem~\ref{thm:graph-domination}, which tells us how to use graphical domination in order to rule in and rule out various fixed point supports.

\begin{proof}[Proof of Theorem~\ref{thm:graph-domination}]
First, observe using Lemma~\ref{lem:comb-dom} that $k>_\sigma j$. Now the statements (a), (b), and (c) all follow immediately from parts (a), (b), and (c) of Lemma~\ref{lem:domination}.
\end{proof}

Note that the converse of Lemma~\ref{lem:comb-dom} is not true: there could still be a domination relationship even if there is no graphical domination. For example, Appendix Section~\ref{appendixB} shows six graphs in Figure~\ref{fig:domination-free-forbidden} that are forbidden, and thus are \underline{not} domination-free by Theorem~\ref{thm:domination}. So there must be general domination in each of these graphs, despite the absence of graphical domination. 

Note also that, unlike graphical domination, general domination may be parameter dependent, even within the legal range. For example, all the graphs in Figure~\ref{fig:param-depend-min-perm} of Appendix Section~\ref{appendixC} have parameter-dependent domination relationships, and this is reflected by the fact that $\FP(G)$ depends on $\varepsilon$ and $\delta$.

\subsection{Proof of Theorem~\ref{thm:uniform-in-degree} (uniform in-degree)}\label{sec:proofs-uniform-in-degree}

In this section, we use general domination to prove Theorem~\ref{thm:uniform-in-degree}, giving conditions for when a uniform in-degree subset supports a fixed point.  We begin by showing that a uniform in-degree subset always supports a fixed point in its restricted subgraph and that fixed point has uniform firing rate values.  

\begin{lemma}\label{lemma:uniform}
If $\sigma$ has uniform in-degree $d$, then $\sigma \in \FP(G|_\sigma)$ and the corresponding fixed point $x^*$ is uniform, with values
$$x_i^* = \dfrac{\theta}{|\sigma|+\delta(|\sigma|-d-1) - \varepsilon d} \;\;\; \text{for each} \;\;\; i \in \sigma.$$
\end{lemma}

\begin{proof} 
Because $G|_\sigma$ has uniform in-degree, the row sums of $I-W_\sigma$ are all equal.  This implies that the all-ones vector $\one_\sigma$ is an eigenvector of $I-W_\sigma$, with eigenvalue $R$ equal to the row sum.  Now consider the vector $x^*$ satisfying $x_i^* = \theta/R$ for each $i \in \sigma$, and $x_k^* = 0$ for each $k \notin \sigma$.  Clearly, $(I-W_\sigma) x_\sigma^* = \theta \one_\sigma,$ and so $x_\sigma^*$ is a fixed point of the network restricted to $\sigma$.  Moreover, since all vertices in $G|_\sigma$ have in-degree $d$, then each row of $W_\sigma$ has $d$ terms with value $-1+\varepsilon$ and $|\sigma|-d-1$ terms with value $-1-\delta$.  This allows us to compute the row sum as
$$R = 1 - \sum_{i \in \sigma}W_{1i} = |\sigma| + \delta(|\sigma|-d-1) - \varepsilon d ,$$ 
yielding $x_i^* = \theta/(|\sigma| + \delta(|\sigma|-d-1) - \varepsilon d)$ for $i \in \sigma$, as desired.  Note that $x_i^*>0$ for the full range of $d$ values, so this fixed point always satisfies fixed-point condition (i).  \end{proof}

Next, we give the survival rule of a uniform in-degree fixed point in terms of domination.

\begin{lemma}[uniform in-degree domination]\label{lemma:uniform-domination}
Suppose $\sigma$ has uniform in-degree $d$, and suppose $j \in \sigma$ and $k \notin \sigma$.  Let $d_k = |\{i \in \sigma \mid i \to k\}|$ be the number of edges $k$ receives from $\sigma$.  Then $\sigma \in \FP(G|_\sigma)$ and
\begin{itemize}
\item[(i)] $k>_\sigma j$ if $d_k > d$, and
\item[(ii)] $j >_\sigma k$, if $d_k \leq d.$
\end{itemize}
\end{lemma}

\begin{proof}
By Cramer's rule, $s_i^\sigma = \det(I-W_\sigma)x_i^*$ (see Lemma~\ref{lemma:cramer} and equation~\eqref{eq:x_i-s_i}).  Thus, by Lemma~\ref{lemma:uniform}, $s_i^\sigma = s_j^\sigma$ for all $i,j \in \sigma$, when $\sigma$ is uniform in-degree.  This implies $\sigma \in \FP(G|_\sigma),$ and also allows us to factor $|s_i^\sigma|$ out of the sums for checking domination, so that $k >_\sigma j$ if and only if $\sum_{i \in \sigma} \Wtil_{ki} > \sum_{i \in \sigma} \Wtil_{ji}$.
Now observe that $\sum_{i \in \sigma} \Wtil_{ki} = d_k(-1+\varepsilon) + (|\sigma|-d_k)(-1-\delta)$, while 
$\sum_{i \in \sigma} \Wtil_{ji} = d(-1+\varepsilon) + (|\sigma|-d-1)(-1-\delta) -1 = 
d(-1+\varepsilon) + (|\sigma|-d)(-1-\delta) + \delta,$ since $j \in \sigma$ and $\Wtil_{jj} = -1$. 
In particular, if $d_k = d$ then we have  $\sum_{i \in \sigma} \Wtil_{ki} < \sum_{i \in \sigma} \Wtil_{ji}$, so that $j >_\sigma k$.  It is now easy to check that if $d_k > d$, then $k >_\sigma j$, while if $d_k \leq d$, then $j >_\sigma k$.
\end{proof}

Finally, we combine these results to prove Theorem~\ref{thm:uniform-in-degree} and prove the stability conditions.

\begin{proof}[Proof of Theorem~\ref{thm:uniform-in-degree}]
By Lemma~\ref{lemma:uniform}, if $\sigma$ has uniform in-degree, then it supports a fixed point in $G|_\sigma$, and thus $i \sim_\sigma j$ for all $i, j \in \sigma$.  By Lemma~\ref{lemma:uniform-domination}, for each $k \notin \sigma$, we have $j >_\sigma k$ precisely when $d_k \leq d$.  Thus by Theorem~\ref{thm:domination} (domination), $\sigma$ supports a fixed point in $G|_{\sigma \cup k}$ if and only if $d_k \leq d$.

For the stability conditions, recall that the fixed point is stable precisely when all the eigenvalues of $-I+W_\sigma$ have negative real part, or equivalently all the eigenvalues of $I-W_\sigma$ have positive real part.  First consider $d < |\sigma|/2$.  Observe that the uniform in-degree implies that the all-ones vector $1$ is an eigenvector of $I-W_\sigma,$ with eigenvalue $\lambda$ equal to the row sum:
$$\lambda = |\sigma| + (|\sigma|-d-1)\delta - d\varepsilon = |\sigma| + (|\sigma|-1)\delta - d(\delta+\varepsilon).$$
When $d < |\sigma|/2$, we have $d \leq (|\sigma|-1)/2$, and 
thus $\lambda > |\sigma|$ whenever $|\sigma| >1$ because $\varepsilon < \delta.$ 
On the other hand, since the sum of the eigenvalues equals the trace, $\operatorname{Tr}(I-W_\sigma) = |\sigma|$, we see that $I-W_\sigma$ must have a negative eigenvalue.  This implies that the fixed point is unstable.

Next consider when $d=|\sigma|-1$, so that $\sigma$ is a clique.  In this case, $I-W_\sigma = (1-\varepsilon)11^T + \varepsilon I_\sigma,$ and so the eigenvalues are $|\sigma|(1-\varepsilon)+\varepsilon$ and $\varepsilon$.  Clearly, these are positive
for $0<\varepsilon<\frac{\delta}{\delta+1}<1$, so we can conclude that the fixed point is stable. 
\end{proof}

\subsection{Domination and simply-added splits}\label{sec:domination-simply-added}

When $\sigma = \tau \: {\cup} \:\omega$, where $\omega$ is simply-added to $\tau$, we find that domination relationships with respect to $\tau$ are preserved with respect to $\sigma$.  

\begin{lemma}\label{lemma:simply-added}
Let $\sigma = \tau \:  {\cup} \:\omega$, where $\omega$ is simply-added to $\tau$.  Then for any $j,k \in \tau$,
$$(i) \;\; k >_\sigma j \; \Leftrightarrow \; k >_\tau j, \;\; \text{and} \;\;
(ii) \;\; k \sim_\sigma j  \; \Leftrightarrow \; k\sim_\tau j.$$

\noindent Furthermore, for any $j \in \tau$ and $\ell \notin \sigma$:

(iii) If for all $i \in \omega$ such that $i \to j$ we also have $i \to \ell$, then 
$$\ell >_\tau j \;\Rightarrow \;\ell >_\sigma j , \;\; \text{and } \;\;\ell \geq_\tau j \;\Rightarrow \;\ell \geq_\sigma j.$$ 

(iv) If for all $i \in \omega$ such that $i \to \ell$ we also have $i \to j$, then 
$$j >_\tau \ell \;\Rightarrow \;j >_\sigma \ell , \;\; \text{and } \;\; j \geq_\tau \ell \;\Rightarrow \;j \geq_\sigma \ell.$$ 

\end{lemma}

\noindent In particular, the lemma tells us that if $\tau \notin \FP(G|_\tau)$ then $\sigma \notin \FP(G)$.

\begin{proof}
Using Theorem~\ref{thm:simply-added} we see that for each $i \in \tau$, $|s_i^\sigma| = |\alpha||s_i^\tau|$, where $\alpha = s_i^\omega$, which is constant across $i \in \tau$.  In particular, for any fixed $j \in \tau$ we have
$$w_j^\sigma = \sum_{i \in \sigma} \Wtil_{ji}|s_i^\sigma| = \sum_{i \in \tau} \Wtil_{ji}|s_i^\sigma| + \sum_{i \in \omega} \Wtil_{ji}|s_i^\sigma| =  \sum_{i \in \tau} \Wtil_{ji}|\alpha||s_i^\tau| + \sum_{i \in \omega} \Wtil_{ji}|s_i^\sigma|= |\alpha|w_j^\tau +  \sum_{i \in \omega} \Wtil_{ji}|s_i^\sigma|,$$
where the last sum is identical for all $j \in \tau$ because $\omega$ is simply-added to $\tau$.  This means that
for any $j,k \in \tau$, 
$$w_k^\sigma - w_j^\sigma = |\alpha|(w_k^\tau - w_j^\tau).$$
Thus, (i) and (ii) both hold.

By the same logic as above, for any $j \in \tau$ and $\ell  \notin \sigma$, we have $w_\ell^\sigma = |\alpha|w_\ell^\tau +  \sum_{i \in \omega} \Wtil_{\ell i}|s_i^\sigma|$ and thus,
$$w_\ell^\sigma - w_j^\sigma = |\alpha|(w_\ell^\tau - w_j^\tau) + \sum_{i \in \omega} (\Wtil_{\ell i} - \Wtil_{ji})|s_i^\sigma|.$$
In (iii), we have $\Wtil_{\ell i} \geq \Wtil_{ji}$ for all $i \in \omega$. Thus if $\ell >_\tau j$, so that $w_\ell^\tau > w_j^\tau$, then $w_\ell^\sigma - w_j^\sigma >0$ and hence $\ell >_\sigma j$.  Similarly, we see that $\ell \geq_\tau j \;\Rightarrow \;\ell \geq_\sigma j$. 
Finally, in (iv) we have 
$\Wtil_{\ell i} \leq \Wtil_{ji}$ for all $i \in \omega$. Thus if $j >_\tau \ell$, so that  $w_j^\tau > w_\ell^\tau$, then $w_\ell^\sigma - w_j^\sigma<0$ and thus $j >_\sigma \ell$.  Similarly, $j \geq_\tau \ell \;\Rightarrow \;j \geq_\sigma \ell$. 
\end{proof}

Lemma~\ref{lemma:simply-added} shows that identifying a simply-added split of $\sigma$ can be useful for identifying domination relationships with respect to $\tau$ that lift to domination relationships with respect to $\sigma$.  This is particularly useful in composite graphs, since for any component of a composite graph, the rest of the graph is simply-added to that component.  In particular, we use Lemma~\ref{lemma:simply-added} to prove the partial survival rules for disjoint unions and clique unions embedded in a larger graph.  

\begin{proof}[Proof of Proposition~\ref{prop:disjoint-union-survival} (survival of disjoint union)]
Suppose there exists $i \in [N]$ such that $\sigma_i \notin \FP(G)$. 
If $\sigma_i \notin \FP(G|_{\sigma_i})$, then by Lemma~\ref{lemma:composite1} $\sigma \notin \FP(G)$. Assuming $\sigma_i \in \FP(G|_{\sigma_i})$, then the fact that $\sigma_i \notin \FP(G)$ implies there exists $j \in \sigma_i$ and $\ell \notin \sigma_i$ such that $\ell \geq_{\sigma_i} j$ (by Theorem~\ref{thm:domination}).  Note, however, that for all $k \in \sigma \setminus \sigma_i$,  we have $j >_{\sigma_i} k$ because (by the disjoint union) there are no edges from $\sigma_i$ to $k$. Hence $\ell \notin \sigma$, and it suffices by Theorem~\ref{thm:domination} to show that $\ell \geq_\sigma j$ in order to conclude that $\sigma \notin \FP(G)$. To do this, we will use Lemma~\ref{lemma:simply-added} for 
$\sigma = \tau \cup \omega$, where $\tau = \sigma_i$ and $\omega = \sigma \setminus \sigma_i$. Since there are no edges from $\omega$ to $\sigma_i$, the condition in part (iii) of Lemma~\ref{lemma:simply-added} is trivially satisfied, and thus $\ell \geq_{\sigma_i} j$ implies $\ell \geq_{\sigma} j$, as desired.
\end{proof}

\begin{proof}[Proof of Proposition~\ref{prop:clique-union-survival} (survival of clique union)]
Since $\sigma$ is a permitted motif, for all $j, \ell \in \sigma$ we have $j \sim_\sigma \ell$, by Theorem~\ref{thm:domination}.  By Theorem~\ref{thm:domination}, to show $\sigma \in \FP(G)$, we must show that for each 
$k\in [n]\setminus\sigma$, there exists a $j \in \sigma$ such that $j >_\sigma k$. 
Fix $k\in [n]\setminus\sigma$. By hypothesis, there exists a $\sigma_i$ such that $\sigma_i \in \FP(G|_{\sigma_i \cup k})$. Thus for any $j \in \sigma_i$, we have $j >_{\sigma_i} k$.  We will now use Lemma~\ref{lemma:simply-added} to show that this domination relationship lifts to $\sigma$. Following the notation of the lemma, let $\sigma = \tau \cup \omega$ for $\tau = \sigma_i$ and $\omega = \sigma \setminus \sigma_i$, and observe that because $\sigma$ is a clique union, $\omega$ is simply-added to $\tau$. Moreover, since $j \in \sigma_i = \tau$ receives edges from all nodes in $\omega$, the condition in part (iv) of Lemma~\ref{lemma:simply-added} is trivially satisfied. Thus, $j >_{\sigma_i} k$ implies $j >_{\sigma} k$, as desired. We conclude that $\sigma \in \FP(G|_{\sigma \cup k})$ for each $k \in [n] \setminus \sigma$, and thus $\sigma \in \FP(G)$. 
\end{proof}
\medskip

It turns out that even when $\tau$ is domination-free, it can still be useful to have a simply-added split in order to identify domination relationships in $\sigma$.  Example~\ref{ex:forbidden-simply-added} illustrates how one can use a simply-added split where $\tau$ has uniform in-degree to infer domination relationships in $\sigma$ from graphical domination that occurs in some ``equivalent" graph.  

\begin{figure}[!ht]
\begin{center}
\includegraphics[width=.7\textwidth]{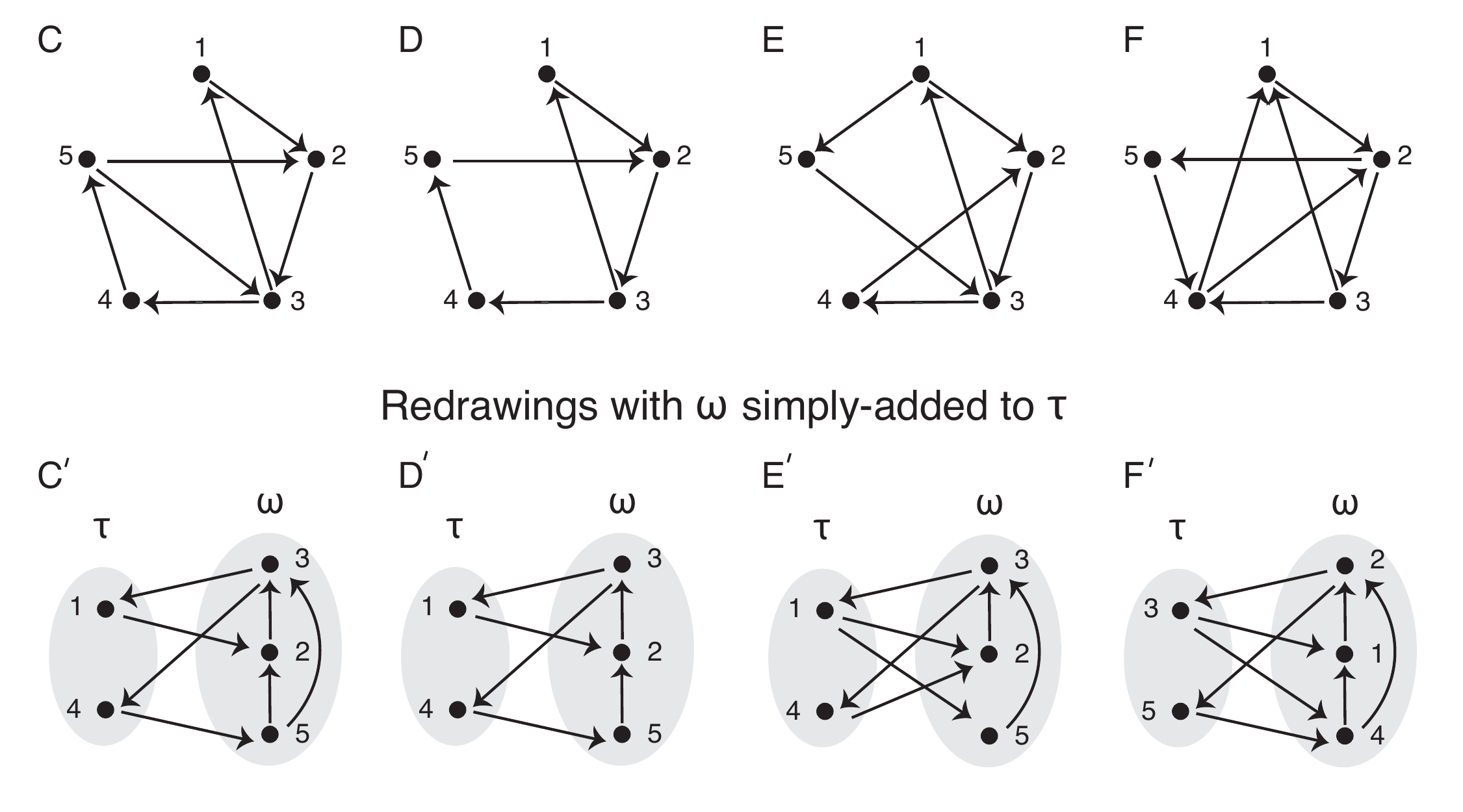}
\caption{\textbf{Inferring domination from a simply-added split.} (First row) Oriented graphs C--F from Figure~\ref{fig:domination-free-forbidden}, which are forbidden without graphical domination.  (Second row) Redrawings of the graphs in terms of components $\tau$ and $\omega$, where $\omega$ is simply-added to $\tau$.  
}
\label{fig:domination-free-forbidden-simply-added}
\end{center}
\vspace{-.1in}
\end{figure}

\begin{example}\label{ex:forbidden-simply-added}
Consider the graphs in Figure~\ref{fig:domination-free-forbidden-simply-added}.  In Appendix Section~\ref{appendixB}, these graphs are all shown to be forbidden motifs via parity arguments (these are graphs C--F in Figure~\ref{fig:domination-free-forbidden}).  Here we explore how/when simply-added splits can be used to infer domination relationships to directly show that a graph is forbidden.

Panel C$'$ in Figure~\ref{fig:domination-free-forbidden-simply-added} shows a redrawing of graph C in terms of a simply-added split, where $\tau= \{1,4\}$, $\omega = \{2,3,5\}$, and $\omega$ is simply-added to $\tau$.  Since $\tau$ is an independent set, it has uniform in-degree, and thus $s_1^\tau = s_4^\tau$; furthermore, since $\sigma = \{1,2,3,4,5\}$ is $\tau \cup \omega$, we have that $s_1^\sigma = s_4^\sigma$ by Theorem~\ref{thm:simply-added} (simply-added).  Thus any inputs from node 1 to the rest of the graph are equivalent to inputs from node 4 in terms of the contribution of the $s_i^\sigma$ values (the relevant quantity for computing $w_j^\sigma$ and identifying domination). Hence, in terms of domination relationships with respect to $\sigma$, graph C is equivalent to a graph where the $1 \to 2$ edge is replaced with a $4 \to 2$ edge.  In this equivalent graph, we see that node $2$ inside-in graphically dominates $5$.  Thus in graph C, we have $2 >_\sigma 5$, and so $\sigma$ (i.e.\ the full graph C) is a forbidden motif since it is not domination-free.  

Since graph D only differs from C by the $5 \to 3$ edge, an identical argument shows that graph D is equivalent (in terms of $s_i^\sigma$ values) to a graph where node $2$ graphically dominates $5$, and so in graph D, we have $2 >_\sigma 5$.  Thus, graph D is a forbidden motif.

Observe that graphs E and F can also be decomposed as $\omega$ simply-added to $\tau$, where $\tau$ has uniform in-degree, and thus the values of $s_i^\sigma$ for $i \in \tau$ are all equal (see Figure~\ref{fig:domination-free-forbidden-simply-added}E$'$ and F$'$).  However, there is no equivalent graph with inside-in graphical domination for either of these graphs.  To see this in graph E$'$, note that for each pair of nodes $j, k \in \omega$ with $j \to k$ and $k \not\to j$ (conditions 2 and 3 of the definition of graphical domination), we have that node $j$ receives more inputs from $\tau$ than node $k$ does; thus condition 1 can never be satisfied in any equivalent graph, and so $k \not>_\sigma j$ for any such pair.  The same argument can be made for graph F$'$ to show that graphical domination cannot be used to show the motif is forbidden.  Thus to show that these graphs are forbidden motifs, one must explicitly compute the $s_i^\sigma$, or rely on parity arguments as in Appendix Section~\ref{appendixB}.
\end{example}

\section{Discussion}

In this work, we have introduced two new characterizations for the fixed points of competitive TLNs: first in terms of sign conditions (Theorem~\ref{thm:sgn-condition}), and later in terms of domination (Theorem~\ref{thm:domination}). Specializing to CTLNs, we used these tools to prove key theorems on graphical domination (Theorem~\ref{thm:graph-domination}) and simply-added splits (Theorem~\ref{thm:simply-added}), as well as to derive survival rules for uniform in-degree motifs (Theorem~\ref{thm:uniform-in-degree}).  
These methods then enabled us to prove a series of graph rules in Section~\ref{sec:graph-rules}, which allow one to determine elements of $\FP(G)$ by direct analysis of the graph $G$. Finally, in Section~\ref{sec:building-blocks}, we have shown how this ``graphical calculus'' can be extended to larger networks comprised of simpler building blocks.

Any conclusions about a network derived from graph rules are automatically parameter independent. Since 
some CTLNs do have parameter-dependent permitted motifs (see Appendix Section~\ref{appendixC}), we know that 
graph rules cannot fully determine $\FP(G)$ in all cases. Nevertheless, it is likely that there are many more graph rules we have yet to discover, covering additional cases where motifs are permitted or forbidden in a parameter-independent manner. In particular, the style of graphical analysis illustrated in Example~\ref{ex:forbidden-simply-added} seems to hint at a missing graph rule. Moreover, if Conjecture~\ref{conjecture} is true then there are additional building block graph rules of the form given in Theorem~\ref{thm:composite-permitted} that apply to composite graphs with more complicated skeletons.

To what extent do graph rules extend to more general TLNs? One thing we can immediately say, based on the determinant form of the sign conditions in Theorem~\ref{thm:sgn-condition}, is that for any CTLN there must be an open neighborhood in the $(W,b)$ parameter space in which $\FP(W,b) = \FP(G,\varepsilon,\delta)$ -- that is, the fixed point supports of all TLNs in this neighborhood match those of the CTLN. This follows from the fact that the $s_i^\sigma$ are all polynomials in the entries of $W$ and $b$, and thus vary continuously as a function of these parameters. In particular, for purposes of determining fixed points, the constraints on $(W,b)$ imposed by CTLNs are not fine-tuned, and so the inferences we make from graph rules are robust to at least small perturbations of the network parameters. That said, we currently have no theoretical handle on how big these neighborhoods are, beyond the fact that they are open sets (and thus full-dimensional).

There are also non-perturbative approaches to generalizing graph rules.
In Appendix Section~\ref{appendixD}, we show how to associate a graph to any TLN $(W,\theta)$ having uniform input. For $n=1,2$, this graph fully characterizes the permitted motifs, just as it does for CTLNs. However, various CTLN graph rules break down at $n=3$, as illustrated in Example~\ref{ex:cautionary-example}. Currently, we are developing weaker versions of the graph rules that do extend to these more general settings, by applying the theory of oriented matroids to hyperplane arrangements associated to TLNs.

Finally, a comment on the nonlinearity. The proofs in this paper all rely on the precise form of the threshold-nonlinearity in~\eqref{eq:network2}, as well as on the assumption that $W$ is competitive. Based on our own computational observations, however, we expect that many qualitative aspects of these results should continue to hold for other nonlinearities, and/or less strict assumptions on $W$ (and $b$). But these questions are beyond the scope of the current paper, and so we leave them for future work.

\bigskip

\noindent{\bf Acknowledgments.} This work was supported by NIH R01 EB022862 and NSF DMS-1516881. The authors would also like to thank Anda Degeratu for proving an earlier version of Lemma~\ref{lemma:alternation}.

\bibliographystyle{unsrt}
\bibliography{CTLN-refs}

\begin{thebibliography}{10}

\bibitem{Seung-Nature}
R.~H. Hahnloser, R.~Sarpeshkar, M.A. Mahowald, R.J. Douglas, and H.S. Seung.
\newblock Digital selection and analogue amplification coexist in a
  cortex-inspired silicon circuit.
\newblock {\em Nature}, 405:947--951, 2000.

\bibitem{XieHahnSeung}
X.~Xie, R.~H. Hahnloser, and H.S. Seung.
\newblock Selectively grouping neurons in recurrent networks of lateral
  inhibition.
\newblock {\em Neural Comput.}, 14:2627--2646, 2002.

\bibitem{HahnSeungSlotine}
R.~H. Hahnloser, H.S. Seung, and J.J. Slotine.
\newblock Permitted and forbidden sets in symmetric threshold-linear networks.
\newblock {\em Neural Comput.}, 15(3):621--638, 2003.

\bibitem{flex-memory}
C.~Curto, A.~Degeratu, and V.~Itskov.
\newblock Flexible memory networks.
\newblock {\em Bull. Math. Biol.}, 74(3):590--614, 2012.

\bibitem{net-encoding}
C.~Curto, A.~Degeratu, and V.~Itskov.
\newblock Encoding binary neural codes in networks of threshold-linear neurons.
\newblock {\em Neural Comput.}, 25:2858--2903, 2013.

\bibitem{pattern-completion}
C.~Curto and K.~Morrison.
\newblock Pattern completion in symmetric threshold-linear networks.
\newblock {\em Neural Computation}, 28:2825--2852, 2016.

\bibitem{Hopfield1}
J.J. Hopfield.
\newblock Neural networks and physical systems with emergent collective
  computational abilities.
\newblock {\em Proc. Natl. Acad. Sci.}, 79(8):2554--2558, 1982.

\bibitem{CTLN-preprint}
K.~Morrison, A.~Degeratu, V.~Itskov, and C.~Curto.
\newblock Diversity of emergent dynamics in competitive threshold-linear
  networks: a preliminary report.
\newblock Available at \verb!https://arxiv.org/abs/1605.04463!

\bibitem{book-chapter}
K.~Morrison and C.~Curto.
\newblock {\em Predicting emergent dynamics from network connectivity}.
\newblock Book chapter in Algebraic and Combinatorial Computational Biology,
  edited by R. Robeva and M. Macaulay. Elsevier, 2018.

\bibitem{CTLN-paper}
K.~Morrison, J.~Geneson, C.~Langdon, A.~Degeratu, V.~Itskov, and C.~Curto.
\newblock Emergent dynamics from network connectivity: a minimal model.
\newblock \emph{In preparation.}

\end{thebibliography}
\pagebreak

\section{Appendix}

\subsection{Permitted and forbidden motifs of size $n \leq 3$}\label{appendixA}

\begin{figure}[!ht]
\begin{center}
\includegraphics[width=\textwidth]{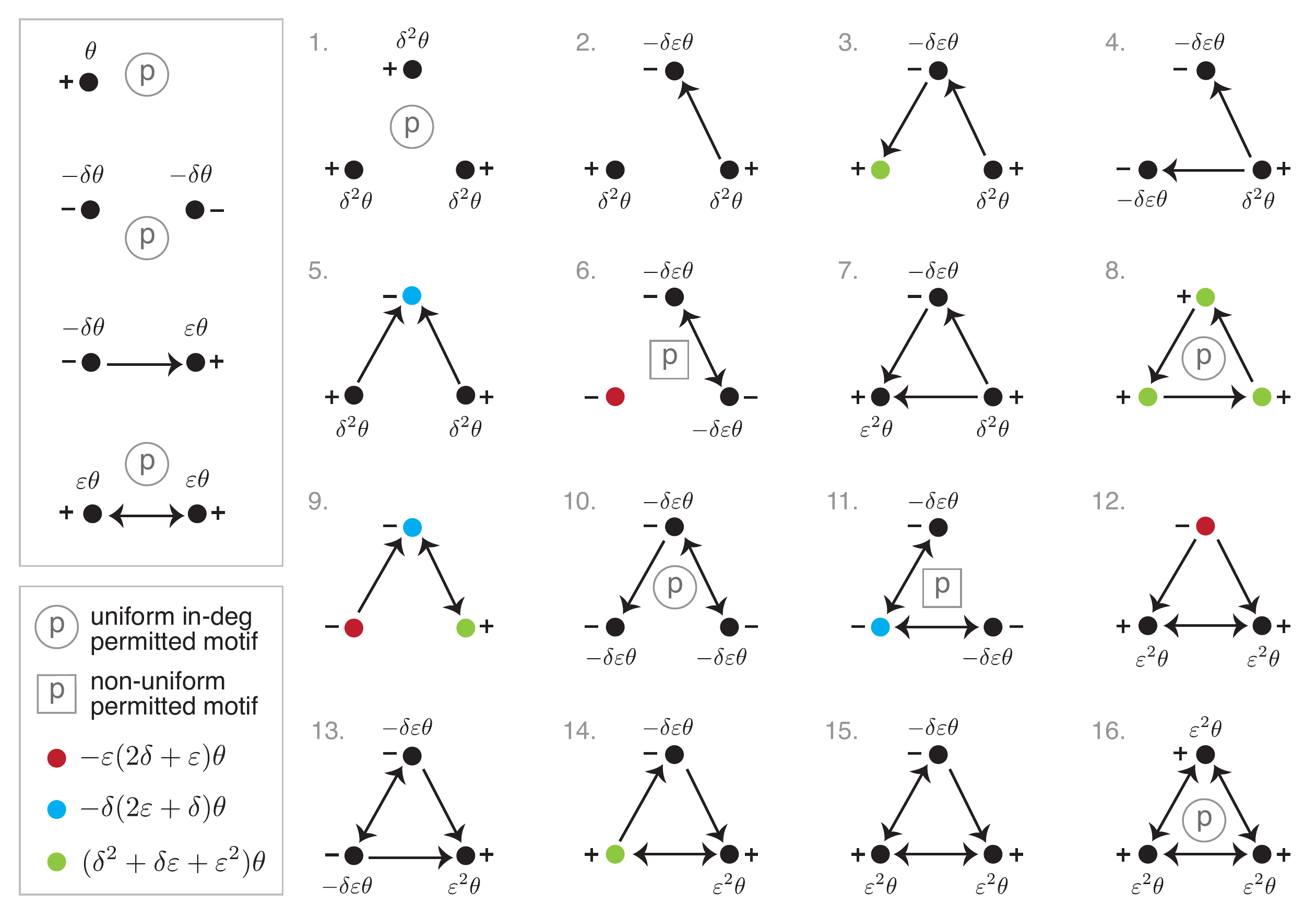}
\vspace{-.1in}
\end{center}
\caption{\textbf{All directed graphs on $n\leq 3$ nodes.} (Left) The 4 non-isomorphic directed graphs with $n\leq2$.  (Right) The 16 non-isomorphic directed graphs with $n=3$.   Each node is labeled with its value of $s_i^{\sigma}$, and its sign, for the full support $\sigma = \{1, \ldots, n\}$.  Note that even though the $s_i^\sigma$ values for the colored nodes are non-monomial (bottom left), their signs are also constant throughout the legal parameter range. Thus, using the sign conditions (Theorem~\ref{thm:sgn-condition}), we see that whether a graph is a permitted or forbidden motif is parameter independent for $n \leq 3$. Permitted motifs are labeled with a `p' that lies inside a circle (if the motif is uniform in-degree) or a square (if not). All other graphs are, by definition, forbidden.}
\label{fig:n3-graphs-s_i}
\end{figure}

\subsection{Parameter-dependent $\FP(G)$ for $n=5$}\label{appendixC}
In Section~\ref{sec:parameter-independence}, we saw that $\FP(G)$ is independent of parameters $\varepsilon$ and $\delta$ when $n \leq 4$; however, there are three permitted motifs of size 4 whose survival in a larger graph is parameter dependent.  These permitted motifs are reprinted here in Figure~\ref{fig:param-depend-min-perm}, panels A -- C.  The rest of the figure shows example graphs of size $n=5$ whose $\FP(G)= \FP(G, \varepsilon, \delta)$ depends on the choice of $\varepsilon$ and $\delta$. This parameter dependence is a result of these graphs containing one or more permitted motifs whose survival is parameter dependent.  Note that none of these $n=5$ graphs is oriented, and therefore, is not covered by Theorem~\ref{thm:oriented-param-independent}.  

\begin{figure}[!h]
\begin{center}
\includegraphics[width=\textwidth]{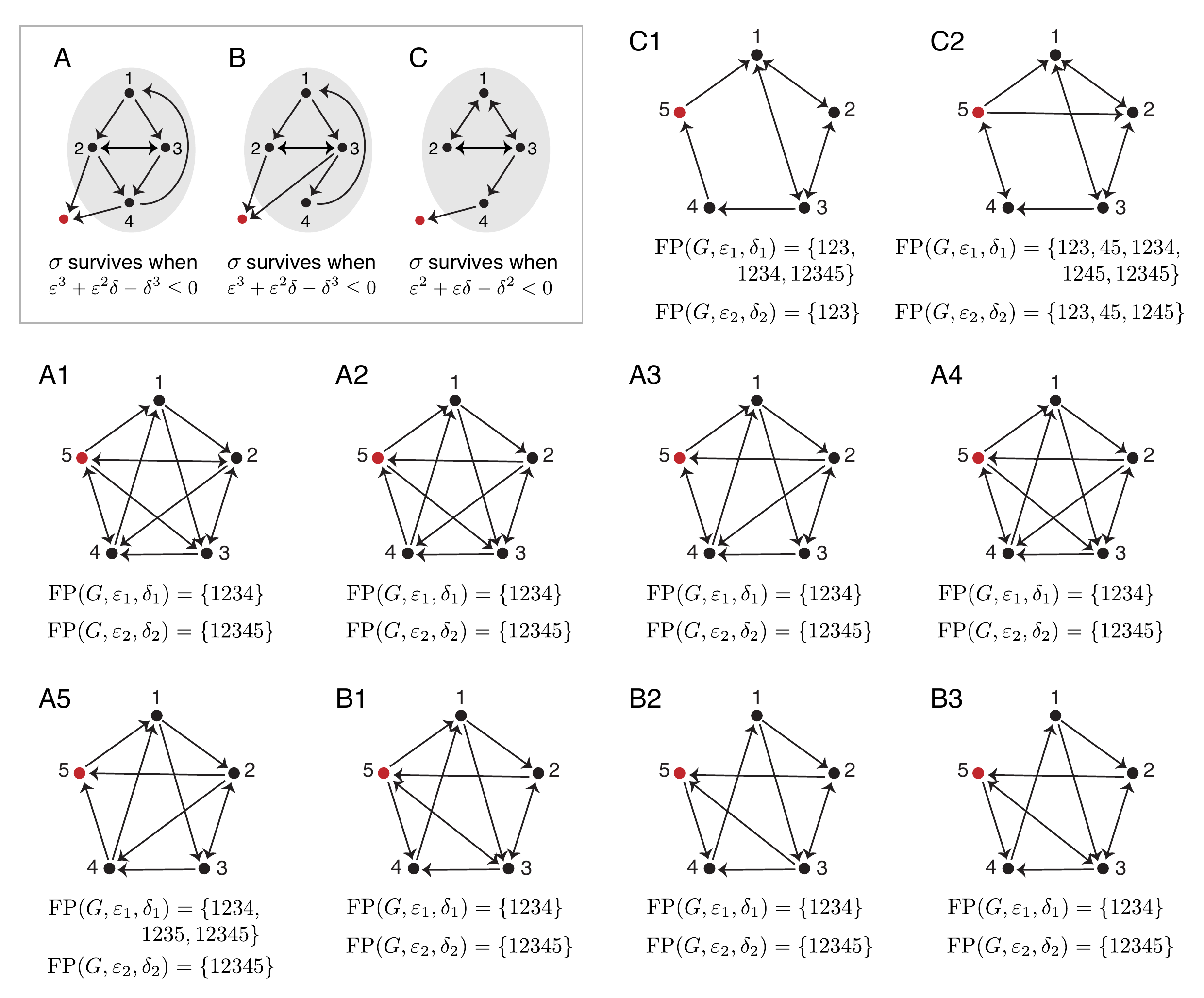}
\vspace{-.1in}
\end{center}
\caption{(A-C) The size 4 permitted motifs with parameter dependent survival for the embeddings shown. (A1-A5) Some example graphs on 5 nodes with the graph A embedding as a subgraph.  As a result $\FP(G, \varepsilon, \delta)$ is parameter dependent for these graphs.  Below each graph is $\FP(G, \varepsilon_1, \delta_1)$ for $(\varepsilon_1, \delta_1) = (0.25, 0.5)$ together with $\FP(G, \varepsilon_2, \delta_2)$ for $(\varepsilon_2, \delta_2) = (0.1, 0.12)$.  (B1-B3) Some example graphs with the graph B embedding as a subgraph.  (C1-C2) Some example graphs with the graph C embedding as a subgraph.  }
\label{fig:param-depend-min-perm}
\end{figure}

\FloatBarrier

\subsection{Using parity to detect forbidden motifs without graphical domination}\label{appendixB}

Although inside-in graphical domination tells us that a motif in a CTLN is forbidden, the absence of graphical domination is {\it not} sufficient to guarantee that a motif is permitted.  
Nevertheless, in many cases graph rules can still be used to determine whether such a motif is permitted or forbidden. In particular, Rule~\ref{rule:parity} (parity) can be used to determine if a motif is permitted/forbidden provided we know what happens for all proper subsets.

To illustrate this idea, we consider the family of \emph{oriented}\footnote{Recall a directed graph is \emph{oriented} if it has no bidirectional edges.} graphs with no sinks on $n~\leq~5$ vertices. Precisely six of these graphs are forbidden motifs despite the absence of graphical domination (see Figure~\ref{fig:domination-free-forbidden}). Fortunately, using graph rules we can identify
 all permitted motifs with $|\sigma| \leq 4$ that can arise as proper subgraphs in this family, together with their survival rules. Combining this with parity, we will show that the motifs in Figure~\ref{fig:domination-free-forbidden} are all forbidden.

\begin{figure}[!ht]
\includegraphics[width=\textwidth]{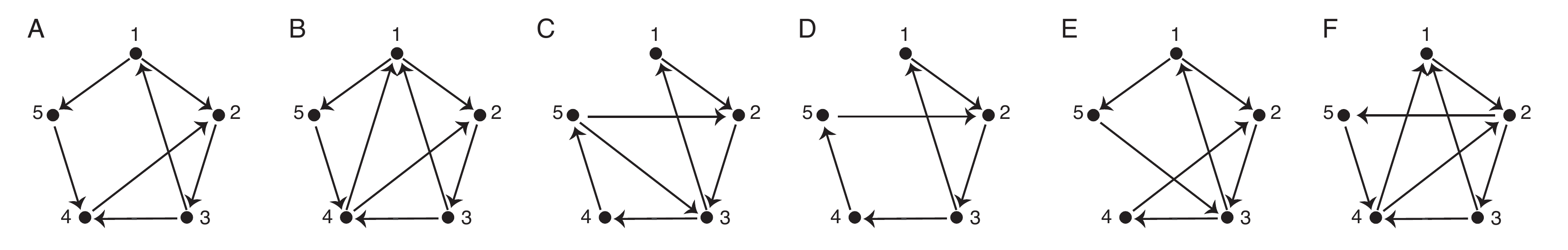}
\vspace{-.1in}
\caption{Oriented graphs with no sinks on $n=5$ that are forbidden motifs despite the absence of inside-in graphical domination. }
\label{fig:domination-free-forbidden}
\vspace{-.1in}
\end{figure}

\paragraph{Survival rules for oriented permitted motifs with $\boldsymbol{|\sigma|\leq 4}$.} Figure~\ref{fig:permitted-motifs} shows the five permitted motifs with $|\sigma| \leq 4$ that can potentially support fixed points within an oriented graph with no sinks.  Note that independent sets of arbitrary size are also permitted motifs, but by Rule~\ref{rule:indep-set} an independent set can only support a fixed point if it is a union of sinks, and so no independent set will ever survive in a graph with no sinks.  

\begin{figure}[!ht]
\vspace{-.1in}
\begin{center}
\includegraphics[width=.85\textwidth]{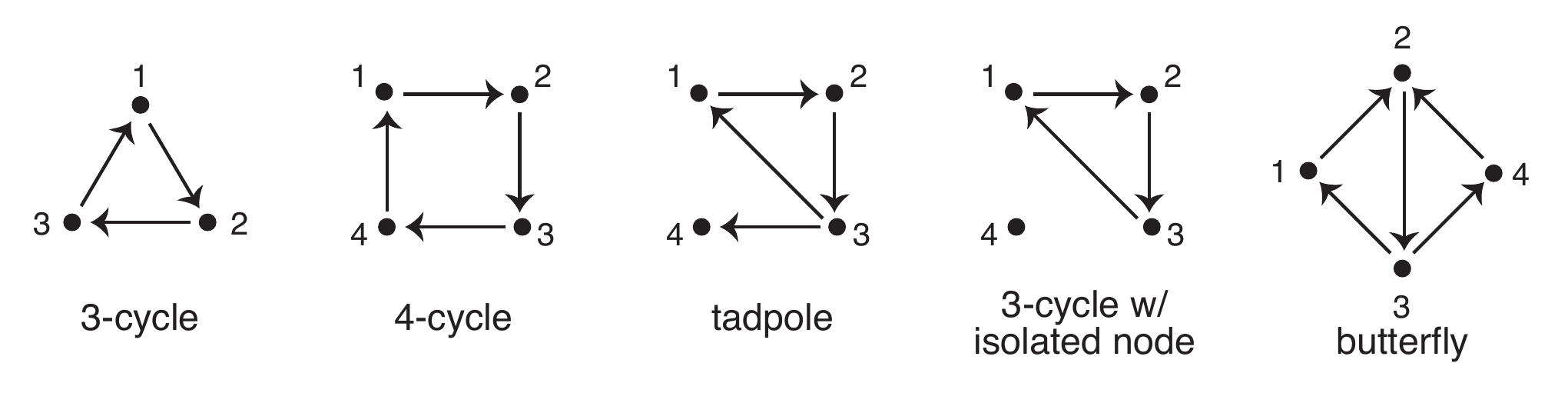}
\vspace{-.1in}
\caption{Permitted motifs of size $|\sigma| \leq 4$ (excluding independent sets) that can arise as subgraphs of an oriented graph with no sinks.
}
\label{fig:permitted-motifs}
\end{center}
\vspace{-.2in}
\end{figure}

The first three permitted motifs in Figure~\ref{fig:permitted-motifs} are all uniform in-degree 1.  Thus the survival rule for these motifs is given by Rule~\ref{rule:uniform-in-deg}: $\sigma \in \FP(G)$ if and only if no external node $k \notin \sigma$ receives 2 or more edges from the nodes in $\sigma$.  

The fourth motif is the disjoint union of a 3-cycle with an isolated node. Recall that a necessary condition for a disjoint union to survive is that each of its components survive (Proposition~\ref{prop:disjoint-union-survival}). Since the isolated node cannot survive in a graph with no sinks, we see that this motif will never survive in our family of $n=5$ graphs.

Finally, the butterfly graph is permitted by a parity argument: since the graph has no sinks and no bidirectional edges, there are no fixed point supports of size less than or equal to 2.  The only permitted motifs of size three are the two 3-cycles -- $(123)$ and $(234)$ -- which both survive by Rule~\ref{rule:uniform-in-deg}.  Thus, by parity, the full graph must be a fixed point support.
The survival rules for the butterfly graph are summarized in the following lemma.

\begin{lemma}\label{lemma:butterfly-survival}
Let $\sigma = \{1,2,3,4\}$, and let $G$ be any graph such that $G|_\sigma$ is the butterfly graph in Figure~\ref{fig:permitted-motifs}.  Then $\sigma \in \FP(G)$ if and only if every $k \notin \sigma$ either receives at most one edge from $\sigma$, or receives two edges from among nodes $1, 2$ and $4$.  
\end{lemma}

\begin{proof}
To derive the butterfly graph's survival rules, let $\sigma = \{1,2,3,4\}$ and consider all possible configurations of outgoing edges from the $\sigma$ to an added node 5.  It is easy to check that if node 5 receives less than two edges from $\sigma$, there will always be inside-out graphical domination, and so $\sigma$ will survive.  In contrast, if node 5 receives three or more edges from $\sigma$, then 5 will outside-in dominate some node in $\sigma$, and so $\sigma$ will not survive.  

When node 5 receives exactly two edges from $\sigma$, survival is more complicated and depends on which nodes 5 receives from. There are four cases.
(i) $\ 1, 3 \to 5$ (equivalently $3, 4 \to 5$ by symmetry):  $5$ outside-in dominates $1$, and so $\sigma$ does not survive.
(ii) $\ 2, 3 \to 5$:  $5$ outside-in dominates $3$, and so  $\sigma$ does not survive.  
(iii) $\ 1, 4 \to 5$: $2$ inside-out dominates $5$, and so $\sigma$ survives.
(iv) $\, 1, 2\to 5$ (equivalently $2, 4 \to 5$): In this case, there are no graphical domination relationships, and so we explicitly compute the values of $s_i^\sigma$ and check the sign conditions (Theorem~\ref{thm:sgn-condition}).  We find that 
$$s_1^\sigma = s_4^\sigma = -\delta\theta(\varepsilon^2 + \varepsilon\delta + \delta^2), \quad
s_2^\sigma = -\delta\theta(2\varepsilon^2+2\varepsilon\delta + \delta^2), \quad
s_3^\sigma = -\delta\theta(2\varepsilon^2 + 3\varepsilon\delta + 2\delta^2),$$
while $s_5^\sigma = \delta\theta(5\varepsilon^3 + 6\varepsilon^2\delta + 5\varepsilon\delta^2 + \delta^3)$.  Thus, $\sgn s_5^\sigma = - \sgn s_i^\sigma$ for all $i \in \sigma$, and so $\sigma$ survives.  
\end{proof}

\paragraph{Proof that the six graphs in Figure~\ref{fig:domination-free-forbidden} are forbidden.}

Recall that in an oriented graph with no sinks, no independent set can ever survive as a fixed point support, nor can any subgraph that is a 3-cycle with an isolated node. We can thus restrict attention to the remaining permitted motifs shown in Figure~\ref{fig:permitted-motifs}.  

Graph A  in Figure~\ref{fig:domination-free-forbidden} has 3-cycles 123 and 234, both surviving.
The tadpole 1235 does not survive, as node 4 receives two edges from it. There
are no restricted subgraphs that are 4-cycles, but we see that 1234 is the butterfly graph, and it survives by Lemma~\ref{lemma:butterfly-survival}.  Thus, $\FP(G) \supseteq \{123, 234, 1234\}$ and the only other possible fixed point support, $12345$, is ruled out by parity. Hence graph A is a forbidden motif.  

Graph B is identical to graph A except for the addition of the $4 \to 1$ edge.  As a result, the 3-cycle $234$ no longer survives.  Additionally, the subgraph on $1234$ is no longer the butterfly graph, and is not a permitted motif by inside-in domination (node 1 graphically dominates 4).  Thus, $\FP(G) \supseteq \{123\}$, and the only other candidate support, $12345$, is ruled out by parity.  Graph B is thus also a forbidden motif.

Similar parity arguments based on subgraph survival can be used for each of the remaining graphs in Figure~\ref{fig:domination-free-forbidden}.  We leave the remaining graphs as an exercise to the reader.

\subsection{Associating a graph to a general competitive TLN}\label{appendixD}
We have seen that the graph of a CTLN determines many properties of its collection of fixed point supports.  Thus it is natural to consider associating a graph to a general competitive TLN with constant external input $\theta>0$, and ask if properties of this graph also shape the fixed point supports. Recall that for such networks both Theorem~\ref{thm:sgn-condition} (sign conditions) and Theorem~\ref{thm:domination} (general domination) apply and either can be used to characterize the fixed point supports of the network.

We associate to each network $(W,\theta)$ a directed graph $G_{W}$ as follows:
$$G_{W} \text{ has an edge from } j \rightarrow i \;\;(\text{for } i \neq j)  \;\Leftrightarrow \; W_{ij}> -1.$$ 
For ease of notation, we will write $W_{ij}$ as $-1 + c_{ij}$ with $c_{kj}<1$ (to ensure $W$ is competitive).  In this case, $j \to i$ in $G_{W, \theta}$ if and only if $c_{ij} >0$.  Note that $G_{W}$ is a {\em simple} directed graph (no self loops).  This generalizes the graph used in CTLNs.  Indeed, if $W = W(G,\varepsilon,\delta)$ is a CTLN network with graph $G$, then $G_{W} = G$.

The graph $G_{W}$ can be used to describe certain aspects of the fixed point supports of the network~\eqref{eq:network2}.  First, we see that sinks of the graph correspond to WTA (stable) fixed points.  Recall that a vertex in a graph is a {\it sink} if it has out-degree 0.

\begin{lemma}\label{lemma:TLN-sinks}
A TLN $(W, \theta)$ has a fixed point supported on a single neuron $j$ if and only if $j$ is a sink of $G_{W}$.  Moreover, when $\{j\} \in \FP(W, \theta)$, the corresponding fixed point is stable.  
\end{lemma}
\begin{proof}
We prove this by computing the values of $s_j^{\{j\}}$ and $s_k^{\{j\}}$ for any $k \neq j$, then using Theorem~\ref{thm:sgn-condition} (sign conditions) to determine when $\{j\}$ supports a fixed point.  Note that $s_j^{\{j\}} = \det((I-W_{\{j\}})_j; \theta) = \theta$ and so $\{j\}$ is a permitted motif with positive index.  Furthermore, the eigenvalue of $I-W_{\{j\}}$ is $\theta$, which is positive, and so $\{j\}$ supports a stable fixed point in its restricted subnetwork.  This fixed point survives as a fixed point of $(W, \theta)$ precisely when $\sgn s_k^{\{j\}} =-\sgn s_j^{\{j\}} =-1$.  Computing $s_k^{\{j\}}$, we find
$$
s_k^{\{j\}} = \det((I-W_{\{j,k\}})_k; \theta)= \det(I-W_{\{j,k\}}) = \det {\begin{bmatrix} 
1 & \theta \\ 
1-c_{kj} & \theta
\end{bmatrix}} = \theta c_{kj}$$
where we have written $W_{kj}$ as $1-c_{kj}$.  Thus $\sgn s_k^{\{j\}} =-1$ if and only if $c_{kj}<0$, which by definition of $G_{W, \theta}$ occurs precisely when $j \not\to k$ in $G_{W}$.  By Theorem~\ref{thm:sgn-condition} (sign conditions), $\{j\} \in \FP(W, \theta)$ precisely when this condition is satisfied for all $k \neq j$, i.e.\ when $j$ has no outgoing edges, and so $j$ is a sink in $G_{W, b}$.

\end{proof}

We also find necessary conditions for a fixed point to be supported on exactly two neurons; specifically, we show which subsets of size two are permitted motifs.  

\begin{lemma}\label{lemma:TLN-size-2}
For a TLN $(W, \theta)$, a pair of neurons $\sigma =\{i,j\}$ is a permitted motif if and only if \\$\sigma$ is an independent set or a clique in $G_{W, \theta}$. 
Moreover, when $\sigma =\{i,j\} \in \FP(W,\theta)$, the corresponding fixed point is stable if and only if $\sigma$ is a clique.  
\end{lemma}
\begin{proof}
Let $W_{ij} = 1-c_{ij}$, $W_{ji} = 1-c_{ji}$, and $\sigma = \{i,j\}$.  With this notation, we see $s_i^{\sigma} = \det((I-W_{\{i,j\}})_i; \theta) = \theta c_{ij}$ and $s_j^{\sigma} = \theta c_{ji}$.  By Theorem~\ref{thm:sgn-condition} (sign conditions), $\sigma$ is a permitted motif if and only if $\sgn s_i^{\sigma} = \sgn s_j^{\sigma}$, i.e.\ $\sgn c_{ij} = \sgn c_{ji}$.  These signs agree precisely when (1) $c_{ij}, c_{ji} >0$ so that $i \leftrightarrow j$ and $\sigma$ is a clique in $G_W$, or (2) $c_{ij}, c_{ji} <0$ and $\sigma$ is an independent set.

Since $|\sigma|=2$ and $I-W_\sigma$ has positive trace, we see that the eigenvalues of $I-W_\sigma$ are positive (ensuring $\sigma$ supports a stable fixed point) precisely when $\det(I-W_\sigma) = c_{ij} + c_{ji} + c_{ij}c_{ji}>0$.  Since $W$ is competitive, $c_{ij}, c_{ji} <1$, and so $\det(I-W_\sigma)$ is positive precisely when $c_{ij}, c_{ji} >0$, so that $\sigma$ is a clique.
\end{proof}

Lemmas~\ref{lemma:TLN-sinks} and \ref{lemma:TLN-size-2} show that the fixed points of size $|\sigma|\leq 2$ of a general competitive TLN are completely determined by the graph, matching the case of CTLNs.  We might then hope that the CTLN survival rules of these fixed points would hold in the general case or that other graph rules would also go through.  Example~\ref{ex:cautionary-example} shows that unfortunately this is not the case.  
\bigskip

\begin{figure}[!h]
\begin{center}
\includegraphics[width=.65\textwidth]{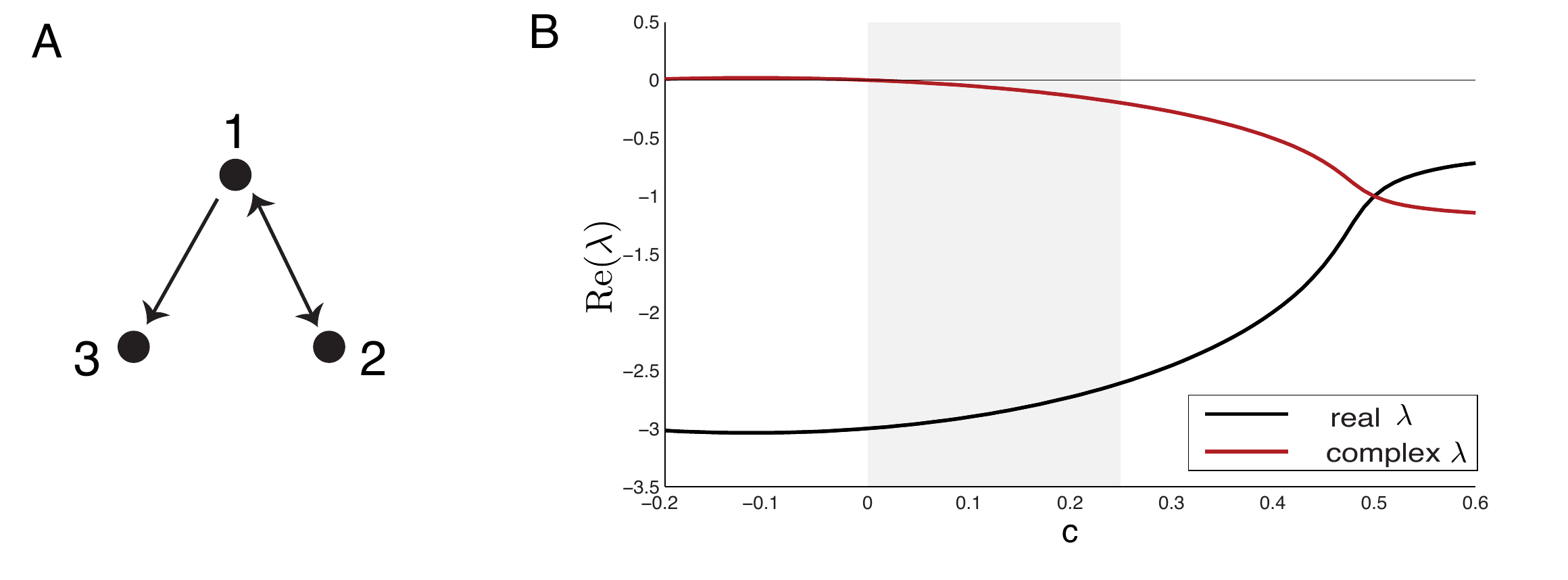}
\caption{(A) The graph $G_{W}$ corresponding to the network $W$ in Example~\ref{ex:cautionary-example}. (B) Plot of the real part of the eigenvalues of $-I+W$ as a function of $c$, for $W$ in Example~\ref{ex:cautionary-example}.  The matrix (and hence the corresponding full support fixed point) is stable throughout the range $0 < c < 0.25$.  At $c = 0$, the real and complex eigenvalues have real parts $-3$ and $0$, respectively.}
\label{fig:cautionary-example}
\end{center}
\vspace{-.2in}
\end{figure}

\begin{example}[cautionary example] \label{ex:cautionary-example}

Consider the competitive TLN $(W, \theta)$ with $\theta =1$ and
\[
W=\small{\left(\begin{array}{ccc} 
0 & -1+2c & -1-4c\\
-1+2c & 0 & -1-c\\
-1+4c & -1-c & 0
\end{array}\right)}
\]
for any $c$ with $0 < c < 0.25$, to ensure $W$ is competitive.  This network has corresponding graph $G=G_{W}$  with $1 \leftrightarrow 2$, $1 \to 3$, and no other edges (see Figure~\ref{fig:cautionary-example}A).   In this example, we compute $\FP(W, \theta)$ using Lemmas~\ref{lemma:TLN-sinks} and \ref{lemma:TLN-size-2} and Theorem~\ref{thm:sgn-condition} (sign conditions).

By Lemma~\ref{lemma:TLN-sinks}, the only singleton fixed point support is $\{3\}$ since it is the only sink in the graph.  By Lemma~\ref{lemma:TLN-size-2}, $\{1,2\}$ and $\{2,3\}$ are the only permitted motifs of size 2, but we must check if these survive as fixed points of $(W, \theta)$.  For $\sigma = \{1,2\}$, we have $s_1^\sigma = s_2^\sigma = 2 c$ and $s_3^\sigma = 2 c^2$; since $\sgn s_3^\sigma = \sgn s_i^\sigma$ for $i \in \sigma$, we see that $\{1,2\}$ is \underline{not} a fixed point support, despite being a target-free clique in $G_{W}$.  In contrast, for $\sigma = \{2,3\}$, we have $s_2^\sigma = s_3^\sigma = - c$ and $s_1^\sigma = c^2$, satisfying the sign conditions of Theorem~\ref{thm:sgn-condition}.  Thus, $\{2,3\}$ \underline{is} a fixed point support, despite being an independent set that is not a union of sinks of $G_{W, \theta}$.  This shows that both Rule~\ref{rule:indep-set} and Rule~\ref{rule:clique} for CTLNs do not extend to general TLNs.

Lastly, consider $\sigma = \{1,2,3\}$.  We have $s_1^\sigma = c^2$, $s_2^\sigma = 4c^2$, and $s_3^\sigma = 2c^2$.  Since the signs all agree, $\sigma \in \FP(W, \theta)$.  Since $\sigma$ has uniform in-degree $d=1$, in any CTLN, $\sigma$ would also be a fixed point support, but this fixed point would be unstable by Theorem~\ref{thm:uniform-in-degree} since $d < \frac{|\sigma|}{2}$.  However, for this general TLN $(W,\theta)$, we see that $\sigma$ supports a \underline{stable} fixed point for all the allowable values of $0 < c < 0.25$ (see Figure~\ref{fig:cautionary-example}B).

\end{example}

\end{document}